\documentclass[acmsmall]{acmart}
\pdfoutput=1
\usepackage[ready]{multippl}
\usepackage[nameinlink,capitalise]{cleveref}
\usepackage{preamble}
\usepackage[]{moeptikz}
\usepackage{multirow}
\usepackage{forest}
\usepackage{graphicx}
\usepackage{cancel}
\usepackage{wrapfig}
\usepackage{afterpage}
\usepackage{algorithm}
\usepackage{algpseudocode}
\usepackage{subcaption}
\usepackage{rotating} %
\usepackage{float} %
\usepackage{xcolor}
\definecolor{y}{HTML}{fff7bc}
\newcommand{\hilite}[1]{\colorbox{y}{\(#1\)}}

\makeatletter
\newcommand*\bigcdot{\mathpalette\bigcdot@{.75}}
\newcommand*\bigcdot@[2]{\mathbin{\vcenter{\hbox{\scalebox{#2}{$\m@th#1\bullet$}}}}}
\makeatother

\newcommand\withappendix{}
\ifdefined\withappendix
\newcommand\appcref[1]{\cref{#1}}
\else
\newcommand\appcref[1]{the appendix}
\fi

\setcopyright{none} %

\usetikzlibrary{calc}
\usetikzlibrary{shapes}
\usetikzlibrary{arrows.meta}
\usetikzlibrary{positioning}
\tikzset{%
  zeroarrow/.style = {-stealth,densely dashed},
  onearrow/.style = {-stealth,solid},
  c/.style = {circle,draw,solid,minimum width=1em, minimum height=1em},
  r/.style = {rectangle,draw,solid,minimum width=1em, minimum height=1em}
  0 my edge/.style={densely dashed, my edge, gray,font=\ttfamily},
  my edge/.style={-{Stealth[]}},
}

\newcommand{\multippl}[0]{\textsc{MultiPPL}}
\newcommand{\dice}[0]{\textsc{Dice}}
\newcommand{\stan}[0]{\textsc{Stan}}

\renewcommand\x{x} %

\tikzset{>=latex}
\usepackage{subcaption,enumitem}
\usepackage{standalone}
\usepackage{nicefrac} %
\usepackage{wrapfig}
\lstdefinestyle{basic}{
 columns=[c]fixed,
 basicstyle=\footnotesize\ttfamily,
 keywordstyle=\bfseries,
 upquote=true,
 commentstyle=,
 breaklines=true,
 showstringspaces=false,
 breakatwhitespace=false,
 captionpos=b,
 keepspaces=true,
 numbers=left,
 numbersep=5pt,
 showspaces=false,
 showtabs=false,
 tabsize=2,
 mathescape=true,
 keywords={exact,flip,observe},
 keywordstyle=\color{indigo50},
 morekeywords={exact,flip,return,observe},
}

\lstdefinelanguage{custom}{
  columns=[c]fixed,
}

\forestset{
  BDD/.style={
    for tree={
      if n children=0{rectangle}{if n=0{circle, minimum size=0.9cm}{circle, scale=0.90}},%
      draw,%
      edge={
        my edge,%
      },
      if n=0{
        edge+={0 my edge},%
      }{},
      font=\sffamily,%
    }
  },
}

\title{Multi-Language Probabilistic Programming}

\begin{abstract}
There are many different probabilistic programming languages that are
specialized to specific kinds of probabilistic programs. From a usability and
scalability perspective, this is undesirable: today, probabilistic programmers
are forced up-front to decide which language they want to use and cannot
mix-and-match different languages for handling heterogeneous programs.
  To rectify this, we seek a foundation for sound interoperability for
probabilistic programming languages: just as today's Python programmers can
resort to low-level  C programming for performance, we argue that probabilistic
programmers should be able to freely mix different languages for meeting the
demands of heterogeneous probabilistic programming  environments. As a first
step towards this goal, we introduce \textsc{MultiPPL}, a probabilistic
multi-language that enables programmers to interoperate between two different
probabilistic programming languages: one that leverages a  high-performance
exact discrete inference strategy, and one that uses approximate importance
sampling.  We give a syntax and semantics for \textsc{MultiPPL}, prove soundness
of its inference algorithm, and provide empirical evidence that it enables
programmers to perform inference on complex heterogeneous probabilistic
programs and flexibly exploits the strengths and weaknesses of two languages
simultaneously.%
\end{abstract}

\begin{document}

\maketitle

\section{Introduction}
Scalable and reliable probabilistic inference remains a significant barrier for
applying and using probabilistic programming languages (PPLs) in practice. The
core of the inference challenge is that there is no universal approach:
different kinds of inference strategies are specialized for different
kinds of probabilistic programs. For example, \textsc{Stan}'s inference strategy
is highly effective on continuous and differentiable programs such as
hierarchical Bayesian models, but struggles on programs with high-dimensional
discrete structure such as graph and network reachability~\citep{carpenter2017stan}. On the other extreme, languages
like \textsc{Dice} and \textsc{ProbLog} scale well on purely-discrete problems,
but the price for this scalability is that they must forego any support
whatsoever of continuous probability
distributions~\citep{holtzen2020scaling,fierens2015inference}.
In an ideal world, a probabilistic programmer would not have to commit to one language
or the other: they could use a \dice{}-like language for high-performance scalable inference on the
discrete portion of the program, a \stan{}-like language for the portion to which
it is well-suited, and be able to transfer data and control-flow between these
two languages for heterogeneous programs.

This raises a key question: how should we orchestrate the handoff between two
probabilistic programming languages whose underlying semantics may be radically
different and seemingly incompatible? This question of sound language
interoperability has been extensively explored in the context of non-probabilistic
languages~\citep{matthews2007operational,patterson2022semantic,new2016fully,tobin2011languages,dagand2018foundations,scherer2018fab,tov2010stateful},
where the goal is to prove properties such as type-soundness and termination
in a multi-language setting.
As a starting point, \citet{matthews2007operational} introduced an effective model for capturing the
interaction between two languages by \emph{language embedding}:
the syntax and operational semantics of a Scheme-like language and an
ML-like language are unified into
a multi-language and syntactic
boundary terms are added to mediate transfer of control and data between the two languages.
Using this embedding approach, they were able to establish the type-soundness of the
multi-language. Their approach
relied on a careful enforcement of soundness properties on the boundaries, for
instance inserting dynamic guards or contracts to ensure soundness when
typed ML values are transferred to untyped Scheme.

We introduce the notion of \emph{sound inference interoperability}:
whereas sound interoperability of traditional languages
ensures that multi-language programs are type-sound,
sound \emph{inference} interoperability ensures
that probabilistic multi-language programs
correctly represent the intended probability distribution.
Our main goal will be to establish sound inference interoperability of two
PPLs: the first is called \disc{} and is similar to \dice{}, and the second is
called \cont{} and it makes use of importance-sampling-based inference. These
two languages are a nice pairing: \disc{} provides scalable exact discrete inference
at the expense of expressivity, as it does not support continuous random
variables and unbounded loops. On the other hand, \cont{} provides significant
expressivity (it supports continuous random variables) at the cost of relying on importance-sampling-based
approximate inference.
Following \citet{matthews2007operational}, we embed both \disc{}
and \cont{} into a multi-language we call \multippl{}.
Together, these two languages cover a broad spectrum of interesting
probabilistic programs that are difficult to handle today.
We will show in \cref{sec:motivation} and
\cref{sec:eval} that examples from networking and probabilistic
graphical models benefit from the ability to flexibly use different languages
and inference algorithms within a unified multi-language.
Traditional multi-language semantics establishes sound interoperability
by proving type-soundness of the combined multi-language~\citep{matthews2007operational}.
Analogously, we establish sound inference interoperability between \disc{} and \cont{},
guaranteeing that well-typed \multippl{} programs correctly represent
the intended probability distribution. Our contributions are as follows:
\begin{itemize}[leftmargin=*]
	\item We introduce \multippl{},
	      a multi-language in the style of \citet{matthews2007operational}
	      that enables interoperation between a discrete exact probabilistic
	      programming language \disc{} and a continuous approximate probabilistic programming language \cont{}.
	\item In \cref{sec:multippl}
	      we construct two models of \multippl{}
	      by combining appropriate semantic domains for \disc{} and \cont{} programs:
	      a high-level model capturing the probability distribution
	      intended by a given \multippl{} program,
	      and a low-level model
	      capturing the details of our particular implementation strategy.
	      We then prove that these two semantics agree, establishing
	      correctness of the implementation~(\cref{thm:toplevel-soundness}).
	      We identify two key requirements for ensuring sound inference interoperability between
	      exact and approximate programs: \disc{} programs must additionally enforce
	      \emph{sample consistency} for ensuring \disc{} values pass safely into \cont{},
	      and \cont{} programs must additionally perform
	      \emph{importance weighting} for ensuring the safety of \disc{} conditioning.

	\item In \cref{sec:eval} we validate the practical effectiveness of \multippl{} through our
	      provided implementation. We evaluate the efficacy of \multippl{} by
	      modeling complex independence structures through real-world problems in the
	      domain of networking and probabilistic graphical models. We provide insights
	      into \multippl{}'s approach to probabilistic inference and characterize the
	      nuanced landscape that interoperation produces.

\end{itemize}

\section{Overview}\label{sec:motivation}

\begin{figure}[t]
	\begin{subfigure}[b]{0.4\linewidth}
		\centering
		\scalebox{0.8}{
			\begin{tikzpicture}[x=1cm,y=1cm]
				\node[] (start) at (-2,0) {};
				\node[nuc] (r1) at (0,0) {R1};
				\node[nuc] (r2) at (2,1) {R2};
				\node[nuc] (r3) at (2,-1) {R3};
				\node[nuc] (r4) at (4,0) {R4};
				\draw[->] (start) -- (r1);
				\draw[->] (r1) -- (r2);
				\draw[->] (r1) -- (r3);
				\draw[->] (r3) -- (r4);
				\draw[->] (r2) -- (r4);
			\end{tikzpicture}
		}
		\vspace{0.2cm}
	\end{subfigure}
	\hfill
	\begin{subfigure}[b]{0.5\linewidth}
		\begin{lstlisting}[mathescape=true]
function step() {
  let lambda = uniform(0, 5) in
  let numPackets = poisson(lambda) in
  for i in 0..numPackets {
    forwardPacket(R1)
  }
}
\end{lstlisting}
	\end{subfigure}
	\caption{A small network and a fragment of a probabilistic program encoding of
		the packet arrival problem.}
	\label{fig:motiv}
\end{figure}
\begin{figure}
	\begin{tabular}{rlrcl}
		\disc & Expressions     & $\eM,\eN$       & ::=   & $\eco X$ | $\etrue$ | $\efalse$ | $\eM\eand\eN$ | $\enot\eM$                                                          \\
		      &                 &                 & |     & $\eunitelt$ | $\epair\eM\eN$ | $\efst~\eM$ | $\esnd~\eM$                                                              \\
		      &                 &                 & |     & $\eret~\eM$ | $\eletin {\eco X}\eM\eN$ | $\eite\scoe\eM\eN$                                                           \\
		      &                 &                 & |     & $\eflip~\scoe$ | $\eobserve~\eM$ | $\esample[\scoe]$                                                                  \\
		      & Types           & $\eA,\eB$       & ::=   & $\eunit$ | $\ebool$ | $\eA\etimes\eB$                                                                                 \\
		      & Contexts        & $\eDelta$       & ::=   & $\eco{X_{1}} : \eco{A_{1}},\dots,\eco{X_n} : \eco{A_n}$                                                               \\
		\\
		\cont & Expressions     & $\scoe$         & ::=   & $\sco x$ | $\strue$ | $\sfalse$ | $r$ | $\seone\sadd\setwo$ | $\sneg\scoe$ | $\seone\smul\setwo$ | $\seone\sle\setwo$ \\
		      &                 &                 & |     & $\sunitelt$ | $\spair\seone\setwo$ | $\sfst~\scoe$ | $\ssnd~\scoe$                                                    \\
		      &                 &                 & |     & $\sret~\scoe$ | $\sletin {\sco x}{\seone}{\setwo}$ | $\site{\seone}{\setwo}{\sethr}$                                  \\
		      &                 &                 & |     & $\sco{d}$ | $\sco{\texttt{obs}(\scoe_{o}, d)}$  | $\sexact[\eco M]$                                                   \\
		      & Distributions   & $\sco{d}$       & ::=   & $\sflip~\scoe$ | $\sunif~\seone~\setwo$ | $\spois~\scoe$                                                              \\
		      & Types           & $\ssigma,\stau$ & ::=   & $\sunit$ | $\sbool$ | $\sreal$ | $\ssigma\stimes\stau$                                                                \\
		      & Contexts        & $\sGamma$       & ::=   & $\sco{x_1} : \sco{\tau_1},\dots,\sco{x_n} : \sco{\tau_n}$                                                             \\
		      & Number literals & $r$             & $\in$ & $\R$                                                                                                                  \\
	\end{tabular}

	\caption{Syntax of \multippl{}. In \disc{}, we require ${\sco e} \in [0,1]$ for
		$\eflip$. In \cont{}, the syntax of distributions $\sco d$ denotes probability distributions. In $\sobs$ these terms condition, otherwise they are immediately sampled.}
	\label{fig:syntax}
\end{figure}

We argue that it is often the case that realistic probabilistic programs consist of sub-programs
that are best handled by \emph{different} probabilistic programming languages. Consider for example the
\emph{packet arrival}
situation visualized in \cref{fig:motiv}.
In this example, at each time step,
network packets are arriving according to a Poisson distribution, a fairly
standard setup in discrete-time queueing theory~\citep{meisling1958discrete}.
Then, each packet is forwarded through the network, whose topology is visualized
as a directed graph. The goal is to query for various properties about the
network's behavior: for instance, the probability of a packet reaching the end of the network,
or of a packet queue overflowing. This example task is inspired by prior
work on using probabilistic programming languages to perform network
verification~\citep{gehr2018BayonetProbabilistic,smolka2019scalable}.

The situation in \cref{fig:motiv} is a small illustrative example of packet arrival,
but programs like it are extremely challenging for today's PPLs because they
mix different kinds of program structure. Lines 2 and 3 manipulate continuous and
countably-infinite-domain random variables, which precludes the use of \dice{}.
However, graph reachability and queue behavior are complex discrete distributions, which
are difficult for \stan{} due to their inherent non-differentiability and
high-dimensional discrete structure. In order to scale on this example, we would like to be able to use
an inference algorithm like \stan{}'s for lines 2 and 3, and an inference algorithm like
\dice{}'s for lines 4--6.

Our approach to designing a language capable of handling
situations like that described in \cref{fig:motiv} is to enable the
programmer to seamlessly transition between programming in two different PPLs: \cont{}, an
expressive language that supports sampling-based inference and continuous random
variables, and \disc{}, a restricted discrete-only language that supports
scalable exact discrete inference. Following \citet{matthews2007operational}, we
describe a \emph{probabilistic multi-language} that embeds both
languages into a single unified syntax: see \cref{fig:syntax}. In \cref{sec:syntax}
we discuss the intricacies of the syntax in \cref{fig:syntax} in full detail,
including typing judgments found in \cref{fig:high-level-types} and the appendix; here we briefly
note its high-level structure and discuss examples.

These languages delineate our two syntactic categories:
\begin{enumerate}[noitemsep,leftmargin=*]
	\item \disc{} terms, shown in \LE{purple}, that support discrete probabilistic
	      operations such as Bernoulli random variables and Bayesian conditioning.  The
	      syntax is standard for an ML-like functional language with the addition of
	      probabilistic constructs: $\eflip~\scoe$ introduces a Bernoulli random
	      variable that is $\etrue$ with probability $\scoe \in [0,1]$ and $\efalse$
	      otherwise; the construct $\eobserve~\eM$ conditions on $\eM$. Notably, \disc{}
	      lacks introduction forms for continuous random variables or real numbers and
	      so in order to define the Bernoulli-distributed random variable using
	      $\eflip$, we must rely on interoperation to construct our distribution.

	\item \cont{} terms, shown in \LS{orange}, additionally support standard continuous
	      operations and sampling capabilities from two distributions inexpressible in
	      \disc{}: a Uniform distribution $\sunif~\seone~\setwo$ over the interval
	      $[\seone,\setwo]$, with $\seone,\setwo \in \R$, and a Poisson distribution
	      $\spois$ with rate $\scoe \in \R$ being greater than zero. The
	      syntax $\sco{\texttt{obs}(\scoe_{\sco o}, d)}$ %

	      denotes conditioning on the event that a sample drawn from distribution $\sco d$
	      is equal to $\scoe_{\sco o}$. %

\end{enumerate}

\begin{wrapfigure}{r}{0.40\linewidth}
	\vspace{-3em}
	\begin{lstlisting}[mathescape=true,caption={\textsc{TwoCoins}},label={lst:twocoins},escapechar=|]
$\sletin{\sco \theta}{\sunif~0~1}{}$                         |\label{line:twocoins:unif}|
${\sco\openbanana} \eletin{\eco X}{\eflip~\sco\theta}{}$  |\label{line:twocoins:x}|
${\phantom\openbanana} \eletin{\eco Y}{\eflip~\sco\theta}{}$ |\label{line:twocoins:y}|
${\phantom\openbanana} \eobserve~\eco{X \lor Y}~\einkw$ |\label{line:twocoins:obs}|
${\phantom\openbanana} \eret~{\eco X} ~ \sco{\closebanana_S}$ |\label{line:twocoins:ret}|
\end{lstlisting}
\end{wrapfigure}

Mediating between the \disc{} and \cont{} sublanguages are the boundaries
$\esample[\scoe]$ and $\sexact[\eM]$: the boundary $\esample[\scoe]$ allows
passing from \cont{} to \disc{}, and the boundary $\sexact[\eM]$ allows passing
from \disc{} to \cont{}. This style of presentation is similar to \citet{patterson2022InteroperabilityRealizability}.

\cref{lst:twocoins} shows an example program in our multi-language which
passes a uniformly-sampled real value $\sco \theta$ from \cont{} into \disc{}
and uses the resulting value as a prior for sampling two independent Bernoulli
random variables.
The outer-most language is \cont{}. On \cref{line:twocoins:unif}, $\sco{\theta}$
is bound to a sample drawn from the uniform distribution on the unit interval.
Then, on Lines \ref{line:twocoins:x}--\ref{line:twocoins:ret}, we begin evaluation of
a \disc{} program inside the boundary term \sexact{}. We flip two coins
$\eco{X}$ and $\eco{Y}$ (\cref{line:twocoins:x,line:twocoins:y}, respectively)
in the \disc{} sub-language, whose prior parameters are both $\sco\theta$.
On \cref{line:twocoins:obs}, we observe that one of the two coins was true,
taking advantage of syntactic sugar where $\eobserve$ is bound to a discarded
variable name.
\cref{line:twocoins:ret} brings us to the final line of our program, where we
query for the probability that $\eco X$ is true.
The next
two sub-sections will explain our approach to bridging
the two languages.

\subsection{\disc{} and \cont{} inference}
Before we describe the intricacies of language interoperation, we first provide
some high-level intuition for how we wish to perform inference on \disc{} and
\cont{} independently.
First, we give a denotational semantics for \multippl{} that we denote $\dbracket{-}$
which associates each \multippl{} term with a probability distribution on
\multippl{} values (see \cref{sec:multippl} for a formal definition of these
semantics). Here we will briefly illustrate these semantics by example: the semantics
$\dbracket{\eflip~\sco p}$ produces a Bernoulli distribution that is true with probability $\sco p \in [0,1]$;
the semantics $\dbracket{\sunif~\seone~\setwo}$ produces a uniform distribution on the interval $[\seone, \setwo] \in \R$.

The goal of inference is to efficiently evaluate the denotation of a probabilistic
program. While \disc{} and \cont{} share a unified denotation, they have
very different approaches to inference.
The key advantage of our multi-language approach is that
we can specialize the design of \cont{} and \disc{} to take full advantage of
structural differences between their underlying inference algorithms: for
\disc{} we will use an exact inference strategy based on knowledge compilation
similar to \dice{}~\citep{holtzen2020scaling}, and for \cont{} we will rely on
approximate inference via sampling.
In the next two subsections we give a high-level overview of these standard
approaches.

\newcommand*\bddpredicate[1]{\framebox{{\footnotesize \textsf{#1}}}}

\subsubsection{Exact inference via knowledge compilation}
\label{sec:discinf}
\begin{figure}
	\begin{subfigure}[b]{0.4 \linewidth}
		\centering
		\begin{lstlisting}[mathescape=true,escapechar=|]
$\eletin{\eco X}{\eflip~0.4}{}$  |\label{line:etwocoins:x}|
$\eletin{\eco Y}{\eflip~0.3}{}$ |\label{line:etwocoins:y}|
$\eobserve~\eco{X \lor Y}~\einkw$ |\label{line:etwocoins:obs}|
$\eret~{\eco X}$ |\label{line:etwocoins:ret}|
\end{lstlisting}
		\caption{Example \disc{} program.}\label{fig:etwocoins:program}
	\end{subfigure}
	\quad
	\begin{subfigure}[b]{0.45\linewidth}
		\centering
		\footnotesize

		\begin{tikzpicture}[node distance=0.5cm and 0.025cm]
			\node[c] (x) {$f_X$};
			\node[draw,rectangle] (x-one) [below left=of x] {$\BDDTrue$};
			\node[draw,rectangle] (x-zero) [below right=of x] {$\BDDFalse$};

			\draw[onearrow]  (x) -- (x-one)  node [midway,left,xshift=0pt,yshift=2pt] {0.4};
			\draw[zeroarrow] (x) -- (x-zero) node [midway,right,xshift=0pt,yshift=2pt] {0.6};
		\end{tikzpicture}
		\qquad
		\begin{tikzpicture}[node distance=0.5cm and 0.025cm]\footnotesize
			\node[c] (x) {$f_X$};
			\node[c] (y) [below right=of x] {$f_Y$};
			\node[draw,rectangle] (x-one) [below left=of x,xshift=-2pt] {$\BDDTrue$};
			\node[draw,rectangle] (y-one) [below left=of y,xshift=-2pt] {$\BDDTrue$};
			\node[draw,rectangle] (y-zero) [below right=of y] {$\BDDFalse$};

			\draw[zeroarrow] (x) --  (y) node [midway,right,xshift=0pt,yshift=2pt] {0.6};
			\draw[onearrow] (x) -- (x-one) node [midway,left,xshift=0pt,yshift=2pt] {0.4};
			\draw[onearrow] (y) -- (y-one) node [midway,left,xshift=0pt,yshift=2pt] {0.3};
			\draw[zeroarrow] (y) -- (y-zero) node [midway,right,xshift=0pt,yshift=2pt] {0.7};
		\end{tikzpicture}
		\caption{BDD representations of formulas.}\label{fig:etwocoins:bdds}
	\end{subfigure}
	\caption{Motivating example showing the compilation of the \disc{} program in \ref{fig:etwocoins:program} to BDDs in \ref{fig:etwocoins:bdds}.
		On the left of \ref{fig:etwocoins:bdds} is a BDD representing the distribution formula $\varphi = f_X$; on the right
		is the BDD representing the accepting formula $\alpha = f_X \lor f_Y$. \bddpredicate{T} and \bddpredicate{F} represent true and false values, respectively.
	}\label{fig:etwocoins}
\end{figure}
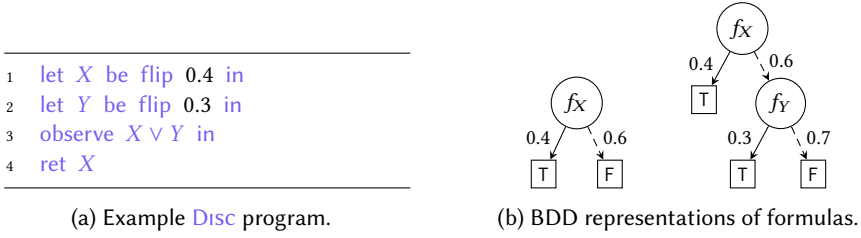

Here, we illustrate the principles of exact inference in \disc{} via example;
\cref{sec:multippl} provides a formal treatment of these semantics.
In \cref{fig:etwocoins:program}, we reproduce the \disc{} program compiled in
Lines \ref{line:twocoins:y}--\ref{line:twocoins:ret} of \cref{lst:twocoins}, but
instantiate the priors of $\eco Y$ and $\eco Z$ with numeric literals 0.4 and
0.3, respectively.
Our example in \cref{fig:etwocoins:program} denotes the
probability distribution of Bernoulli 0.3, given that one of the two weighted
coin flips is true; its semantics is $\dbracket{\cref{fig:etwocoins:program}}(\true) = \frac{0.4}{0.58} \approx 0.689$.

The exact inference strategy used by \disc{} is to perform probabilistic
inference via weighted model
counting~\citep{darwiche2002knowledge,sang2005performing,chavira2008probabilistic},
following \citet{holtzen2020scaling}.
The key idea is to interpret the probabilistic program
as a weighted Boolean formula whose models are in one-to-one
correspondence with paths through the program, and where each path is associated
with a weight that matches the probability of that path in the program.
Concretely,
a \emph{weighted Boolean formula} is a
pair $(\varphi,w)$ where $\varphi$ is a Boolean formula and $w$ is a
\emph{weight function} that associates literals (assignments to variables) in
$\varphi$ with real-valued weights. Then, the \emph{weighted model count}
of a weighted Boolean formula is the weighted sum of models:
\begin{align}
	\WMC(\varphi, w) = \sum_{\{m \vDash \varphi\}} \prod_{\{\ell \in m\}} w(\ell).
\end{align}
To perform \disc{} inference by reduction to weighted model counting, we
associate each \disc{} program with a pair of Boolean formulae in a manner similar to
\citet{holtzen2020scaling}: (1) an \emph{accepting formula} $\alpha$ that encodes the
paths through the program that does not violate observations; and (2)
a \emph{distribution formula} $\varphi$ such that $\WMC(\varphi \land \alpha)$
is the unnormalized probability of the program returning $\true$.
For instance, we would compile \cref{fig:etwocoins:program} into accepting formula $\varphi = f_X$
and $\alpha = f_X \lor f_Y$, where $f_X$ is a Boolean variable that represents the outcome
of $\eflip~0.4$ and $f_Y$ represents the outcome of $\eflip~0.3$. Then, the weight
function is $w(f_X) = 0.4, w(\overline{f_X}) = 0.6, w(f_Y) = 0.4, w(f_Y) = 0.3, w(\overline{f_Y}) = 0.7$.
Then, we can compute the semantics of \cref{fig:etwocoins:program} by performing
weighted model counting:
\begin{align*}
	\dbracket{\cref{fig:etwocoins:program}}(\true) =
	\frac{\WMC(\varphi \wedge \alpha, w)}{\WMC(\alpha, w)} =
	\frac{\WMC(f_X, w)}{\WMC(f_X \lor f_Y, w)} =
	\frac{0.4}{0.4 + 0.6 \cdot 0.3} = \frac{0.4}{0.58} \approx 0.689
\end{align*}

The weighted model counting task is well-studied, and there is an array of
high-performance implementations for solving
it~\citep{holtzen2020scaling,chavira2008probabilistic,sang2005performing}.  One
approach that is particularly effective is \emph{knowledge compilation}, which
compiles the Boolean formula into a representation for which weighted model
counting can be performed efficiently (typically, polynomial-time in the size of
the compiled representation). A common target for this compilation process is
\emph{binary decision diagrams} (BDDs), shown in \cref{fig:etwocoins:bdds}.
A BDD is a rooted DAG whose internal nodes are labeled with
Boolean variables and whose leaves are labeled with either true or false values.
A BDD is read top-down: solid edges denote true assignments to variables,
and dashed edges denote false assignments.
Once a Boolean
formula is compiled to a BDD, inference can be performed in polynomial time (in the
size of the BDD) by performing a bottom-up traversal of the DAG.

While highly effective for discrete probabilistic inference tasks with finite domains,
inference via knowledge compilation has a critical weakness: it cannot support continuous
random variables or unbounded discrete random variables due to the requirement that each
program be associated with a (finite) Boolean formula. Hence, the design of \disc{}
must be carefully restricted to only permit programs that can be compiled
to Boolean formulae, which is why it does not contain syntactic support for these features.

\subsubsection{Approximate inference via sampling}
\label{sec:introsamp}
A powerful alternative to exact inference is approximate inference via
sampling. The engine that drives sampling-based inference is the
expectation estimator.
The expectation estimator is widely used as a foundation for approximate
inference strategies for probabilistic
programs~\citep{carpenter2017stan,mansinghka2014VentureHigherorder,staton2016SemanticsProbabilistic,ramsey2002stochastic,murray2018DelayedSampling,culpepper2017ContextualEquivalence}.
We will use it to give an inference algorithm for \cont{}.
Concretely, suppose we want to use the
expectation estimator to approximate the semantics of the \cont{} program
$\dbracket{\sflip ~ 1/4}$. To do this, we can draw $N=100$ samples
from the program: in roughly $1/4$ of these samples, the program will output $\true$.
This approach is known as \emph{direct sampling}, and is one way of utilizing the expectation
estimator to design approximate inference algorithms.

Formally, let $\Omega$ be a sample space, $\Pr$ a probability density function, and
let $X : \Omega \rightarrow \real$ be a real-valued random variable out of the sample
space. Then, the expectation of $X$ is defined as
$\E_{\Pr}[X] = \int \Pr(\omega)X(\omega) d\omega$.
The \emph{expectation estimator} approximates
the expectation of a random variable $X$ by drawing $N$ samples from $\Pr$:
\begin{align}
	\label{eq:expect}
	\E_{\Pr}[X] \approx \frac{1}{N} \sum_{\x \sim \Pr}^N X(\x).
\end{align}

There are many more advanced
approaches to sampling-based inference beyond direct sampling such as
Hamiltonian Monte-Carlo~\citep{carpenter2017stan,neal2011mcmc}; at their core,
all these approximate inference algorithms follow the same principle of drawing some number of samples
from the program and using that to estimate the semantics.

When compared with the exact inference strategy described in \cref{sec:discinf},
sampling-based inference has the key advantage that it only requires the ability
to sample from the probabilistic program: each time a random quantity is
introduced, it can be dealt with by eagerly sampling. This makes sampling an
ideal inference algorithm for implementing flexible and expressive languages
with many features: unlike \disc{}, it is straightforward to add interesting
features like continuous random variable and unbounded loops to \cont{} without
wholesale redesigning of its inference algorithm. This gain in expressivity
comes at the cost of precision: unlike \disc{}, \cont{} is only able to provide
an approximation to the final expectation.

\begin{figure}
	\begin{subfigure}[b]{0.37\linewidth}
		\begin{lstlisting}[mathescape=true,escapechar=|]
$\sletin{\sco x}{\sflip~0.20}{}$|\label{line:stwocoins:sample}|
$\sco\openbanana \eletin{\eco Y}{\eflip~0.25}{}$
$\phantom\openbanana \eobserve~\eco{\esample[x] \lor Y}~\einkw$  |\label{line:stwocoins:obs}|
$\phantom\openbanana \eret~{\eco Y}\sco{\closebanana_S}$
\end{lstlisting}
		\caption{Motivating example.}
		\label{fig:interop1:a}
	\end{subfigure}
	\begin{subfigure}[b]{0.3\linewidth}
		\begin{lstlisting}[mathescape=true,escapechar=|]
$\sco\openbanana \eletin{\eco Y}{\eflip~0.25}{}$
$\phantom\openbanana \eobserve~\eco{\etrue \lor Y}~\einkw$|\label{line:stwocoins:obs:b}|
$\phantom\openbanana \eret~{\eco Y}\sco{\closebanana_S}$
\end{lstlisting}
		\caption{Sampled $\LS{x} = \true$.}
		\label{fig:interop1:b}
	\end{subfigure}
	\begin{subfigure}[b]{0.3\linewidth}
		\begin{lstlisting}[mathescape=true,escapechar=|]
$\sco\openbanana \eletin{\eco Y}{\eflip~0.25}{}$
$\phantom\openbanana \eobserve~\eco{\efalse \lor Y}~\einkw$|\label{line:stwocoins:obs:c}|
$\phantom\openbanana \eret~{\eco Y}\sco{\closebanana_S}$
\end{lstlisting}
		\caption{Sampled $\LS{x} = \false$.}
		\label{fig:interop1:c}
	\end{subfigure}

	\caption{Interpreting \cont{} values in \disc{}.}
	\label{fig:interop1}
\end{figure}

\subsection{Sound interoperation}\label{sec:overview:interop}
Now we move on to our main goal, establishing sound interoperation between
the underlying inference strategies of \disc{} and \cont{} by identifying two key invariants that must be maintained
when transporting \disc{} and \cont{} values across boundaries: \emph{importance
	weighting} and \emph{sample consistency}.
At first, there appears to be a straightforward way of establishing interoperation
between these two languages: when a \cont{} value $\LS{v_s}$ is interpreted in
a \disc{} context, it is lifted in a Dirac-delta distribution on $\LS{v_s}$.
\Cref{fig:interop1:a} gives an illustration of this scenario: first, on \cref{line:stwocoins:sample}
we sample a value (either $\true$ or $\false$) for $\LS{x}$. Then, on \cref{line:stwocoins:obs} we
interpret $\LS{x}$ within an exact context. \Cref{fig:interop1:b,fig:interop1:c}
show the two possible liftings: the sampled value is given its straightforward interpretation
in the exact context.
When a \disc{} value $\LE{v_e}$ is interpreted in a \cont{} context,
one can draw a sample $\LS{v_s}$ from the \emph{exact} distribution denoted by $\LE{v_e}$.
However, we will show in the next two sub-sections that a naive approach that fails to preserve key
invariants will result in incorrect inference results, and that one must maintain
careful invariants in order to ensure soundness of inference across boundaries. %

\subsubsection{Importance weighting}\label{sec:importance-weighting} Let us more carefully
consider the situation shown in \cref{fig:interop1:a}. First, we observe that
the desired semantics is $\dbracket{\cref{fig:interop1:a}}(\true) = 0.25 / 0.4 = 0.625$. Suppose we were to follow
a naive multi-language inference procedure of drawing 100 samples
by eagerly evaluating values for $\LS{x}$. Following \cref{sec:introsamp}, approximately 20 of these samples will yield the program in
\cref{fig:interop1:b} and approximately 80 will yield the program in \cref{fig:interop1:c}.
Observe that
$\dbracket{\cref{fig:interop1:b}}(\true) = 0.25$ and $\dbracket{\cref{fig:interop1:c}}(\true) = 1$. So
our naive estimate $\dbracket{\cref{fig:interop1:a}}(\true)$ would be:
\begin{align}
	\label{eq:wrong}
	\dbracket{\cref{fig:interop1:a}}(\true) \stackrel{?}{\approx}
	\frac{20}{100} \dbracket{\cref{fig:interop1:b}}(\true) + \frac{80}{100} \dbracket{\cref{fig:interop1:c}}(\true)
	= 0.85
\end{align}
Something went wrong -- we expected the result of \cref{eq:wrong} to be 0.625. This naive
sampling approach significantly over-estimated $\dbracket{\cref{fig:interop1:a}}(\true)$.
The issue is that, in this naive approach, the observation that occurs on
\cref{line:stwocoins:obs}~(\cref{fig:interop1:a}) is not taken into account when
sampling a value for $\LS{x}$: samples where $\LS{x}=\true$ are under-sampled
relative to their true probability, and samples where $\LS{x}=\false$ are
over-sampled.

This example illustrates that a naive approach to interoperation is unsound. To fix it,
one approach is to adjust the \emph{relative importance} of each sample: we will
still sample $\LS{x}=\true$ roughly 20\% of the time, but we will decrease the overall
importance of this sample. The key idea comes from \emph{importance sampling}, which is a
refinement of the expectation estimator given in \cref{eq:expect} but enables estimating
an expectation $\E_{p}[X]$ by sampling from a \emph{proposal distribution} $q$:
\begin{align}\label{eq:is}
	\mathbb{E}_{p}[X] = \int X(\x)p(\x)d\x = \int X(\x)\frac{p(\x)}{q(\x)}q(\x) d\x = \E_{q} \left[X(x) \frac{p(\x)}{q(\x)}\right].
\end{align}

The above holds as long as the proposal $q$ \emph{supports} $p$ (i.e., satisfies
the property that, for all $\x$, if $p(x) > 0$ then $q(x) > 0$). The ratio
$p(\x)/q(\x)$ is called the \emph{importance weight} of the sample $\x$:
intuitively, if $\x$ is more likely according to the true distribution $p$ than
the proposal $q$, the importance ratio will be greater than 1; similarly, if
$\x$ is less likely according to $p$ than $q$, its weight will be less than 1.
In this instance, the proposal $q$ is semantics of the program with all observe
statements deleted, and $p$ is $\dbracket{\cref{fig:interop1:a}}$.

Unfortunately, already having access to $p$ defeats the purpose of
approximating. In addition, our programs $p$ always incorporate a normalization
constant $Z$, such that
\begin{equation}\label{eq:phat}
	p(\x) = \hat{p}(\x) / Z,
\end{equation}
with $\hat{p}$ being the unnormalized distribution. Summing the probability of
$\hat{p}$ for all $x$ in the domain of $p$ yields
$Z=\int \hat{p}(\x)d\x$. Computing this normalization constant is
expensive, and amounts to calculating $p$ directly.  In our setting, calculating
this normalization constant is identical to the
denotation of \cref{line:stwocoins:obs} in the exact setting. To avoid solving
for this in our importance sampler, we can incorporate \cref{eq:phat} into our expectation
\cref{eq:is} and jointly approximate our query alongside $Z$,

\begin{align}\label{eq:snis}
	\mathbb{E}_{p}[X] = \int X(\x)p(\x) d\x = \frac{\int X(\x)\hat{p}(\x)d\x }{\int \hat{p}(\x)d\x } = \frac{\int X(\x)\frac{\hat{p}(\x)}{q(\x)}q(\x)d\x }{\int \frac{\hat{p}(\x)}{q(\x)}q(\x)d\x } = \frac{\E_{\x \sim q} \left[X(x) \frac{\hat{p}(\x)}{q(\x)}\right]}{\E_{\x \sim q} \left[\frac{\hat{p}(\x)}{q(\x)}\right]}.
\end{align}

The above is called a \emph{self-normalized importance sampler}~\cite{robert1999monte}.
Here, in the denominator, we construct the normalizing constant for $q$ to be the ratio of the unnormalized $\hat{p}$ to $q$: the \cref{line:stwocoins:obs:b} in
\cref{fig:interop1:b} when $\sco x = \true$ and the \cref{line:stwocoins:obs:c}
in \cref{fig:interop1:c} when $\sco x = \false$. Notice that the probability of evidence encoded by $\eobserve$ statements in \cref{fig:interop1:b} and \cref{fig:interop1:b} are
$\dbracket{\eco{\true \lor Y}}(\true) = 1$ and
$\dbracket{\eco{\false \lor Y}}(\true) = 0.25$, respectively.

Sampling 100 draws of $\sco x$, again, with 20 samples yielding the program in
\cref{fig:interop1:b} and 80 samples yielding the program in \cref{fig:interop1:c}, \cref{eq:snis} now returns our expected result:
\begin{align*}
	\dbracket{\cref{fig:interop1:a}}(\true) \approx
	\frac{\frac{20}{100} \, 1 \cdot \dbracket{\cref{fig:interop1:b}}(\true) + \frac{80}{100} \, 0.25 \cdot \dbracket{\cref{fig:interop1:c}}(\true)}{\frac{20}{100} \, 1 + \frac{80}{100} \, 0.25 }
	= \frac{20\cdot 0.25 + 80 \cdot 0.25}{20 + 20} = 0.625
\end{align*}

\subsubsection{Sample consistency}\label{sec:sample-consistency}
\begin{wrapfigure}{r}{0.3\linewidth}
	\begin{lstlisting}[mathescape=true,caption={\textsc{Sample consistency}},label={lst:interop3},escapechar=|]
$\eletin{\eco X}{\eflip~0.5}{}$
$\eco\openbanana \sletin{\sco y}{\sexact[\eco X]}{}$
$\phantom\openbanana \sletin{\sco z}{\sexact[\eco X]}{}$
$\phantom\openbanana \sret~{\sco{y \land z}}\eco{\closebanana_E}$
\end{lstlisting}
\end{wrapfigure}

Importance weighting is not all that is necessary to ensure sound interoperability:
we must also ensure that \disc{} values are safely interpreted with a \cont{} context.
Consider the example in \cref{lst:interop3}. There are two observations to make
about this program. The first is that we embed a \cont{} program into a \disc{}
context; this results in a sampler that evaluates all \cont{} fragments while
preserving the semantics of all \disc{} variables in order to produce a sample.
The next thing to notice is that a \disc{} program denotes a distribution; in
the semantics of \cont{}, when we come across a distribution a sample is
immediately drawn from it.

Again, we can propose a naive strategy for performing inference on this program:
one where we draw a new sample each time
we encounter a distribution.  Notice that Line 2 holds a reference $\eco X$ to
$\eflip~0.5$, denoting a Bernoulli distribution. When we evaluate this boundary,
with probability $1/2$ we sample $\sco y = \true$; suppose we sample $\sco y = \true$. We encounter this reference,
again, on Line 3 and suppose we sample $\sco z = \false$.
Finally, on Line 4, we evaluate the Boolean expression, resulting in $\false$,
which is lifted into the Dirac-delta distribution in \disc{}. Running this
program $n$ number of times, we will expect to see the expectation of
$\sco{y \land z}$ with $\sco y$ and $\sco z$ as two independent draws of the
fair Bernoulli distribution.
At this point, something strange has occurred: by referencing a single variable
in \disc{}, we have simulated two independent $\eflip$s.

Intuitively, the sampled value for $\LS{z}$ must be
\emph{the same} as the sampled value for $\LS{y}$. Operationally, to ensure this
is the case, any samples drawn across at the $\sexact$ boundary additionally
constraining \disc{} program's accepting criteria so that all subsequent samples
remain consistent.

\section{\multippl: Multi-Language Probabilistic Programming}
\label{sec:multippl}

In this section we present \multippl{},
a multilanguage that supports both exact and sampling-based inference.
\cref{sec:syntax,sec:typing} describe the syntax of \multippl{} programs
and \multippl{}'s type system.
We then present two semantic models of \multippl{}.
First, \cref{sec:high-level-model} presents a high-level model $\erasesem-$
capturing the probability distribution generated by a \multippl{}
program; this model specifies the intended behavior of our implementation.
Second, \cref{sec:low-level-model} presents a low-level model $\implsem-$
capturing our particular inference strategy; taking the intuition we have built up in \cref{sec:overview:interop} and providing the precise way in which our implementation
combines knowledge compilation with importance sampling.
Finally,
\cref{sec:soundness} connects these two models:
we show that $\implsem-$ soundly refines $\erasesem-$,
establishing sound inference interoperability between \disc{} and \cont{}
with respect to our inference strategy.

\subsection{Syntax}\label{sec:syntax}

The syntax of \multippl{} is given in \cref{fig:syntax}.
\multippl{} is a union of two sublanguages,
\disc{} and \cont{}, that support exact and sampling-based inference.
To streamline
the presentation of the models in \cref{sec:high-level-model,sec:low-level-model},
each sublanguage is then subdivided into
pure and effectful fragments.

The sampling-based sublanguage \cont{} is a first-order probabilistic
programming language with Booleans, tuples, and real numbers. In \cont{}, the
pure fragment includes not only the basic operations on Booleans and pairs, but
also arithmetic operations on real numbers. The effectful fragment
additionally includes primitive operations $\sunif~\seone~\setwo$ for
generating uniformly-distributed real numbers in the interval $[\seone,\setwo]$,
$\spois~\scoe$ for generating Poisson-distributed integers with rate $\scoe$,
and $\sobs$-expressions denoting conditioning operators for these
distributions.

The exact sublanguage \disc{}, reminiscent of \dice{}~\cite{holtzen2020scaling},
is a discrete first-order probabilistic programming language with Booleans and tuples.
The pure fragment of \disc{} includes the basic operations on Booleans and
pairs, while the effectful fragment includes constructs for sequencing and
branching, as well as the primitive operations $\eflip~\scoe$ -- for
generating Bernoulli-distributed Booleans with parameter $\scoe$ of type $\sreal$,
and $\eobserve~\eM$ -- for conditioning on an event $\eM$.

The \disc{}
branching construct $\eite\scoe\eM\eN$ requires the guard $\scoe$ to be a
\cont{} term.
This is not an essential restriction, but rather
required for sound inference interoperability
with respect to the specific implementation strategy we have chosen.
As sketched in \cref{sec:discinf,sec:introsamp},
standard sampling-based inference maintains a weight for the current trace,
while exact inference maintains a weight map and an accepting formula.
In our implementation, we wanted a language whose inference algorithm
would stay as close to these traditional
inference algorithms as possible while avoiding incorrect weighting schemes.
To do this while maintaining safe inference interoperability,
one must have the rather subtle invariant
that $\eitekw$ expressions in the \disc{}
sublanguage have then- and else- branches that importance-weight
their respective traces by the same amount.
The syntactic restriction on $\eite\scoe\eM\eN$ is a simple way of
ensuring this is always the case: probabilistic choice is removed, and only one branch need ever be considered.
In our implementation, we also permit $\eitekw$ expressions
where both branches are boundary-free \disc{} programs, as exact inference
for such programs can be performed just as in \citet{holtzen2020scaling}, without
touching the importance weight.
These special cases could
be avoided by maintaining
an auxiliary Boolean formula tracking a \emph{path condition},
which encodes during inference the then- and else- branches
of $\eitekw$ expressions taken to reach a given subterm.
This would allow arbitrary $\eitekw$ expressions in the \disc{}
sublanguage, at the expense of additional overhead of maintaining
this path condition during inference.
In our design of \multippl{}, we have decided to restrict the
syntax of the language rather than impose a performance cost;
in practice, this has been sufficient to express
all of the examples in \cref{sec:eval}.

\subsection{Typing}
\label{sec:typing}

\begin{figure}
	{\small
		\fbox{$\eA\convrel\stau$}
		\begin{mathpar}
			\inferrule{~}{\eunit \convrel \sunit}
			\and
			\inferrule{~}{\ebool \convrel \sbool}
			\and
			\inferrule{\eA\convrel \stau \\ \eB\convrel\ssigma}
			{\eA\etimes\eB\convrel \stau\stimes\ssigma}
		\end{mathpar}}
	\caption{\label{fig:convrel}Rules for convertibility between \disc{} types $\eA$ and \cont{} types $\stau$.}
\end{figure}

\begin{figure}
	{\small
		\fbox{$\ecomptydefault\eM\eA$}
		\begin{mathpar}
			\inferrule{\epuretydefault\eM\eA}{\ecomptydefault{\eret~\eM}{\eA}}
			\and
			\inferrule{
				\ecomptydefault\eM\eA
				\\
				\ecompty\sGamma{\eDelta,\eco X\ofty \eA}\eN\eB}
			{\ecomptydefault{\eletin{\eco X}\eM\eN}{\eB}}
			\and
			\inferrule {
				\spuretydefault\scoe\sbool
				\\
				\ecomptydefault\eM\eA
				\\
				\ecomptydefault\eN\eA
			} {\ecomptydefault{\eite\scoe\eM\eN}{\eA}}
			\and
			\inferrule {\spuretydefault\scoe\sreal}{\ecomptydefault{\eflip~\scoe}{\ebool}}
			\and
			\inferrule {\epuretydefault\eM\ebool} {\ecomptydefault{\eobserve~\eM}{\eunit}}
			\and
			\inferrule {\scomptydefault{\scoe}{\stau}
				\\ \eA\convrel\stau} {\ecomptydefault{\esample[\scoe]}{\eA}}
		\end{mathpar}

		\vspace{1cm}
		\fbox{$\scomptydefault\scoe\stau$}
		\begin{mathpar}
			\inferrule{\spuretydefault{\scoe}\stau}{\scomptydefault{\sret~{\scoe}}{\stau}}
			\and
			\inferrule{
				\scomptydefault{\seone}\ssigma
				\\
				\scompty{\sGamma,\sco x\ofty \ssigma}{\eDelta}{\setwo}\stau}
			{\scomptydefault{\sletin {\sco x}{\seone}{\setwo}}{\stau}}
			\and
			\inferrule {
				\spuretydefault{\seone}\sbool
				\\
				\scomptydefault{\setwo}\stau
				\\
				\scomptydefault{\sethr}\stau
			} {\scomptydefault{\site\seone\setwo\sethr}{\stau}}
			\and
			\inferrule {\spuretydefault\scoe\sreal}{\scomptydefault{\sflip~\scoe}{\sbool}}
			\and
			\inferrule {
				\spuretydefault{\seone}{\sreal}\\
				\spuretydefault{\setwo}{\sreal}
			}
			{\scomptydefault{\sunif~\seone~\setwo}{\sreal}}
			\and
			\inferrule {\spuretydefault\scoe\sreal} {\scomptydefault{\spois~\scoe}{\sreal}}

			\and

			\inferrule {\spuretydefault{\seobs}\sbool  \\ \spuretydefault\seone\sreal}
			{\scomptydefault{\sobserve{\seobs}{\sflip~\seone}}{\sunit}}

			\and

			\inferrule {
				\spuretydefault{\seobs}{\sreal}\\
				\spuretydefault{\seone}{\sreal}\\
				\spuretydefault{\setwo}{\sreal}
			}
			{\scomptydefault{\sobserve{\seobs}{\sunif~\seone~\setwo}}{\sunit}}

			\and

			\inferrule
			{\spuretydefault{\seobs}{\sreal}\\
				\spuretydefault\scoe\sreal}
			{\scomptydefault{\sobserve{\seobs}{\spois~\scoe}}{\sunit}}

			\and

			\inferrule {\ecomptydefault\eM{\eA}
				\\ \eA\convrel\stau} {\scomptydefault{\sexact[\eM]}{\stau}}
		\end{mathpar}}
	\caption{\label{fig:effectful-typing}Typing rules for the effectful fragment of \multippl{}.}
\end{figure}

The syntax of types and typing contexts is given in \cref{fig:syntax}.
\disc{} types $\eA$ include Booleans and pairs;
\cont{} types $\stau$ additionally include a type of real numbers.
A \disc{} typing context $\eDelta$ is a mapping of \disc{} variables to \disc{} types,
and a \cont{} typing context $\sGamma$ is a mapping of \cont{} variables to \cont{} types.
By convention we will denote \disc{} syntactic elements with capital
letters and \cont{} elements with lower-case Greek letters. This section is best read
in color, where we use orange monotype font for \cont{} terms and purple
sans-serif font for \disc{} terms.

\multippl{} contains two sublanguages that each have a pure and effectful part,
so there are correspondingly four forms of typing judgment. For the pure fragments,
\begin{itemize}
	\item $\epuretydefault\eM\eA$ says the pure \disc{} term $\eM$ has \disc{} type $\eA$ in \disc{} context $\eDelta$.
	\item $\spuretydefault\scoe\stau$ says the pure \cont{} term $\scoe$ has \cont{} type $\stau$ in \cont{} context $\sGamma$.
\end{itemize}
These judgments are standard and deferred to the appendix.

The typing judgments for effectful \multippl{} terms
are parameterized by a combined context $\sGamma;\eDelta$,
as an effectful term may mention variables
from both \disc{} and \cont{} via boundaries:
\begin{itemize}
	\item $\ecomptydefault\eM\eA$ says the effectful \disc{} term $\eM$ has \disc{} type $\eA$ in combined context $\sGamma;\eDelta$.
	\item $\scomptydefault\scoe\stau$ says the effectful \cont{} term $\scoe$ has \cont{} type $\stau$ in combined context $\sGamma;\eDelta$.
\end{itemize}
These judgments are defined in \cref{fig:effectful-typing}.
Note that in the rule
for $\eflip~\scoe$, the parameter $\scoe$
can be an arbitrary pure \cont{} term; this allows expressing the \textsc{TwoCoins} example
from \cref{sec:motivation}.
In principle, one could allow arbitrary effectful \cont{} programs $\scoe$ as
parameter to $\eflip$ instead of just pure ones, but we have not found this to
be useful in practice.
The typing judgments for the boundaries $\esample[\scoe]$ and $\sexact[\eM]$
allow converting  \disc{} terms of type $\eA$ into \cont{} terms of type $\stau$
and vice versa,
so long as $\eA$ and $\stau$
are \emph{convertible},
written $\eA\convrel\stau$.
The convertibility relation
is defined in
\cref{fig:convrel}; it simply states that
\disc{} types can be converted into
their \cont{} counterparts in the expected way,
and that the \cont{} type $\sreal$ has no \disc{} counterpart.

\subsection{High-level semantic model} \label{sec:high-level-model}

\newcommand\scomptyclosed[2]{\scompty{\sco{\bigcdot}\,}{\eco{\bigcdot}}{#1}{#2}}

\begin{figure}
	{\footnotesize
		\[
			\begin{aligned}
				\llbrspc\eunit          & = \{\star\}                    \\
				\llbrspc\ebool          & = \{\top,\bot\}                \\
				\llbrspc{\eA\etimes\eB} & = \llbrspc\eA\times\llbrspc\eB
			\end{aligned}
			\hspace{1em}
			\begin{aligned}
				\llbr\sunit                & = \{\star\}                    \\
				\llbr\sbool                & = \{\top,\bot\}                \\
				\llbr\sreal                & = \R                           \\
				\llbr{\ssigma\stimes\stau} & = \llbr\ssigma\times\llbr\stau
			\end{aligned}
			\hspace{1em}
			\begin{aligned}
				 & \llbr{\eDelta} = \prod_{\eco{X}\in\dom\eDelta}\llbrspc{\eDelta(\eco X)}
			\end{aligned}
			\hspace{1em}
			\begin{aligned}
				\llbr{\sGamma} & = \prod_{\sco{x}\in\dom\sGamma}\llbr{\sGamma(\sco x)}
			\end{aligned}
		\]}
	\caption{\label{fig:high-level-types}Interpreting types and typing contexts.}
\end{figure}

\begin{figure}
	{\footnotesize
		\[
			\begin{aligned}
				\puresemarg{\epuretydefault{\eM}{\eA}}                               & : \llbrspc\eDelta\to\llbrspc\eA \\
				\puresemarg{\epuretydefaultsem{\eco X}{\eA}}(\delta)                 & = \delta(\eco X)                \\
				\puresemarg{\epuretydefaultsem{\etrue}{\ebool}}(\delta)              & = \top                          \\
				\puresemarg{\epuretydefaultsem{\efalse}{\ebool}}(\delta)             & = \bot                          \\
				\puresemarg{\epuretydefaultsem{\eM\eand\eN}{\ebool}}(\delta)         & =
				\begin{cases}
					\top, &
					\text{if }
					\puresemarg{\epuretydefaultsem{\eM}{\ebool}}(\delta)
					=
					\puresemarg{\epuretydefaultsem{\eN}{\ebool}}(\delta)
					= \top
					\\
					\bot, & \text{otherwise}
				\end{cases}
				\\
				\puresemarg{\epuretydefaultsem{\enot\eM}{\ebool}}(\delta)            & =
				\begin{cases}
					\bot, &
					\text{if }
					\puresemarg{\epuretydefaultsem{\eM}{\ebool}}(\delta)
					= \top
					\\
					\top, & \text{otherwise}
				\end{cases}
				\\
				\puresemarg{\epuretydefaultsem{\eunitelt}{\eunit}}(\delta)           & = \star                         \\
				\puresemarg{\epuretydefaultsem{\epair\eM\eN}{\eA\etimes\eB}}(\delta) & =
				(\puresemarg{\epuretydefaultsem\eM\eA}(\delta),
				\puresemarg{\epuretydefaultsem\eN\eB}(\delta))                                                         \\
				\puresemarg{\epuretydefaultsem{\efst~\eM}{\eA}}(\delta)              & =
				\pi_1(\puresemarg{\epuretydefaultsem{\eM}{\eA\etimes\eB}}(\delta))
				\\
				\puresemarg{\epuretydefaultsem{\esnd~\eM}{\eB}}(\delta)              & =
				\pi_2(\puresemarg{\epuretydefaultsem{\eM}{\eA\etimes\eB}}(\delta))
			\end{aligned}
			\hspace{2em}
			\begin{aligned}
				\puresemarg{\spuretydefault{\scoe}{\stau}}                                       & : \llbr\sGamma\mbleto{\llbr\stau} \\
				\puresemarg{\spuretydefaultsem{\sco x}{\stau}}(\gamma)                           & = \gamma(\sco x)                  \\
				\puresemarg{\spuretydefaultsem{\strue}{\sbool}}(\gamma)                          & = \top                            \\
				\puresemarg{\spuretydefaultsem{\sfalse}{\sbool}}(\gamma)                         & = \bot                            \\
				\puresemarg{\spuretydefaultsem{r}{\sreal}}(\gamma)                               & = r                               \\
				\puresemarg{\spuretydefaultsem{\seone\sadd\setwo}{\sreal}}(\gamma)               & =
				\puresemarg{\spuretydefaultsem{\seone}\sreal}(\gamma)+
				\puresemarg{\spuretydefaultsem\setwo\sreal}(\gamma)                                                                  \\
				\puresemarg{\spuretydefaultsem{\sneg\scoe}{\sreal}}(\gamma)                      & =
				-\puresemarg{\spuretydefaultsem{\scoe}\sreal}(\gamma)                                                                \\
				\puresemarg{\spuretydefaultsem{\seone\smul\setwo}{\sreal}}(\gamma)               & =
				\puresemarg{\spuretydefaultsem{\seone}\sreal}(\gamma)\cdot
				\puresemarg{\spuretydefaultsem\setwo\sreal}(\gamma)                                                                  \\
				\puresemarg{\spuretydefaultsem{\seone\sle\setwo}{\sbool}}(\gamma)                & =
				\begin{cases}
					\top, & \text{if }\puresemarg{\spuretydefaultsem{\seone}\sreal}(\gamma)
					\le \puresemarg{\spuretydefaultsem\setwo\sreal}(\gamma)                 \\
					\bot, & \text{otherwise}
				\end{cases}                                              \\
				\puresemarg{\spuretydefaultsem{\sunitelt}{\sunit}}(\gamma)                       & = \star                           \\
				\puresemarg{\spuretydefaultsem{\spair\seone\setwo}{\ssigma\stimes\stau}}(\gamma) & =
				(\puresemarg{\spuretydefaultsem{\seone}{\ssigma}}(\gamma),
				\puresemarg{\spuretydefaultsem{\setwo}{\stau}}(\gamma))
				\\
				\puresemarg{\spuretydefaultsem{\sfst~\scoe}{\ssigma}}(\gamma)                    & =
				\pi_1(\puresemarg{\spuretydefaultsem{\scoe}{\ssigma\stimes\stau}}(\gamma))
				\\
				\puresemarg{\spuretydefaultsem{\ssnd~\scoe}{\stau}}(\gamma)                      & =
				\pi_2(\puresemarg{\spuretydefaultsem{\scoe}{\ssigma\stimes\stau}}(\gamma))
				\\
			\end{aligned}
		\]}
	\caption{\label{fig:high-level-pure}Interpreting pure terms.}
\end{figure}

\begin{figure}
	{\footnotesize
		\fbox{$\erasesem{\ecomptydefault{\eM}{\eA}}:\llbr\sGamma\times\llbr\eDelta \to \prmonad{\llbrspc\eA}$}
		\vspace{1em}
		\[
			\begin{aligned}
				\erasesemarg{\ecomptydefaultsem{\eret~\eM}{\eA}}(\gamma,\delta)             & =
				\gret(\puresemarg{\epuretydefaultsem{\eM}{\eA}}(\delta))
				\\
				\erasesemarg{\ecomptydefaultsem{\eletin{\eco X}\eM\eN}{\eB}}(\gamma,\delta) & =
				\left(\begin{aligned}
						       & x \gets \erasesemarg{\ecomptydefaultsem\eM\eA}(\gamma,\delta);                                  \\
						       & \erasesemarg{\ecomptysem\sGamma{\eDelta,\eco X\ofty \eA}\eN\eB}(\gamma,\delta[\eco X\mapsto x])
					      \end{aligned}\right)                                                        \\
				\erasesemarg{\ecomptydefaultsem{\esample[\scoe]}{\eA}}(\gamma,\delta)       & = \valstoe\eA\stau{\erasesemarg{\scomptydefaultsem{\scoe}{\stau}}(\gamma,\delta)} \\
			\end{aligned}
			\hspace{2em}
			\begin{aligned}
				\erasesemarg{\ecomptydefaultsem{\begin{aligned}
							                                 & \eif~\scoe \\
							                                 & \ethen~\eM \\
							                                 & \eelse~\eN
						                                \end{aligned}}{\eA}}(\gamma,\delta)                        & =
				\left(\begin{aligned}
						       & \mathrm{if}~\puresemarg{\spuretydefaultsem\scoe\sbool}(\gamma)      \\
						       & \mathrm{then}~\erasesemarg{\ecomptydefaultsem\eM\eA}(\gamma,\delta) \\
						       & \mathrm{else}~\erasesemarg{\ecomptydefaultsem\eN\eA}(\gamma,\delta)
					      \end{aligned}\right)                                                                     \\
				\erasesemarg{\ecomptydefaultsem{\eflip~\scoe}{\ebool}}(\gamma,\delta)  & = \overline{\flip(\puresemarg{\spuretydefaultsem\scoe\sreal}(\gamma))}  \\
				\erasesemarg{\ecomptydefaultsem{\eobserve~\eM}{\eunit}}(\gamma,\delta) & = \score(\Ind{\puresemarg{\epuretydefaultsem\eM\ebool}(\delta) = \top}) \\
			\end{aligned}
		\]
		\\
		\vspace{2em}
		\fbox{$\erasesem{\scomptydefault{\scoe}{\stau}} : \llbr\sGamma\times\llbr\eDelta \to \prmonad{\llbr\stau}$}
		\vspace{1em}
		\[
			\begin{aligned}
				\erasesemarg{\scomptydefaultsem{\sret~{\scoe}}{\stau}}(\gamma,\delta)                    & = \gret(\puresemarg{\spuretydefaultsem{\scoe}\stau}(\gamma)) \\
				\erasesemarg{\scomptydefaultsem{\sletin {\sco x}{\seone}{\setwo}}{\stau}}(\gamma,\delta) & =
				\left(\begin{aligned}
						       & x\gets \erasesemarg{\scomptydefaultsem{\seone}\ssigma}(\gamma,\delta);                                   \\
						       & \erasesemarg{\scomptysem{\sGamma,\sco x\ofty \ssigma}\eDelta\setwo\stau}(\gamma[\sco x\mapsto x],\delta) \\
					      \end{aligned}\right)
				\\
			\end{aligned}
			\hspace{2em}
			\begin{aligned}
				\erasesemarg{\scomptydefaultsem{\begin{aligned}
							                                 & \sif~\seone   \\
							                                 & \sthen~\setwo \\
							                                 & \selse~\sethr
						                                \end{aligned}}{\stau}}(\gamma,\delta)                      & =
				\left(\begin{aligned}
						       & \mathrm{if}~{\puresemarg{\spuretydefaultsem{\seone}\sbool}(\gamma)}          \\
						       & \mathrm{then}~{\erasesemarg{\scomptydefaultsem{\setwo}\stau}(\gamma,\delta)} \\
						       & \mathrm{else}~{\erasesemarg{\scomptydefaultsem{\sethr}\stau}(\gamma,\delta)}
					      \end{aligned}\right)                                                                 \\
				\erasesemarg{\scomptydefaultsem{\sexact[\eM]}{\stau}}(\gamma,\delta) & = \valetos\eA\stau{\erasesemarg{\ecomptydefaultsem\eM{\eA}}(\gamma,\delta)} \\
			\end{aligned}
		\]
		\vspace{0.5em}
		\[
			\begin{aligned}
				\erasesemarg{\scomptydefaultsem{\sflip~\scoe}{\sreal}}(\gamma,\delta)                            & =
				\overline{\flip(\puresemarg{\spuretydefaultsem\scoe\sreal}(\gamma))}                                                                                                                                                                                                                               \\
				\erasesemarg{\scomptydefaultsem{\sunif~\seone~\setwo}{\sreal}}(\gamma,\delta)                    & = \overline{\unif(
					\puresemarg{\spuretydefaultsem\seone\sreal}(\gamma),
				\puresemarg{\spuretydefaultsem\setwo\sreal}(\gamma))}                                                                                                                                                                                                                                              \\
				\erasesemarg{\scomptydefaultsem{\spois~\scoe}{\sreal}}(\gamma,\delta)                            & =
				\overline{\pois(\puresemarg{\spuretydefaultsem\scoe\sreal}(\gamma))}                                                                                                                                                                                                                               \\
				\erasesemarg{\scomptydefaultsem{\sobserve{\seobs}{\sflip~\seone}}{\sreal}}(\gamma,\delta)        & = \score\!\left(\flip(\puresemarg{\spuretydefaultsem\seone\sreal}(\gamma)) (\puresemarg{\spuretydefaultsem{\seobs}\sbool}(\gamma)) \right)                                                      \\
				\erasesemarg{\scomptydefaultsem{\sobserve{\seobs}{\spois~\seone}}{\sreal}}(\gamma,\delta)        & = \score\!\left(\pois(\puresemarg{\spuretydefaultsem\seone\sreal}(\gamma)) (\puresemarg{\spuretydefaultsem{\seobs}\sreal}(\gamma)) \right)                                                      \\
				\erasesemarg{\scomptydefaultsem{\sobserve{\seobs}{\sunif~\seone~\setwo}}{\sreal}}(\gamma,\delta) & = \score\!\left(\unif(\puresemarg{\spuretydefaultsem\seone\sreal}(\gamma), \puresemarg{\spuretydefaultsem\setwo\sreal}(\gamma)) (\puresemarg{\spuretydefaultsem{\seobs}\sbool}(\gamma)) \right) \\
			\end{aligned}
		\]
	}
	\caption{\label{fig:high-level-effectful} Interpreting effectful terms.
		We use Haskell-style syntactic sugar for the usual monad operations.
	}
\end{figure}

This section defines a high-level model $\erasesem-$ of \multippl{}
to serve as the definition of sound inference interoperability
for the \multippl{} multilanguage.

Setting aside details of any particular inference strategy,
a \multippl{} program $\scomptyclosed\scoe\stau$
should produce a conditional probability distribution over values of type $\stau$.
Following standard techniques for modelling probabilistic programs
with conditioning~\cite{staton2016SemanticsProbabilistic},
we interpret types and typing contexts as measurable spaces,
pure terms as measurable functions,
and effectful terms via a suitable monad.

\cref{fig:high-level-types} gives the interpretations of types and typing contexts.
\disc{} types denote finite discrete measurable spaces
and \cont{} types denote arbitrary measurable spaces.
These interpretations are then lifted to typing contexts in the usual way:
a \disc{} context $\eDelta$ denotes the measurable space of substitutions $\delta$
such that $\delta(\eco x)\in\llbr{\eDelta(\eco x)}$
for all $\eco x\in\dom\eDelta$,
and a \cont{} contexts $\sGamma$ denotes the measurable space of substitutions $\gamma$
such that $\gamma(\sco x)\in\llbr{\sGamma(\sco x)}$
for all $\sco x\in\dom\sGamma$.

\cref{fig:high-level-pure} gives the standard interpretations of pure
terms~\cite{staton2016SemanticsProbabilistic}. Pure \disc{} terms $\epuretydefault\eM\eA$
denote functions $\llbr\eM : \llbr{\eDelta}\to\llbr{\eA}$,
automatically measurable because every \disc{} type denotes a discrete
measurable space.
Pure \cont{} terms $\spuretydefault\scoe\stau$
denote measurable functions $\llbr\scoe : \llbr{\sGamma}\to\llbr{\stau}$.

Following \citet{staton2016SemanticsProbabilistic}, to interpret effectful
terms we make use of the monad $\prmonad A = \giry([0,1]\times A)$, obtained by
combining the writer monad for the monoid $([0,1],\times,1)$ of weights with the
probability monad $\giry$~\citep{giry2006categorical}.  Under this
interpretation, a \multippl{}
program $\scomptyclosed\scoe\stau$ denotes a distribution over pairs $(w,v)$,
where $v$ is a value of type $\stau$ produced by a particular run of $\scoe$ and
$w$ is the weight accumulated by both \cont{} and \disc{} observe expressions.

\cref{fig:high-level-effectful} interprets
effectful \multippl{} terms using $\prmonad$.
A \disc{} term $\ecomptydefault\eM\eA$ is interpreted
as a measurable function $\erasesem\eM:\llbr\sGamma\times\llbr\eDelta\to\prmonad{\llbr\eA}$,
and a \cont{} term $\scomptydefault\scoe\stau$
is interpreted as a measurable function
$\erasesem\scoe:\llbr\sGamma\times\llbr\eDelta\to\prmonad{\llbr\stau}$.
To model the basic probabilistic operations,
the interpretation additionally makes use of the following primitives:
\begin{itemize}[leftmargin=*]
	\item
	      $\overline{(\bigcdot)}:\giry(A)\to\prmonad(A)$ lifts
	      distributions on $A$ into $\prmonad$ by setting $w=1$.
	\item
	      $\score : \weightmonoid\to\prmonad\{\star\}$ sends a weight $w$ to the Dirac
	      distribution $\delta_{(w, \star)}$ centered at $(w,\star)$.
	\item For $p\in \R$, $\mathrm{flip}(p)$ is the Bernoulli distribution on
	      $\{\top,\bot\}$ with parameter $p$ if $p\in [0,1]$ and the Dirac
	      distribution $\delta_\bot$ otherwise.
	\item For $a,b\in \R$, $\unif(a,b)$ is the uniform distribution on $[a,b]$ if
	      $a\le b$ and $\delta_{\min(a,b)}$ otherwise.
	\item For $\lambda\in \R$, $\pois(\lambda)$ is the Poisson distribution with rate
	      $\lambda$ if $\lambda > 0$ and $\delta_0$ otherwise.
\end{itemize}
Boundaries have no effect under this interpretation,
reflecting the idea that changing one's inference strategy
should not change the inferred distribution;
semantic values of \disc{} type are implicitly coerced
into semantic values of \cont{} type and vice versa,
thanks to the following lemma:

\begin{figure}
	{\footnotesize
		\[
			\begin{aligned}
				\implsemarg{\ecomptydefaultsem{\eret~\eM}{\eA}}(\Omega)(\gamma,D)              & =
				\gret(\Omega,\idfn,\puresem{\epuretydefaultsem{\eM}{\eA}}\circ D)
				\\
				\implsemarg{\ecomptydefaultsem{\eletin{\eco X}\eM\eN}{\eB}}(\Omega)(\gamma,D)  & =
				\left(\begin{aligned}
						       & (\Omega_1,f_1,X) \gets \implsemarg{\ecomptydefaultsem\eM\eA}(\Omega)(\gamma,D);          \\
						       & (\Omega_2,f_2,Y) \gets \implsemarg{\eN}(\Omega_1)(\gamma,(D\circ f_1)[\eco X\mapsto X]); \\
						       & \gret(\Omega_2, f_1\circ f_2, Y)
					      \end{aligned}\right)
				\\
				\implsemarg{\ecomptydefaultsem{\eite\scoe\eM\eN}{\eA}}(\Omega)(\gamma,D)       & =
				\left(\begin{aligned}
						       & \mathrm{if}\,{\puresemarg{\spuretydefaultsem\scoe\sbool}(\gamma)}     \\
						       & \mathrm{then}~\implsemarg{\ecomptydefaultsem\eM\eA}(\Omega)(\gamma,D) \\
						       & \mathrm{else}~\implsemarg{\ecomptydefaultsem\eN\eA}(\Omega)(\gamma,D)
					      \end{aligned}\right)     \\
				\implsemarg{\ecomptydefaultsem{\eflip~\scoe}{\ebool}}(\Omega)(\gamma,D)        & =
				\left(\begin{aligned}
						       & p := \ite{\puresemarg{\spuretydefaultsem\scoe\sreal}(\gamma)\in[0,1]}
						      {\puresemarg{\spuretydefaultsem\scoe\sreal}(\gamma)}
						      {0}
						      ;                                                                                          \\
						       & \Omega_{\rm flip} := (\{0,1\},\mu,\{0,1\})\text{ where $\mu(1) = p$};                   \\
						       & \Omega' := \Omega\otimes\Omega_{\rm flip};                                              \\
						       & X := \omega'\mapsto \mathrm{if}~\pi_2(\omega')=1~\mathrm{then}~\top~\mathrm{else}~\bot; \\
						       & \gret(\Omega',\pi_1,X)
					      \end{aligned}\right)
				\\
				\implsemarg{\ecomptydefaultsem{\eobserve~\eM}{\eunit}}(\Omega,\mu,E)(\gamma,D) & =
				\left(\begin{aligned}
						       & F := (\puresem{\epuretydefaultsem\eM\ebool}\circ D)^{-1}(\top); \\
						       & \hilite{\score(\mu|_E(F));}                                     \\
						       & \gret((\Omega,\mu,E\cap F),\idfn, \_ \mapsto \star)
					      \end{aligned}\right)
				\\
				\implsemarg{\ecomptydefaultsem{\esample[\scoe]}{\eA}}(\Omega)(\gamma,D)        & =
				\left(\begin{aligned}
						       & (\Omega',f,x) \gets \implsemarg{\scomptydefaultsem{\scoe}{\stau}}(\Omega)(\gamma,D); \\
						       & \gret(\Omega',f,\_\mapsto \valstoe\eA\stau x)
					      \end{aligned}\right)
			\end{aligned}
		\]}
	\vspace{1.0em}
	{\footnotesize
		\[
			\begin{aligned}
				\implsemarg{\scomptydefaultsem{\sret~{\scoe}}{\stau}}(\Omega)(\gamma,D)                            & =
				\gret(\Omega, \idfn,\puresemarg{\spuretydefaultsem{\scoe}\stau}(\gamma))
				\\
				\implsemarg{\scomptydefaultsem{\sletin {\sco x}{\seone}{\setwo}}{\stau}}(\Omega)(\gamma,D)         & =
				\left(\begin{aligned}
						       & (\Omega_1,f_1,x)\gets\implsemarg{\scomptydefaultsem{\seone}\ssigma}(\Omega)(\gamma,D)  \\
						       & (\Omega_2,f_2,y)\gets\implsemarg{\setwo}(\Omega_1)(\gamma[\sco x\mapsto x],D\circ f_1) \\
						       & \gret(\Omega_2,f_1\circ f_2,y)
					      \end{aligned}\right)
				\\
				\implsemarg{\scomptydefaultsem{\site\seone\setwo\sethr}{\stau}}(\Omega)(\gamma,D)                  & =
				\left(\begin{aligned}
						       & \mathrm{if}~{\puresemarg{\spuretydefaultsem{\seone}\sbool}(\gamma)}            \\
						       & \mathrm{then}~{\implsemarg{\scomptydefaultsem{\setwo}\stau}(\Omega)(\gamma,D)} \\
						       & \mathrm{else}~{\implsemarg{\scomptydefaultsem{\sethr}\stau}(\Omega)(\gamma,D)}
					      \end{aligned}\right)
				\\
				\implsemarg{\scomptydefaultsem{\sflip~\scoe}{\sbool}}(\Omega)(\gamma,D)                            & =
				\left(\begin{aligned}
						       & x\gets \overline{\flip(\puresemarg{\spuretydefaultsem\scoe\sreal}(\gamma))} \\
						       & \gret(\Omega,\idfn,x)
					      \end{aligned}\right)
				\\
				\implsemarg{\scomptydefaultsem{\sunif~\seone~\setwo}{\sreal}}(\Omega)(\gamma,D)                    & =
				\left(\begin{aligned}
						       & x\gets \overline{\unif(\puresemarg{\spuretydefaultsem\seone\sreal}(\gamma),
						      \puresemarg{\spuretydefaultsem\setwo\sreal}(\gamma))}                          \\
						       & \gret(\Omega,\idfn,x)
					      \end{aligned}\right)
				\\
				\implsemarg{\scomptydefaultsem{\spois~\scoe}{\sreal}}(\Omega)(\gamma,D)                            & =
				\left(\begin{aligned}
						       & x\gets \overline{\pois(\puresemarg{\spuretydefaultsem\scoe\sreal}(\gamma))} \\
						       & \gret(\Omega,\idfn,x)
					      \end{aligned}\right)
				\\
				\implsemarg{\scomptydefaultsem{\sobserve{\seobs}{\sflip~\seone}}{\sunit}}(\Omega)(\gamma,D)        & =
				\left(\begin{aligned}
						       & \score(\flip(\puresemarg{\spuretydefaultsem\seone\sreal}(\gamma))\left(\puresemarg{\spuretydefaultsem\seobs\sbool}(\gamma)\right)); \\
						       & \gret(\Omega,\idfn,\star)
					      \end{aligned}\right)
				\\
				\implsemarg{\scomptydefaultsem{\sobserve{\seobs}{\sunif~\seone~\setwo}}{\sunit}}(\Omega)(\gamma,D) & =
				\left(\begin{aligned}
						       & \score(\unif(\puresemarg{\spuretydefaultsem\seone\sreal}(\gamma),
							      \puresemarg{\spuretydefaultsem\setwo\sreal}(\gamma))
						      \left(\puresemarg{\spuretydefaultsem\seobs\sbool}(\gamma)\right));   \\
						       & \gret(\Omega,\idfn,\star)
					      \end{aligned}\right)
				\\
				\implsemarg{\scomptydefaultsem{\sobserve{\seobs}{\spois~\seone}}{\sunit}}(\Omega)(\gamma,D)        & =
				\left(\begin{aligned}
						       & \score(\pois(\puresemarg{\spuretydefaultsem\seone\sreal}(\gamma))
						      \left(\puresemarg{\spuretydefaultsem\seobs\sbool}(\gamma)\right));   \\
						       & \gret(\Omega,\idfn,\star)
					      \end{aligned}\right)
				\\
				\implsemarg{\scomptydefaultsem{\sexact[\eM]}{\stau}}(\Omega)(\gamma,D)                             & =
				\left(\begin{aligned}
						       & ((\Omega',\mu',E'),f,X)\gets \implsemarg{\ecomptydefaultsem\eM{\eA}}(\Omega)(\gamma,D) \\
						       & x\gets \big(\omega'\gets \mu'|_{E'};~ \gret(\valetos\eA\stau{X(\omega')})\big)         \\
						       & \gret((\Omega',\mu',E'\hilite{\cap ~ X^{-1}(x)}), f, x)
					      \end{aligned}\right)
			\end{aligned}
		\]}
	\caption{\label{fig:low-level-effectful}Low-level interpretation of effectful \multippl{} terms.
		Parts crucial for sound inference interoperability are highlighted, appearing in the denotation of $\eobserve~\eM$ and $\sexact[\eM]$. Best read in color.}
\end{figure}

\begin{lemma}[natural embedding]
	If $\eA\convrel\stau$ then $\llbr\eA = \llbr\stau$.
\end{lemma}
\begin{proof} By induction on $\eA\convrel\stau$.  \end{proof}

\subsection{Low-level model} \label{sec:low-level-model}

This section presents a low-level model $\implsem-$ of \multippl{},
capturing the particular details of our inference strategy.

The interpretations of types, typing contexts,
and the pure fragment of \multippl{} are identical to
the ones given in \cref{sec:high-level-model}.
Where $\implsem-$ differs from $\erasesem-$ is in the interpretation
of effectful terms.  Key to this interpretation is the construction of a
suitable semantic domain for interpreting effectful terms in a way that
faithfully reflects the details of our implementation.
We construct this semantic domain by combining
models of exact and sampling-based inference.

Our model of sampling-based inference is entirely standard,
making use of the monad $\prmonad$
of \cref{sec:high-level-model}.
This monad captures the fact that a sampler
performs inference by
drawing weighted samples from
the distribution defined by a probabilistic program~\cite{staton2016SemanticsProbabilistic}.

Our model of exact inference, on the other hand,
is novel. As explained in \cref{sec:discinf}
and documented in full detail in \citet{holtzen2020scaling},
exact inference via knowledge compilation
performs inference by maintaining two pieces of
state: a \emph{weight map} \(w\) associating Boolean literals
to probabilities, and a Boolean formula \(\alpha\),
called the \emph{accepting formula},
that encodes the paths through the program that do not violate
observe statements. The final result of knowledge compilation
is itself a Boolean formula \(\varphi\); the posterior
distribution can then be calculated by performing weighted
model counting on \(\varphi\land \alpha\) and \(\alpha\) with respect
to the weight map \(w\).

The defining trait of this knowledge compilation strategy
is that it \emph{maintains an exact representation of
	the underlying probability space throughout probabilistic program execution}.
At any given moment during knowledge compilation, there is an underlying
\emph{sample space}: the space of models over the collection of Boolean variables
generated so far.
The purpose of the weight map \(w\) is to represent a \emph{distribution}
over this sample space: the probability of a given model can be computed
by multiplying the weights of all of its literals.
Together, the sample space and the weight map form a \emph{probability space},
which is statefully manipulated throughout the knowledge compilation process.
Upon each encounter of a \(\eflip\) command,
the probability space grows: this is implemented by
generating a fresh Boolean variable to represent
the result of the \(\eflip\)
and extending \(w\) accordingly.
The purpose of the accepting formula \(\alpha\)
is to represent an \emph{event} in this probability space:
the event consisting of those models that satisfy \(\alpha\).
Upon each encounter of an \(\eobserve\) command,
this event shrinks: this is implemented
by conjoining the condition being observed onto \(\alpha\).
Finally, the purpose of the output formula \(\varphi\)
is to represent a \emph{random variable}, which is to say
a Boolean-valued function out of the sample space:
the formula \(\varphi\) represents the random variable
that takes values \(\top\) for those models
that satisfy \(\varphi\) and \(\bot\) otherwise.

What is essential about this setup is that
it \emph{maintains a conditional probability space}
\((\Omega,\mu,E)\), consisting of a sample space \(\Omega\) (the space
of models), a probability measure \(\mu\) on it (represented by the weight map),
and an event \(E\) denoting the result of all \(\eobserve\) statements
so far (represented by the accepting formula),
and that it \emph{produces random variables}.
The fact that these probability spaces, events, and random variables
are represented via weighted Boolean formulas, while crucial
for the efficiency of inference, are details of the implementation that
are irrelevant to ensuring safe inference interoperability.
Because of this, our low-level semantics abstracts over these representation
concerns, choosing instead to work directly with probability spaces
and random variables.
Following \citet{li2023lilac},
we model \disc{} programs as statefully manipulating
tuples \((\Omega,\mu,E)\) and producing random variables \(X\).
For example, we model
running the effectful \disc{} program
\[
	{\eco X}\ofty\ebool \vdash_c
	\left(\begin{aligned}
			 & \elet~{\eco Y}~\ebe~{\eflip~1/2}~\einkw \\
			 & \eobserve~(\eco X\eand\eco Y)~\einkw    \\
			 & \eret~{\eco X}
		\end{aligned}\right) : \ebool
\]
given input probability space $(\Omega,\mu,E)$
and random variable $X : \Omega\to\llbr\ebool$
as follows: \begin{itemize}[leftmargin=*]
	\item $\eflip~1/2$ expands the probability space
	      from $(\Omega,\mu,E)$
	      to $(\Omega\times\llbr\ebool, \mu\otimes\Ber1/2, E\times\llbr\ebool)$,
	      and produces the Boolean random variable $Y : \Omega\times\llbr\ebool\to\llbr\ebool$
	      defined by $Y(\omega,b) = b$.
	      This is implemented by generating a new Boolean variable representing \(Y\).
	      Note that $Y$ is defined in terms of the new sample space $\Omega\times\llbr\ebool$.
	      The function $\pi_1 : \Omega\times\llbr\ebool\to\Omega$
	      says how to convert between the old sample space $\Omega$
	      and the new sample space $\Omega\times\llbr\ebool$:
	      the random variable $X$, defined in terms of the old space $\Omega$,
	      can be converted into a random variable $X\circ\pi_1$
	      defined in terms of the new space $\Omega\times\llbr\ebool$ by precomposition
	      with $\pi_1$.
	      Similarly, the conditioning set $E$, a subset of the old space $\Omega$,
	      can be converted into a conditioning set $\pi_1^{-1}(E) = E\times\llbr\ebool$
	      on the new sample space $\Omega\times\llbr\ebool$.
	      In the implementation, these conversions are no-ops: they amount to the fact
	      that a Boolean formula over Boolean variables \(\Gamma\) can be weakened
	      to a Boolean formula over variables \(\Gamma,x\).
	\item $\eobserve~(\eco X\eand\eco Y)$
	      shrinks the new conditioning set $E\times\llbr\ebool$
	      by intersecting it with the subset $G := \{(\omega,b) \mid X(\omega)= Y(b)=\top\}$
	      of $\Omega\times\llbr\ebool$
	      on which $X$ and $Y$ are both $\top$;
	      this produces a new conditioning set
	      $(E\times\llbr\ebool)\cap G$.
	      This is implemented by conjoining the Boolean formula representing \(X\land Y\)
	      onto the accepting formula.
\end{itemize}
In general, we will interpret \multippl{} programs in a semantic domain that combines
this stateful approach to modelling exact inference
with the standard \(\prmonad\)-based approach to modelling sampling-based inference:
a \multippl{} program $\ecomptydefault\eM\eA$
denotes a function that
receives:
\begin{enumerate}
	\item
	      a concrete instantiation $\gamma\in\llbr{\sGamma}$ for free \cont{}
	      variables
	\item
	      a probability space $\Omega$
	      and a random variable $D\in \Omega\to\llbr{\eDelta}$
	      for free \disc{} variables,
\end{enumerate}
and uses the monad $\prmonad$ to produce a weighted sample
consisting of a new probability space $\Omega'$ and a random variable $X\in\Omega'\to\llbr{\eA}$
of outputs. The old and new probability spaces are connected by a function \(f:\Omega'\to\Omega\),
which says
how to convert random variables and events defined in terms of the old space into
random variables and events defined on the new one.
The following definitions make this idea precise.
\begin{definition}
	A \emph{finite conditional probability space}
	is a triple $(\Omega,\mu,E)$
	where
	(1) $\Omega$ is a finite set;
	(2) $\mu : \Omega\to[0,1]$ is a discrete probability distribution,
	and (3) $E$ is a subset of $\Omega$ called the \emph{conditioning set}.
	Let $\fcps$ be the collection of finite conditional probability spaces.
\end{definition}

\newcommand\fcpsto{\xrightarrow\fcps}
\begin{definition}
	A \emph{map of finite conditional probability spaces}
	$f : (\Omega,\mu,E)\to(\Omega',\mu',E')$
	is a measure-preserving map $f : (\Omega,\mu)\to(\Omega',\mu')$
	such that $E\subseteq f^{-1}(E')$.
	For two finite conditional probability spaces
	$(\Omega,\mu,E)$ and $(\Omega',\mu',E')$,
	let $(\Omega,\mu,E)\fcpsto(\Omega',\mu',E')$
	be the set of maps from
	$(\Omega,\mu,E)$ to $(\Omega',\mu',E')$.

	\emph{Note:} For readability, finite conditional probability spaces
	$(\Omega,\mu,E)$ will be written $\Omega$ unless disambiguation is needed.
\end{definition}
With these two definitions in hand, we can give a precise description
to the semantic domains used to construct our low-level model of effectful \multippl{} terms.
Given a finite conditional probability space $\Omega$ as input,
an effectful \disc{} term $\ecomptydefault\eM\eA$
sends a pair of substitutions
for free \disc{} and \cont{} variables
to a distribution over weighted samples consisting of
a new finite conditional probability
space $\Omega'$ and a random variable $\Omega'\to\llbr\eA$ of outputs:
\[
	\implsemarg{\ecomptydefault{\eM}{\eA}}(\Omega)
	\hspace{0.25em}:\hspace{0.25em}
	\llbr\sGamma\times(\Omega\to\llbrarg\eDelta)
	\to
	\prmonad\left(\coprod_{\Omega'\in\fcps}(\Omega'\xrightarrow{\fcps}\Omega)\times (\Omega'\to\llbr{\eco A}) \right)
\]
The notation \(\coprod_{\Omega'\in\fcps}(\Omega'\xrightarrow{\fcps}\Omega)\times(\Omega'\to\llbr{\eA})\)
denotes an indexed coproduct:
an element of this set
is a tuple \((\Omega',f,X)\) consisting of a new finite conditional probability space \(\Omega'\),
a map of finite conditional probability spaces \(f : \Omega'\to\Omega\)
connecting the old and new sample spaces,
and a random variable \(X\) defined on the new sample space.

Analogously, an effectful \cont{} term
$\scomptydefault\scoe\stau$
sends
a pair of substitutions
to a distribution over weighted samples consisting of
a new finite conditional probability
space $\Omega'$ and a value $v\in\llbr\stau$:
\[
	\implsemarg{\scomptydefault{\scoe}{\stau}}(\Omega)
	\hspace{0.25em}:\hspace{0.25em}
	\llbr\sGamma\times(\Omega\to\llbrarg\eDelta)
	\to
	\prmonad\left(\coprod_{\Omega'\in\fcps} (\Omega'\xrightarrow{\fcps}\Omega)\times \llbr\stau \right)
\]
The semantic equations defining
$\implsemarg{\ecomptydefault{\eM}{\eA}}(\Omega)$
and
$\implsemarg{\scomptydefault{\scoe}{\stau}}(\Omega)$
are given in
\cref{fig:low-level-effectful}.
As in \cref{fig:high-level-effectful}, we continue to use
Haskell-style syntactic sugar for the
$\prmonad$ monad operations.
The interpretation of effectful \cont{} programs is
largely similar to the one given by $\erasesem-$;
the primary difference is the plumbing of probability spaces $\Omega$
and maps $f$ throughout.
The interpretation of effectful \disc{} programs
statefully manipulates the probability space
as sketched earlier:
$\eflip~\scoe$ expands the probability space from $\Omega$
to $\Omega\otimes\Omega_{\rm flip}$,
where $\Omega_{\rm flip}$ is a freshly-generated
probability space supporting a Bernoulli-distributed
random variable with parameter $\scoe$, and
$\eobserve~\eM$ shrinks the conditioning set
from $E$ to $E\cap F$, where $F$ is the subset of the sample space
on which $\eM$ is $\top$.
Maps of conditional probability spaces $f$ are
used to convert random variables from old to new sample spaces throughout.

The interpretation of the \cont{}-to-\disc{} boundary $\esample[\scoe]$ is to
draw a weighted sample $x$ from $\scoe$ and return the constant random variable
at $x$.  Conversely, the interpretation of the \disc{}-to-\cont{} boundary
$\sexact[\eM]$ is to compute the random variable $X$ denoted by $\eM$ and then
return a sample $x$ drawn from the distribution of $X$.  The parts of
\cref{fig:low-level-effectful} shown in bold ensure sound inference
interoperability: in the interpretation of $\sexact[\eM]$, the event $X^{-1}(x)$
is added to the conditioning set to ensure \emph{sample consistency}; in the
interpretation of $\eobserve~\eM$, the statement $\score(\mu|_E(F))$ performs
\emph{importance weighting}, to ensure the weight of the current execution
remains correct relative to other possible executions.\footnote{Here, $\mu|_E$ is the distribution $\mu$ conditioned on the event E.}

\subsection{Soundness} \label{sec:soundness}

This section presents our main theoretical result:
the low-level model $\implsem-$ capturing our inference strategy
soundly refines the high-level model $\erasesem-$; that is,
given a complete \multippl{} program $\scoe$,
weighted samples drawn from $\scoe$
according to our knowledge-compilation- and importance-sampling-based inference strategy
follow the same distribution
as samples drawn according to $\erasesem-$.
To make this precise, we first define what it means to run a complete \multippl{}
program, and what it means for two distributions over weighted samples to be equivalent.

\begin{definition}
	For a closed program $\scomptyclosed\scoe\stau$,
	let $\sevalprog(\scoe) $ be the computation
	\[\left(\begin{aligned}
				 & (\_,\_,x)\gets \implsemarg{\scoe}(\fcpsunit)(\emptyset,\emptyset); \\
				 & \gret~ x
			\end{aligned}\right) : \prmonad{\llbr\stau}\]
	where $\emptyset$ is the empty substitution,
	$\fcpsunit$ is the
	unique 1-point probability space.
	Let $\seraseevalprog(\scoe)$ be the computation
	$\erasesemarg{\scoe}(\emptyset,\emptyset) : \prmonad{\llbr\stau}$.
\end{definition}

\begin{definition}
	Two computations $\mu,\nu : \prmonad A$
	are \emph{equal as importance samplers},
	written $\mu\preq\nu$,
	if for all bounded integrable $k : A\to\R$
	it holds that $\Ex_{(a,x)\sim \mu}[a \cdot k(x)]
		=\Ex_{(b,y)\sim \nu}[b \cdot k(y)]$.
\end{definition}

With these definitions in hand, our soundness theorem states that
our inference strategy agrees with the high-level model
up to equality of importance samplers.

\begin{theorem}[soundness] \label{thm:toplevel-soundness}
	If $\scomptyclosed\scoe\stau$
	then
	$\sevalprog(\scoe) \preq \seraseevalprog(\scoe)$.
\end{theorem}

\cref{thm:toplevel-soundness} is proved by induction
on typing, after suitable strengthening of the theorem statement
from closed to open terms.
The essence of the proof boils down to two key lemmas.
The first lemma allows swapping the order of sampling and scoring,
and is crucial to the correctness of our importance reweighting scheme
in interpreting $\eobserve$:
\begin{lemma} \label{lem:score-interchange}
	If $(\Omega,\mu,E)\in\fcps$ then
	$
		\left(\begin{aligned}
				 & \omega\gets \mu;           \\
				 & \score(\Ind{\omega\in E}); \\
				 & \gret~\omega
			\end{aligned}\right)
		\preq
		\left(\begin{aligned}
				 & \score(\mu(E));     \\
				 & \omega\gets \mu|_E; \\
				 & \gret~\omega
			\end{aligned}\right)
	$.
\end{lemma}
The second lemma says that sampling twice --- %
from a marginal on $X$ to get a sample $x$, %
then from the conditional distribution
given $X=x$ --- is the same as sampling once from the joint distribution,
and is crucial to ensuring sample consistency in our implementation of the boundary
$\sexact[\eM]$.
\begin{lemma} \label{lem:cond-nested-integral}
	If $(\Omega,\mu,E)\in\fcps$ and $X : \Omega\to A$ with $A$ finite,
	then \[\left(\begin{aligned}
				 & x\gets \big(\omega\gets \mu; \gret(X\omega)\big); \\
				 & \omega'\gets \mu|_{X^{-1}(x)};                    \\
				 & \gret(x,\omega)
			\end{aligned}\right)
		=
		\left(\begin{aligned}
				 & \omega'\gets \mu;       \\
				 & \gret(X\omega',\omega')
			\end{aligned}\right).\]
\end{lemma}
The full details can be found in \appcref{app:soundness}.

\section{Evaluation}
\label{sec:eval}

In \cref{sec:multippl} we described the theoretical underpinnings of \host{} and
proved it sound. In this section we provide implementation details
and empirical evidence for the utility of \multippl{} by measuring its
scalability on well-known inference tasks and comparing its performance against
existing probabilistic programming systems. We conclude with a discussion of our
evaluation and how these programs relate to the design space of \multippl{} programs.

\subsection{Lightweight Extensions to \multippl{}}

The semantics described in \cref{sec:multippl} provide a minimal model of
multi-language interoperation that is simple and correct. In our
implementation we extend the semantics of \disc{} and \cont{} to support more
features, resulting in a practical and flexible language.

\subsubsection{Extensions to \cont{}}
Importance-sampling languages often include more features than those
described in \cont{}. The grammar for \cont{}, shown in \cref{fig:syntax}, supports three base distributions:
Bernoulli, Uniform, and Poisson distributions. In our implementation
we include many more distributions including Normal, Beta, and Dirichlet
distributions, as well as their corresponding observation expressions.
We also extend \cont{} with unbounded loops and list data structures.

\subsubsection{Extensions to \disc{}}
Our implementation of \disc{} directly leverages the BDD library of
\dice{}~\cite{holtzen2020scaling} and includes support for integers as described
in \citet{holtzen2020scaling}. Integers can be introduced into a \disc{} program
either by embedding a \cont{} integer or through new syntax in \disc{}
representing a discrete distribution. Both terms are translated into one-hot encoded
tuples of Boolean variables: \cont{} integers are translated
dynamically, while discrete categorical distributions are translated by the compiler
statically into the \disc{} grammar shown in \cref{fig:syntax}.

\subsection{Empirical evaluation}

\multippl{} programs encompass a vast design space, including both \cont{} and
\disc{} programs as well any interleaving of these two languages.  To
investigate the efficacy of our implementation and characterize this landscape,
we ask the following questions:

\begin{enumerate}[leftmargin=*]
	\item \emph{Does \host{} capture enough expressive power to represent
		      interesting and practical probabilistic structure while maintaining
		      competitive performance?} We consider four benchmarks with complex
	      conditional independence structures to illustrate the design space of
	      \multippl{} programs. We draw on models in the domains of
	      network analysis~\cite{knight2011Internet,gehr2018BayonetProbabilistic} and
	      Bayesian networks~\cite{beinlich1989ALARM, binder1997Adaptive}.

	\item \emph{How does \host{} compare with contemporary PPLs in using exact
		      and approximate inference with respect to wall-clock time and distance from
		      the exact distribution?} To answer this question, we benchmark against
	      state-of-the-art PPLs which handle both discrete and continuous variables:
	      PSI~\cite{gehr2016PSIExact}, performing exact inference by compilation, and
	      Pyro~\cite{bingham2019Pyro}, using its importance sampling infrastructure for approximate inference.

\end{enumerate}

\subsubsection{Experimental Setup}%

For exact inference, PSI is a best-in-class language that encodes both
discrete and continuous variables using its compiled symbolic approach. For
approximate inference we leverage Pyro's importance sampling
infrastructure. \multippl{} is written in Rust and performs both knowledge
compilation and sampling during runtime evaluation when it encounters a \disc{}
or \cont{} program, respectively. To unify the comparison between these
disparate settings, we delineate our
evaluation criteria along two metrics of sample \textit{efficiency} and sample
\textit{quality}.

The sample \textit{efficiency} of each inference strategy is defined as the
wall-clock time to draw 1000 samples; measured in seconds and recorded in
``Time(s)'' column of the following figures. Comparing the performance of inference algorithms implemented in
different languages is a general challenge. To account for the difference in overhead, we treat \cont{}
as our baseline in the approximate setting.

Sample \textit{quality} is also important and we computed the \emph{L1-distance}
(i.e., the difference of absolute values) between a ground-truth answer, derived
for each task, and the estimated quantity from sampling.
Tasks that only evaluate exact inference always yield an L1-distance of 0: for
these tasks we only report wall-clock time, and we only draw one sample from the
\multippl{} program.

Heuristically, our aim in writing \multippl{} programs is to achieve high
\textit{quality} samples using \disc{} while maintaining reasonable wall-clock
\textit{efficiency} with \cont{}. While this guides the design of our evaluation,
users must decide how this trade-off effects their models on a
case-by-case basis.

All benchmarks involving approximate inference are performed using a fixed
budget of 1000 samples and all statistics collected are averaged over 100
independent trials.\footnote{All evaluations are run on a single thread of an
	AMD EPYC 7543 Processor with 2.8GHz and $500$ GiB of RAM. A software artifact
	is available on Zenodo\cite{stites2025Artifact} and GitHub (\url{https://github.com/stites/multippl})}

\begin{figure}[t]
	\footnotesize
	\centering
	\begin{tabular}{c|ccccccccc}
		\toprule
		\multirow{2}{*}{Model} & \multicolumn{2}{c}{PSI} & \multicolumn{2}{c}{Pyro} & \multicolumn{2}{c}{\host{} (\cont{})} & \multicolumn{2}{c}{\host{}}                                                                                                 \\
		                       & L1                      & Time(s)                  & L1                                    & Time(s)                     & L1                    & Time(s)               & L1                    & Time(s)               \\
		\midrule
		arrival/tree-15        & ---                     & ---                      & 0.365                                 & 12.713                      & 0.355                 & \textbf{0.247       } & \textbf{0.337       } & 0.349                 \\
		arrival/tree-31        & ---                     & ---                      & 0.216                                 & 26.366                      & 0.218                 & \textbf{0.561       } & \textbf{0.179       } & 0.754                 \\
		arrival/tree-63        & ---                     & ---                      & 0.118                                 & 53.946                      & 0.120                 & \textbf{1.469       } & \textbf{0.093       } & 1.912                 \\
		\midrule
		alarm                  & t/o                     & t/o                      & 1.290                                 & 16.851                      & 1.173                 & \textbf{0.433       } & \textbf{0.364       } & 14.444                \\
		insurance              & t/o                     & t/o                      & 0.149                                 & 13.724                      & 0.144                 & \textbf{1.104       } & \textbf{0.099       } & 11.406                \\
		\midrule
		gossip/4               & ---                     & ---                      & \textbf{0.119       }                 & 6.734                       & \textbf{0.119       } & \textbf{0.720       } & 0.118                 & 0.812                 \\
		gossip/10              & ---                     & ---                      & 0.533                                 & 6.786                       & 0.531                 & 1.561                 & \textbf{0.524       } & \textbf{1.373       } \\
		gossip/20              & ---                     & ---                      & 0.747                                 & 7.064                       & \textbf{0.745       } & 3.565                 & 0.750                 & \textbf{2.888       } \\
		\midrule
	\end{tabular}
	\caption{Empirical results of our benchmarks of the arrival, discrete Bayesian
		network, and gossip tasks. ``\host{} (\cont{})'' shows the evaluation of a baseline \cont{} program
		with no boundary crossings into \disc{}, evaluations under the ``\host{}''
		column performs interoperation. ``t/o'' indicates a timeout beyond 30 minutes,
		and ``---'' indicates that the problem is not expressible in PSI because of an
		unbounded loop. %
	}\label{fig:eval:approx}
\end{figure}

\subsubsection{Estimating packet arrival.}\label{sec:eval:arrival}%
Our first evaluation comes from the motivating example of \cref{fig:motiv}.
For this arrival task we are interested in modeling packets traversing a router
network and observe the presence, or absence, of packets at their
destination. Our main interest is in some unobservable router that lives along
the traversal path, and we query the expected number of packets which pass
through this node.  The router network in our evaluation has a tree-based
topology that uses an equal-cost multipath (ECMP)
protocol~\cite{hopp2000Analysis}, as shown in
\cref{fig:eval:arrival:topology}. The ECMP protocol dictates that a packet is
forwarded with uniform probability to all neighboring
routers with equal distance to the goal, as shown in
\cref{fig:eval:arrival:topology}. In this scenario, $n$ packets traverse the
network where $n$ is drawn from
a Poisson distribution with a rate of 3, as described in
\cref{fig:eval:arrival:pseudocode}.
The presence of this Poisson random variable makes this example quite challenging
for many existing PPL inference strategies because the resulting loop
has a statically unbounded number of iterations. We made the following additional
design decisions in making this task:

\begin{figure}[t]
	\begin{subfigure}[b]{0.45\linewidth}
		\centering
		\begin{tikzpicture}[node distance=0.5cm, baseline=-1em]
			\node[draw, circle, style=double,minimum size=0.50cm] at (1.50,1.25) (n0) {};
			\node[draw, circle              ,minimum size=0.50cm] at (0.75,0.75) (n1l) {};
			\node[draw, circle              ,minimum size=0.50cm] at (2.25,0.75) (n1r) {};
			\node[draw, circle              ,minimum size=0.50cm] at (1.00,0.00) (n10l) {};
			\node[draw, circle              ,minimum size=0.50cm] at (0.00,0.00) (n11l) {};
			\node[draw, circle              ,minimum size=0.50cm] at (3.00,0.00) (n11r) {};
			\node[draw, circle, fill=gray   ,minimum size=0.50cm] at (2.00,0.00) (n10r) {};
			\node[draw, rectangle                      ] at (0.00,-0.75) (start) {{\footnotesize start}};
			\draw[->] (n1l)--(n0);
			\draw[->] (n0)--(n1r);
			\draw[->] (n1l)--(n10l);
			\draw[->] (n11l)--(n1l);
			\draw[->] (n1r)--(n10r);
			\draw[->] (n1r)--(n11r);
			\draw[->] (start)--(n11l);
		\end{tikzpicture}
		\subcaption{The arrival network topology.}
		\label{fig:eval:arrival:topology}
	\end{subfigure}
	\qquad
	\begin{subfigure}[b]{0.45\linewidth}
		\begin{algorithmic}
			\State $n \sim \mathrm{Poisson}(\lambda=3)$
			\State $q \gets 0$
			\While{$n > 0$}
			\State $q \gets q + \texttt{network}()$
			\State $n \gets n - 1$
			\EndWhile{}
			\State \Return $q$
		\end{algorithmic}
		\subcaption{Pseudocode describing arrival task.}\label{fig:eval:arrival:pseudocode}
	\end{subfigure}
	\caption{Implementation-generic details for the packet-arrival task. Shown in \ref{fig:eval:arrival:topology}, a packet traverses the network by entering the bottom-left most node, annotated by the start arrow. We observe a successful traversal to the gray-filled node, and we query the double-circle node for its posterior distribution. The PSI, Pyro, and \host{} programs all follow pseudocode shown in \ref{fig:eval:arrival:pseudocode}, where $\texttt{network}$ models the topology.}
	\label{fig:eval:arrival}
\end{figure}
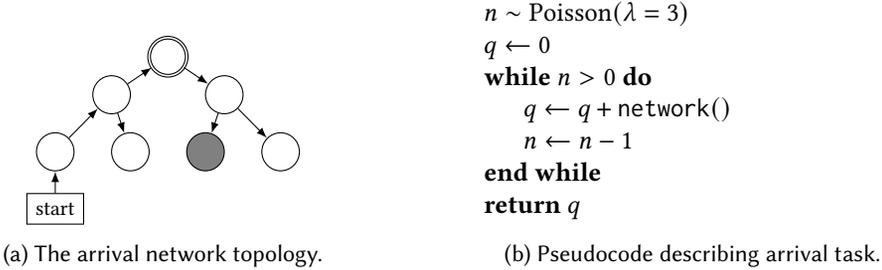
\begin{enumerate}[leftmargin=*]
	\item\textbf{Evidence:} We observe the gray node of the network topology depicted in \cref{fig:eval:arrival:topology}.

	\item\textbf{Query:} We query the expected probability that the packet traverses through a central node of the tree-topology, depicted by the twice-circled node of \cref{fig:eval:arrival:topology}.
	\item\textbf{Boundary decisions:} \cont{} models the Poisson distribution and outer loop. One boundary call is made to the network, defined in \disc{}.
	\item\textbf{Scaling:} We scale this model to topologies of 15, 31, and 63 nodes. %
	\item\textbf{Ground truth:} Our ground truth is defined by writing a \dice{} program for the network, and analytically solving for the expected number of packets.
\end{enumerate}

The rows labeled by ``arrival'' in \cref{fig:eval:approx} summarize the evaluation for this
section.
This table shows that \multippl{}'s samples are significantly higher quality
than Pyro's and the \cont{} program in this experiment. As the topology
increases in size, we see that the \multippl{} program is able to produce
increasingly higher quality samples with respect to L1 distance. This is because
\multippl{} is able to exactly compute packet reachability of a single
traversal in increasingly larger networks using \disc{}, while still able to
express sampling from the Poisson distribution in \cont{}. The \multippl{}
program performing interoperation is an order of magnitude more efficient than
Pyro, however the \cont{} alternative is still most efficient with regard to
wall-clock time and yields similar quality samples as Pyro.
PSI, using its symbolic inference procedure, fails to model the unbounded loop.

\subsubsection{Querying discrete Bayesian networks.}\label{sec:eval:bayesnet}

Bayesian networks~\cite{pearl1988Probabilistic} provide a challenging and
practical source of programs with widespread adoption across numerous domains,
including medicine~\cite{andreassen1991ModelBased,
	onisko2005Probabilistic}, healthcare~\cite{beinlich1989ALARM}, and actuarial sciences~\cite{binder1997Adaptive}.
Even in the purely-discrete setting, Bayesian networks remain a practical challenge when evaluating exact inference strategies due to the complex independence structures intrinsic to this domain.

In this task we study interoperation of our language by modeling two discrete
Bayesian networks: ALARM~\cite{beinlich1989ALARM} and
Insurance~\cite{binder1997Adaptive}. These networks pose a scaling challenge for
exact inference, and form the largest models described in our evaluation: ALARM
contains 509 flip primitives and the Insurance network contains 1008 flip
primitives. Modeling the entirety of the network in \disc{} and sampling 1000
times will result in a time-out for our evaluation, and we must use
interoperation to increase sample quality while keeping sample efficiency
competitive in our benchmark.

The ALARM network models a system for patient monitoring, while the Insurance network estimates the expected costs for a car insurance policyholder. We summarize these tasks as follows:
\begin{enumerate}[leftmargin=*]
	\item\textbf{Evidence:} In both models, we observed one or more leaf nodes. %
	\item\textbf{Query:} We query all root nodes for both of the Bayesian networks. %
	\item\textbf{Boundary decisions:} Variables are defined in \disc{} or \cont{} to heuristically maximize the degree of exact inference permitted while keeping the wall-clock time within 60 seconds.
	\item\textbf{Scaling:} ALARM contains 509 flip primitives and Insurance contains 1008 flip primitives.
	\item\textbf{Ground truth:} The ground truth is defined by an equivalent \dice{} program.
\end{enumerate}

The \cont{} model, with similar sample quality to Pyro, is more efficient than its
\multippl{} counterpart in this evaluation. It is also significantly
more efficient than Pyro and PSI (which timed out on this benchmark). As an importance
sampler, \cont{} and Pyro simply sample each distribution directly, and we see the
Python overhead slowing down the Pyro model.

Our \multippl{} programs demonstrate superior sample quality to the Pyro and \cont{} models. We achieve this by declaring
boundaries that split the ALARM and Insurance networks into sub-networks that are modeled
exactly with \disc{} and keep compiled BDD sizes small. However it should be noted that
placement of boundaries tips the scales in a tradeoff between quality and efficiency.
Optimal interleaving between \cont{} and \disc{} is task-sensitive,
and the \multippl{} programs evaluated only demonstrate a best-effort
approach to modeling.%

\subsubsection{Network takeover with a gossip protocol.}\label{sec:eval:gossip}

\begin{figure}[t]%
	\begin{subfigure}[b]{0.35\linewidth}
		\begin{subfigure}[t]{\linewidth}
			\centering
			\begin{tikzpicture}[node distance=0.5cm, baseline=-1em]
				\node[draw, circle,minimum size=0.5cm] at (0, 0.75) (top) {};
				\node[draw, circle,minimum size=0.5cm] at (0, -0.75) (bottom) {};
				\node[draw, circle,minimum size=0.5cm] at (0.75, 0) (right) {};
				\node[draw, circle,minimum size=0.5cm] at (-0.75, 0) (left) {};
				\draw (top) -- (bottom);
				\draw (top) -- (left);
				\draw (top) -- (right);
				\draw (bottom) -- (left);
				\draw (bottom) -- (right);
				\draw (left) -- (right);
				\node[circle] at (0, -1.5) (phantom) {};

			\end{tikzpicture}
		\end{subfigure}
		\subcaption{Topology of the gossip network.}
		\label{fig:eval:gossip:topology}
	\end{subfigure}%
	\begin{subfigure}[b]{0.63\linewidth}
		\begin{subfigure}[c]{0.40\linewidth}
			\footnotesize
			\begin{algorithmic}
				\State init $\leftarrow$ 0
				\State steps $\sim$ Uniform(4, 8)
				\State {state $\leftarrow$ \\\qquad(true, false, false, false)}
				\State deque $\leftarrow []$
				\For{$n = 1, 2$}
				\State s $\sim$ Discrete($\frac{1}{3}, \frac{1}{3}, \frac{1}{3}$);
				\State deque $\leftarrow$ s$+ 1 ::$ deque
				\EndFor{}
			\end{algorithmic}
		\end{subfigure}
		\begin{subfigure}[c]{0.55\linewidth}
			\footnotesize
			\begin{algorithmic}
				\While{steps $> 0$}
				\State cur $\leftarrow$ head(deque);
				\State deque $\leftarrow$ tail(deque);
				\State {state[cur] $\leftarrow$ state[cur] $\vee$ true}
				\For{$n = 1, 2$}
				\State s $\sim$ Discrete($\frac{1}{3}, \frac{1}{3}, \frac{1}{3}$);
				\State ix $\leftarrow$ if (s < cur) \{ s \} else \{ s + 1 \};
				\State deque $\leftarrow$ deque $++$ [ix]
				\EndFor
				\State steps  $\leftarrow$ steps - 1
				\EndWhile{}
				\State {\Return state}
			\end{algorithmic}
		\end{subfigure}

		\subcaption{Pseudocode for a gossip network task.}
		\label{fig:eval:gossip:pseudocode}
	\end{subfigure}
	\caption{
		Implementation-generic details of the gossip network task. The 4-node topology of the undirected network is shown in \ref{fig:eval:gossip:topology}. Pseudocode to iterate over each time step is provided in \ref{fig:eval:gossip:pseudocode}.  %
	}
	\label{fig:eval:gossip}
\end{figure}
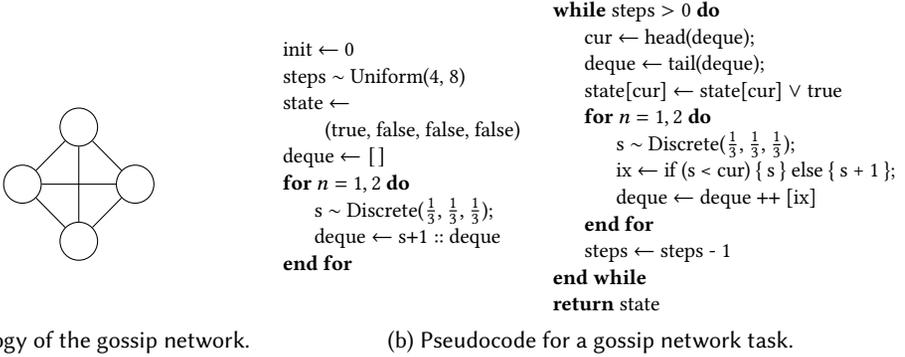

The gossip protocol is a common peer-to-peer communication method for
distributed systems.
In our setting each packet traverses an undirected, fully-connected network using a FIFO
scheduler for transport. %
At each time step, indicated by a tick in the scheduler, a server will schedule
two additional packets to all of its neighbors with each destination drawn
i.i.d. from a uniform distribution.
This task initializes with a \textit{compromised} node which sends two \textit{infected} packets
to its neighbors.
When a server receives an infected packet, it becomes
compromised and can only propagate infected packets for the remainder of the evaluation.
Taken from \citet{gehr2018BayonetProbabilistic}, we sample $n$ time steps from a uniform distribution, %
step $n$ times, then query the model for the expected number of compromised servers.%

This evaluation poses an expressivity challenge to the \disc{}
sub-language, which cannot define the dynamic-length FIFO queue without interoperation with \cont{}.
To handle this requirement, we extend \cont{} to support lists and
define all discrete random variables in \disc{}. At each end of the loop we
update our queue in \cont{}, collapsing any compiled BDDs. %

\begin{enumerate}[leftmargin=*]
	\item\textbf{Evidence:} This task defines a direct sampler and no evidence is given.
	\item\textbf{Query:} The model queries for the expected number of compromised servers after $n$ steps.
	\item\textbf{Boundary decisions:} Discrete variables are in \disc{}, the loop and FIFO queue live in \cont{}.
	\item\textbf{Scaling:} This network scales from 4- to 10- and 20- nodes. %
	\item\textbf{Ground truth:} The PSI model from \citet{gehr2018BayonetProbabilistic} was used to generate the ground truth for a statically defined set of time steps. An enumerative model was also defined to count the number of states. The expectation of these models over the loop is derived analytically.
\end{enumerate}

In \cref{fig:eval:approx}, we see that all terminating evaluations have
similar L1-distances, with Pyro and \cont{} programs producing slightly better quality
samples. The \multippl{} model produces more efficient samples, on average,
which speaks to the minimal overhead of interoperation when knowledge compilation
plays a small role in inference. There is also the possibility that BDDs are
cached and reused, resulting in a small speedup for some intermediate samples
drawn from \disc{}.

As this benchmark comes with a PSI implementation from
\citet{gehr2018BayonetProbabilistic}, we provided a best-effort attempt at
getting this to run including limiting the number of time steps to make the task
more tractable, but we were unable reproduce their results within our 30m
evaluation window.

\subsubsection{Estimating network reliability.}\label{sec:eval:reliability}
The \emph{network reliability} task is interested in a single packet's traversal through a network using a probabilistic routing protocol that is embedded in a larger network.
As a model only involving discrete random variables, we can observe how interoperation effects sample quality and efficiency by looking at programs defined in \cont{}, \disc{}, and in an optimal interoperation configuration.
Consider, again the ECMP protocol from \cref{sec:eval:arrival}. In this task we
modify each router with non-uniform probabilities, as a packet can traverse out
of the sub-network. The sub-network itself is a directed grid, shown in \cref{fig:eval:grid:topology}, with the
probability of traversal being dependent on the packet's origin. Pseudocode for the model is presented in \cref{fig:eval:grid:pseudocode}.

This benchmark observes a packet arriving at the final node in the
sub-network, and queries the probability that this packet passes through each
router in the model. As there are no continuous random variables involved, we
can model this task using either exact or approximate inference.

\begin{figure}[t]
	\begin{subfigure}[b]{0.45\linewidth}
		\centering
		\begin{tikzpicture}[node distance=0.75cm, baseline=-1em]
			\node[draw, circle                         ,minimum size=0.5cm] (x11) {};
			\node[draw, circle,            right of=x11,minimum size=0.5cm] (x12) {};
			\node[draw, circle,            right of=x12,minimum size=0.5cm] (x13) {};
			\node[draw, circle,            below of=x11,minimum size=0.5cm] (x21) {};
			\node[draw, circle,            right of=x21,minimum size=0.5cm] (x22) {};
			\node[draw, circle,            right of=x22,minimum size=0.5cm] (x23) {};
			\node[draw, circle,            below of=x21,minimum size=0.5cm] (x31) {};
			\node[draw, circle,            right of=x31,minimum size=0.5cm] (x32) {};
			\node[draw, circle, fill=gray, right of=x32,minimum size=0.5cm] (x33) {};
			\node[circle] at (0, -2.25) (phantom) {};

			\draw[->] (x11) -- (x12);
			\draw[->] (x12) -- (x13);
			\draw[->] (x11) -- (x21);
			\draw[->] (x12) -- (x22);
			\draw[->] (x13) -- (x23);
			\draw[->] (x21) -- (x22);
			\draw[->] (x22) -- (x23);
			\draw[->] (x21) -- (x31);
			\draw[->] (x22) -- (x32);
			\draw[->] (x23) -- (x33);
			\draw[->] (x31) -- (x32);
			\draw[->] (x32) -- (x33);
		\end{tikzpicture}
		\subcaption{Topology of a 9-node network}\label{fig:eval:grid:topology}
	\end{subfigure}
	\newcommand\nzz{\ensuremath{n_{00}}}
	\newcommand\nzo{\ensuremath{n_{01}}}
	\newcommand\nzt{\ensuremath{n_{02}}}
	\newcommand\noz{\ensuremath{n_{10}}}
	\newcommand\noo{\ensuremath{n_{11}}}
	\newcommand\ntz{\ensuremath{n_{20}}}
	\newcommand\nto{\ensuremath{n_{21}}}
	\newcommand\ntt{\ensuremath{n_{22}}}
	\begin{subfigure}[b]{0.45\linewidth}
		\footnotesize
		\begin{algorithmic}
			\State{\nzz $\sim$ Bern $\frac{1}{3}$}
			\State{\nzo $\sim$ if $\phantom{\neg}\nzz$ then Bern $\frac{1}{4}$ else Bern $\frac{1}{5}$} %
			\State{\noz $\sim$ if $\neg\nzz$           then Bern $\frac{1}{4}$ else Bern $\frac{1}{5}$}
			\State{$p \leftarrow \textrm{if} \phantom{\neg}\noz \vee \phantom{\neg}\nzo$ then $\frac{1}{6}$ else}
			\State{$\phantom{p \leftarrow \textrm{if} }\phantom{\neg}\noz \vee          \neg \nzo$ then $\frac{1}{7}$ else}
			\State{$\phantom{p \leftarrow \textrm{if} }\neg \noz \vee\phantom{\neg}\nzo$ then $\frac{1}{8}$ else $\frac{1}{9}$}
			\State{\noo $\sim$ Bern $p$}
			\State{observe \noo $= \top$}
			\State{\Return (\nzz, \nzo, \noz, \noo)}
		\end{algorithmic}
		\caption{Pseudocode of a 4-node model}\label{fig:eval:grid:pseudocode}
	\end{subfigure}
	\caption{
		An overview of the reliability task, with the topology of the 9-node network in \ref{fig:eval:grid:topology}: a packet is observed in the node shaded gray and all nodes are queried for their posterior distribution. In \ref{fig:eval:grid:pseudocode} we show the pseudocode for a 4-node reliability task, similar structure is used for networks with 9-, 36-, and 81-nodes.}\label{fig:eval:grid}
\end{figure}
\begin{figure}[t]
	\footnotesize
	\centering
	\begin{tabular}{c|cc|ccccccc}
		\toprule
		\multirow{2}{*}{\# Nodes} & PSI     & \host{} (\disc{})      & \multicolumn{2}{c}{Pyro} & \multicolumn{2}{c}{\host{} (\cont{})} & \multicolumn{2}{c}{\host{}}                                                           \\
		                          & Time(s) & Time(s)                & L1                       & Time(s)                               & L1                          & Time(s)               & L1                    & Time(s) \\
		\midrule
		9                         & 546.748 & \textbf{ 0.001       } & 0.080                    & 3.827                                 & 0.079                       & \textbf{0.067       } & \textbf{0.033       } & 0.098   \\
		36                        & t/o     & \textbf{ 0.089       } & 1.812                    & 14.952                                & 0.309                       & \textbf{0.277       } & \textbf{0.055       } & 1.169   \\
		81                        & t/o     & \textbf{40.728       } & 7.814                    & 33.199                                & 0.680                       & \textbf{0.887       } & \textbf{0.079       } & 81.300  \\
	\end{tabular}
	\caption{Exact and approximate results for models performing approximate inference%
	}\label{fig:eval:reliability}
\end{figure}

\begin{enumerate}[leftmargin=*]
	\item\textbf{Evidence:} The final node in the network topology observes a successful packet traversal.
	\item\textbf{Query:} The model queries for the marginal probability on all nodes in the network.
	\item\textbf{Boundary decisions:} \multippl{} programs (in column ``\host{}'' of table \cref{fig:eval:reliability}), model the minor upper- and lower- triangles of the network topology in \disc{} and perform interoperation along the minor diagonal to break the exact inference task into two parts. This maximizes the size of compiled BDDs while providing orders of magnitude improvement in sample efficiency.
	\item\textbf{Scaling:} This network scales in the size of the grid, scaling from 9- to 36- to 81- nodes.
	\item\textbf{Ground truth:} An equivalent \dice{} model was used as the ground truth for this model.
\end{enumerate}

The first two columns of \cref{fig:eval:reliability} show the results of exact
compilation; comparing PSI to \disc{} programs (column ``\host{} (\disc{})'').
Because of the nature of this evaluation, \disc{} can represent the exact
posterior of the model and produce perfect samples with competitive efficiency
for small programs. As the program grows in size, producing samples take
considerably longer: scaling with the size of the underlying
logical formula.

The partially-collapsed and fully-sampled \host{} programs are compared to Pyro in the remaining columns of
\cref{fig:eval:reliability}. \host{} programs (column
``\host{}'') are defined in \disc{} and model the minor diagonal of the
network's grid in \cont{}. Programs fully defined in \cont{} (column ``\host{} (\cont{})'')
sample each node individually in the same manner as Pyro.

In this evaluation \cont{} is more efficient and \host{} programs effectively
leverage \disc{}'s knowledge compilation to produce higher-quality samples. For
smaller models, the defined \host{} programs have efficiency competitive to
\cont{}. As the model scales, the overhead of knowledge compilation increases.
This can be seen by noting the single-sample efficiency of \disc{} programs from
our exact evaluation. As the \multippl{} program scales to 81 nodes the sample
efficiency decreases, suggesting an alternative collapsing
scheme may be preferable for larger programs.

This network reliability evaluation, alongside prior evaluations, demonstrates
that \multippl{} consistently produces higher quality samples compared to
alternatives in the approximate setting. Through these evaluations, we find that
\multippl{} does capture enough expressive power to represent interesting and
practical probabilistic structure while remaining competitive with other
languages. That said, the performance of \multippl{}'s inference poses a nuanced
landscape and we leave a full characterization of this design space to future
work.

\section{Related Work}
Multi-language interoperation between probabilistic programming languages builds
on a wide body of work spanning the programming language and the machine
learning communities. We situate our research in four categories: heterogeneous inference, programmable inference,
multi-language semantics, and the monadic semantics of probabilistic programming
languages.

\newcounter{relwork}
\setcounter{relwork}{1}
\newcommand{\relworksection}[1]{\vspace{0.5em}\noindent\textit{5.\therelwork{}\phantom{sp}\addtocounter{relwork}{1}~#1}}

\relworksection{Heterogeneous inference in probabilistic programming languages.}
There are existing probabilistic programming languages and systems that enable
users to blend different kinds of inference algorithms when performing inference
on a single probabilistic program. Particularly relevant are approaches that
leverage \emph{Rao-Blackwellization} in order to combine exact and approximate
inference strategies into a single system. Within this vein,
\citet{atkinson2022SemisymbolicInference} introduced \emph{semi-symbolic
	inference}, where the idea is to perform exact marginalization over
distributions whose posteriors can be determined to have some closed-form
solution. Other works that use variations of Rao-Blackwellization
\cite{murray2018DelayedSampling,gorinova2021conditional,obermeyer2019tensor} all seek to explicitly marginalize out
portions of the distribution by using closed-form exact posteriors when
feasible. The main difference between these approaches to Rao-Blackwellization
and our proposed approach is that these systems do not expose \emph{separate
	languages} that correspond to different phases of the inference algorithm: they
provide a single unified syntax in which the user programs. As a consequence,
they all rely on (semi-)automated means of automatically discovering which portions can
be feasibly Rao-Blackwellized; this process can be difficult to control and lead
to unpredictable performance. Our multi-language approach
has the following benefits: (1) predictable and interpretable performance due to the explicit
choice of inference algorithm that is exposed to the user; and (2) amenability to
modular formalization, since we can verify the correctness of each inference
strategy and verify the correctness of their composition on the boundary. %
We hope to incorporate the interesting ideas of these related works into \multippl{},
in particular closed-form approaches to exact marginalization of continuous distributions.

There is a broad literature on heterogeneous inference that we hope to
eventually draw on to build a richer vocabulary of sub-languages to add to
\multippl{}. \citet{friedman2018ApproximateKnowledge} described an approach to
collapsed approximate inference that dynamically blends exact inference via
knowledge compilation and approximate inference via sampling; we are curious if
this can be integrated with our system. We also look towards incorporating more
stateful inference algorithms such as Markov-Chain Monte Carlo into
\multippl{}, and aim to investigate this in future work.

\relworksection{Programmable inference.}
Programmable inference (or inference (meta-)programming)
provide probabilistic programmers with a meta-language for defining new inference algorithms within a single
PPL by offering language primitives that give direct access to the inference
runtime~\cite{mansinghka2014VentureHigherorder}. %
\citet{cusumano-towner2019GenGeneralpurpose} provides a black-box
interface to underlying inference algorithms alongside combinators to operate on
these interfaces, \citet{stites2021LearningProposals} designs a domain specific language (DSL) for inference which produces
correct-by-construction importance weights.

We see programmable inference as a viable means of designing new inference
algorithms which we can incorporate into a multi-language. Furthermore, a
multi-language setting can offer inference programmers the ability to abstract
away the nuances of the inference process, lowering the barrier to entry for
this type of development. One common thread through much of the work on
inference programming is core primitives which encapsulate the building blocks
for inference algorithms including resample-move sequential Monte Carlo,
variational inference, many other Markov chain Monte Carlo methods. These primitives
could be designed formally as DSLs, which would be a great addition to a
multi-language and something we look forward to developing in future work.

\relworksection{Nested inference.}%
~Nested inference enriches a probabilistic programming language with a
first-class \texttt{infer} or \texttt{normalize} construct that enables the
programmer to query for the probability of an event inside their probabilistic
programs~\citep{zhang2022reasoning,rainforth2018nesting,baker2009action,stuhlmuller2014reasoning,staton2017commutative}.
Nested inference is a useful programming construct that enables a variety of new
applications, such as in cognitive science where one agent may wish to reason
about the intent of another~\citep{stuhlmuller2014reasoning}.
Nested inference is similar in spirit to our multi-language approach in that it
gives the programmer control over when inference is performed
on their program and what inference algorithm is used.
A key difference between nested inference and our multi-language approach is
that the former provides access to the inference result
whereas MultiPPL's boundary forms do not. This difference is essential.
In our view, there is the following analogy to non-probabilistic programming:
performing nested inference is like invoking a compiler and inspecting the
resulting binary, whereas performing multi-language inference is like
interoperating through an FFI.
In the non-probabilistic setting, these two situations require distinct semantic models ---
compare, for example, formal models of introspection and dynamic code generation~\cite{smith1984reflection,malenfant1996semantics,benton2010step,davies2001modal,matthews2008operational}
with formal models of FFI-based interoperability~\cite{matthews2007operational,patterson2022semantic,sammler2023DimSum,gueneau2023melocoton,korkut2025verified} ---
and we believe the same is likely true of our probabilistic setting.

In the future, it would be interesting to consider integrating nested
inference within a multi-language setting and exploring the consequences
of this new feature on language interoperation.
It would also be quite interesting to investigate whether our multi-language inference
strategy could be compiled to, or expressed in terms of, rich nested inference constructs.
A preliminary analysis reveals a number of basic differences between \multippl{}'s
inference strategy and standard models of nested inference, so such a compilation scheme
would likely require significant modifications to nested inference --- for a detailed technical
discussion, see \cref{app:sec:comparison-with-nested-inference}.

\relworksection{Multi-language semantics.} Today, it is often the case that
probabilistic programming languages are embedded in a host, non-probabilistic
language~\cite{baydin2019EtalumisBringing}.
However, these PPLs assume their host semantics will not interfere with the
semantics of the PPL's inference process. This work is the first of its kind to
build on top of multi-language semantics to reason about inference algorithms.

Multi-language semantics, while new to the domain of probabilistic programming,
has had a large impact on the broader programming language community. They play
a fundamental role in reasoning about
interoperation~\citep{patterson2022semantic}, gradual
typing~\citep{new2020SemanticFoundation,tobin-hochstadt2008Design}, and
compositional compiler verification~\cite{sammler2023DimSum}. There are two
styles of calculi which represent the current approaches to multi-language
interoperation. These are the \emph{multi-languages} approach from
\citet{matthews2007operational} and the a more fine-grained approach by
\citet{siek2006GradualTyping} using a \emph{gradually typed lambda calculus}.

\citet{ye2023GradualProbabilistic} takes a traditional programming language
approach to the gradual typing of PPLs and defines a gradually typed
probabilistic lambda calculus which allows a user to migrate a PPL from an
untyped- to typed- language --- a nontrivial task involving a probabilistic
coupling argument for soundness. In contrast, our work centers on how
multi-languages can help the interoperation of inference algorithms across
semantic domains.

\citet{baydin2019EtalumisBringing} establishes, informally, a common interface
for PPLs to interact with scientific simulators across language boundaries. In
this work, the semantics of the simulator is a black-box distribution defined in
some language, which may or may not be probabilistic, and a separate PPL may
interact with the simulator during the inference process. While
\citet{baydin2019EtalumisBringing} works across language boundaries, they do not
reason about interoperation --- they only involve
one inference algorithm --- and they do not provide any soundness guarantees.
That said, \citet{baydin2019EtalumisBringing} demonstrates a simple boundary allowing for rapid
integration of many practical probabilistic programming languages, something we
also strive for.

\relworksection{Monadic semantics of PPLs.}
Numerous monads have been developed
for use as semantic domains
that capture the various notions of computation
used in probabilistic inference.
The fundamental building block for each of these models is the
probability monad $\giry$, along with its generalizations
to monads of subdistributions and measures~\cite{giry2006categorical}.
Using this probability monad to give semantics to probabilistic programs
goes back to at least
\citet{ramsey2002stochastic}, who further build on this basic setup
by introducing \emph{measure terms} to
efficiently answer expectation-based queries.
\citet{staton2016SemanticsProbabilistic}
make use of the writer monad transformer
applied to the monoid of weights
to obtain a monad suitable for modelling probabilistic programs
with score-based conditioning; we have made essential use of this
monad to define the two semantic models of \multippl{} presented in
\cref{sec:multippl}.
\citet{scibior2018FormallyJustified} use monad transformer stacks,
implemented in Haskell, to obtain a variety of
sampling-based inference algorithms in a compositional manner,
with each layer of the stack encompassing a different component of an inference
algorithm.
Our semantics of \multippl{}
builds on this line of
work in giving monadic semantics to probabilistic computations
by providing a model of exact inference via knowledge compilation
in terms of stateful manipulation of finite conditional
probability spaces and random variables.
In future work, we intend to investigate
whether this state-passing semantics can be packaged
into a monad of its own,
capturing the notion of
computation carried out when performing knowledge compilation,
by making use of recent constructions in
categorical probability~\cite{simpson2018category,simpson2017probability}.

\section{Conclusion}

Performing inference on models with a mix of continuous and discrete random variables is an important modeling challenge for
practical systems and \multippl{} offers a multi-language approach to tackle
this problem. In this work, we provide a sound denotational semantics that
generalizes for all exact inference algorithms and sampling-based
approximate inference that satisfy our semantic domains. We identify two
requirements to establish the correctness of the interoperation described: that
the exact PPL must maintain \emph{sample consistency} and that the approximate
sampling-based PPL must perform \emph{importance weighting}. We %
demonstrate that our implementation of \multippl{} benefits from the
expressiveness of \cont{} and %
makes practical problems representable and
additionally provides tractable inference from \disc{} for complex
discrete-structured probabilistic programs.

Ultimately, we hope that our multi-language perspective can lead to a clean
formal unification of many probabilistic program semantics and inference
strategies. For future work, we hope to extend our semantics to incorporate
local-search inference strategies such as sequential and Markov-Chain Monte
Carlo. With enough coverage across semantics, we also gain the opportunity to
look at probabilistic interoperation by inspecting a shared core calculus for
inference, and would draw on work
from \citet{patterson2022InteroperabilityRealizability}. Finally, by providing a
syntactic approach to inference interoperation, we also open up opportunities to
use static analysis to see when and how we might automatically insert boundaries
to further specialize a model's inference algorithm.

\section*{Acknowledgements}

This project was supported by the National Science Foundation under grant \#2220408.

\bibliographystyle{ACM-Reference-Format}
\bibliography{main}

\clearpage
\appendix
\section{Appendix}

\subsection{Syntax}

\begin{tabular}{rlrcl}
  \disc & Expressions     & $\eM,\eN$       & ::=   & $\eco X$ | $\etrue$ | $\efalse$ | $\eM\eand\eN$ | $\enot\eM$                                                          \\
        &                 &                 & |     & $\eunitelt$ | $\epair\eM\eN$ | $\efst~\eM$ | $\esnd~\eM$                                                              \\
        &                 &                 & |     & $\eret~\eM$ | $\eletin {\eco X}\eM\eN$ | $\eite\scoe\eM\eN$                                                           \\
        &                 &                 & |     & $\eflip~\scoe$ | $\eobserve~\eM$ | $\esample[\scoe]$                                                                  \\
        & Types           & $\eA,\eB$       & ::=   & $\eunit$ | $\ebool$ | $\eA\etimes\eB$                                                                                 \\
        & Contexts        & $\eDelta$       & ::=   & $\eco{X_{1}} : \eco{A_{1}},\dots,\eco{X_n} : \eco{A_n}$                                                               \\
  \\
  \cont & Expressions     & $\scoe$         & ::=   & $\sco x$ | $\strue$ | $\sfalse$ | $r$ | $\seone\sadd\setwo$ | $\sneg\scoe$ | $\seone\smul\setwo$ | $\seone\sle\setwo$ \\
        &                 &                 & |     & $\sunitelt$ | $\spair\seone\setwo$ | $\sfst~\scoe$ | $\ssnd~\scoe$                                                    \\
        &                 &                 & |     & $\sret~\scoe$ | $\sletin {\sco x}{\seone}{\setwo}$ | $\site{\seone}{\setwo}{\sethr}$                                  \\
        &                 &                 & |     & $\sco{d}$ | $\sco{\texttt{obs}(\scoe_{o}, d)}$  | $\sexact[\eco M]$                                               \\
        & Distributions   & $\sco{d}$       & ::=   & $\sflip~\scoe$ | $\sunif~\seone~\setwo$ | $\spois~\scoe$                                                             \\
        & Types           & $\ssigma,\stau$ & ::=   & $\sunit$ | $\sbool$ | $\sreal$ | $\ssigma\stimes\stau$                                                                \\
        & Contexts        & $\sGamma$       & ::=   & $\sco{x_1} : \sco{\tau_1},\dots,\sco{x_n} : \sco{\tau_n}$                                                             \\
        & Number literals & $r$             & $\in$ & $\R$                                                                                                                  \\
\end{tabular}

\subsection{Typing rules}

\subsubsection{Convertibility}

\begin{mathpar}
  \inferrule{~}{\eunit \convrel \sunit}
  \and
  \inferrule{~}{\ebool \convrel \sbool}
  \and
  \inferrule{\eA\convrel \stau \\ \eB\convrel\ssigma}
  {\eA\etimes\eB\convrel \stau\stimes\ssigma}
\end{mathpar}

\subsubsection{Pure exact sublanguage}

\begin{mathpar}
  \inferrule{\eDelta(\eco X) = \eA}
  {\epuretydefault{\eco X}{\eA}}
  \and
  \inferrule {~} {\epuretydefault{\etrue}{\ebool}}
  \and
  \inferrule {~} {\epuretydefault{\efalse}{\ebool}}
  \and
  \inferrule {\epuretydefault\eM\ebool \\ \epuretydefault\eN\ebool} {\epuretydefault{\eM\eand\eN}{\ebool}}
  \and
  \inferrule {\epuretydefault\eM\ebool} {\epuretydefault{\enot\eM}{\ebool}}
  \and
  \inferrule {~} {\epuretydefault{\eunitelt}{\eunit}}
  \and
  \inferrule
  {
    \epuretydefault{\eM}{\eA}
    \\
    \epuretydefault{\eN}{\eB}}
  {\epuretydefault{\epair\eM\eN}{\eA\etimes\eB}}
  \and
  \inferrule
  { \epuretydefault{\eM}{\eA\times\eB} }
  {\epuretydefault{\efst~\eM}{\eA}}
  \and
  \inferrule
  { \epuretydefault{\eM}{\eA\times\eB} }
  {\epuretydefault{\esnd~\eM}{\eB}}
\end{mathpar}

\subsubsection{Effectful exact sublanguage}

\begin{mathpar}
  \inferrule{\epuretydefault\eM\eA}{\ecomptydefault{\eret~\eM}{\eA}}
  \and
  \inferrule{
    \ecomptydefault\eM\eA
    \\
    \ecompty\sGamma{\eDelta,\eco X\ofty \eA}\eM\eB}
  {\ecomptydefault{\eletin{\eco X}\eM\eN}{\eB}}
  \and
  \inferrule {
    \spuretydefault\scoe\sbool
    \\
    \ecomptydefault\eM\eA
    \\
    \ecomptydefault\eN\eA
  } {\ecomptydefault{\eite\scoe\eM\eN}{\eA}}
  \and
  \inferrule {\spuretydefault\scoe\sreal}{\ecomptydefault{\eflip~\scoe}{\ebool}}
  \and
  \inferrule {\epuretydefault\eM\ebool} {\ecomptydefault{\eobserve~\eM}{\eunit}}
  \and
  \inferrule {\scomptydefault{\scoe}{\stau}
  \\ \eA\convrel\stau} {\ecomptydefault{\esample[\scoe]}{\eA}}
\end{mathpar}

\subsubsection{Pure sampling sublanguage}

\begin{mathpar}
  \inferrule{\sGamma(\sco x) = \stau}
  {\spuretydefault{\sco x}{\stau}}
  \and
  \inferrule {~} {\spuretydefault{\strue}{\sbool}}
  \and
  \inferrule {~} {\spuretydefault{\sfalse}{\sbool}}
  \and
  \inferrule {~} {\spuretydefault{r}{\sreal}}
  \and
  \inferrule {\spuretydefault{\seone}\sreal
  \\ \spuretydefault\setwo\sreal}
  {\spuretydefault{\seone\sadd\setwo}{\sreal}}
  \and
  \inferrule {\spuretydefault{\scoe}\sreal}
  {\spuretydefault{\sneg\scoe}{\sreal}}
  \and
  \inferrule {\spuretydefault{\seone}\sreal
  \\ \spuretydefault\setwo\sreal}
  {\spuretydefault{\seone\smul\setwo}{\sreal}}
  \and
  \inferrule {\spuretydefault{\seone}\sreal
  \\ \spuretydefault\setwo\sreal}
  {\spuretydefault{\seone\sle\setwo}{\sbool}}
  \and
  \inferrule {~} {\spuretydefault\sunitelt{\sunit}}
  \and
  \inferrule
  {
    \spuretydefault{{\seone}}{\ssigma}
    \\
    \spuretydefault{\setwo}{\stau}}
  {\spuretydefault{\spair\seone\setwo}{\ssigma\stimes\stau}}
  \and
  \inferrule
  {
    \spuretydefault{{\scoe}}{\ssigma\stimes\stau}
    }
  {\spuretydefault{\sfst~{\scoe}}{\stau}}
  \and
  \inferrule
  { \spuretydefault{{\scoe}}{\ssigma\stimes\stau} }
  {\spuretydefault{\ssnd~{\scoe}}{\stau}}
\end{mathpar}

\subsubsection{Effectful sampling sublanguage}

\begin{mathpar}
  \inferrule{\spuretydefault{\scoe}\stau}{\scomptydefault{\sret~{\scoe}}{\stau}}
  \and
  \inferrule{
    \scomptydefault{\seone}\ssigma
    \\
    \scompty{\sGamma,\sco x\ofty \ssigma}{\eDelta}{\setwo}\stau}
  {\scomptydefault{\sletin {\sco x}{\seone}{\setwo}}{\stau}}
  \and
  \inferrule {
    \spuretydefault{\seone}\sbool
    \\
    \scomptydefault{\setwo}\stau
    \\
    \scomptydefault{\sethr}\stau
  } {\scomptydefault{\site\seone\setwo\sethr}{\stau}}
  \and
  \inferrule {\spuretydefault\scoe\sreal}{\scomptydefault{\sflip~\scoe}{\sbool}}
  \and
  \inferrule {
    \spuretydefault{\seone}{\sreal}\\
    \spuretydefault{\setwo}{\sreal}
    }
  {\scomptydefault{\sunif~\seone~\setwo}{\sreal}}
  \and
  \inferrule {\spuretydefault\scoe\sreal} {\scomptydefault{\spois~\scoe}{\sreal}}
  \and
  \inferrule {\ecomptydefault\eM{\eA}
    \\ \eA\convrel\stau} {\scomptydefault{\sexact[\eM]}{\stau}}

  \and

  \inferrule {\spuretydefault{\seobs}\sbool  \\ \spuretydefault\seone\sreal}
             {\scomptydefault{\sobserve{\seobs}{\sflip~\seone}}{\sunit}}

  \and

  \inferrule {
    \spuretydefault{\seobs}{\sreal}\\
    \spuretydefault{\seone}{\sreal}\\
    \spuretydefault{\setwo}{\sreal}
    }
  {\scomptydefault{\sobserve{\seobs}{\sunif~\seone~\setwo}}{\sunit}}

  \and

  \inferrule
    {\spuretydefault{\seobs}{\sreal}\\
    \spuretydefault\scoe\sreal}
    {\scomptydefault{\sobserve{\seobs}{\spois~\scoe}}{\sunit}}
\end{mathpar}

\subsection{Semantics of types}

\subsubsection{Types}

\begin{align*}
  \llbrspc\eA &: \text{finite discrete measurable space} \\
  \llbrspc\eunit &= \text{the one-point space $\{\star\}$} \\
  \llbrspc\ebool &= \{\top,\bot\} \\
  \llbrspc{\eA\etimes\eB} &= \llbrspc\eA\times\llbrspc\eB
\end{align*}

\begin{align*}
  \llbr\stau &: \text{measurable space} \\
  \llbr\sunit &= \text{the one-point space $\{\star\}$} \\
  \llbr\sbool &= \text{the discrete two-point space $\{\top,\bot\}$} \\
  \llbr\sreal &= \R \\
  \llbr{\ssigma\stimes\stau} &= \llbr\ssigma\times\llbr\stau
\end{align*}

\subsubsection{Contexts}

\begin{align*}
  \llbr{\sGamma} &: \text{measurable space} \\
  \llbr{\sGamma} &= \prod_{\sco{x}\in\dom\sGamma}\llbr{\sGamma(\sco x)}
\end{align*}

\begin{align*}
  &\llbr{\eDelta} : \text{finite discrete measurable space} \\
  &\llbr{\eDelta} = \prod_{\eco{X}\in\dom\eDelta}\llbrspc{\eDelta(\eco X)} \\
\end{align*}

\subsubsection{Convertibility}

\begin{lemma}
  If $\eA\convrel\stau$ then $\llbr\eA = \llbr\stau$.
\end{lemma}
\begin{proof} By induction on $\eA\convrel\stau$.  \end{proof}

\subsection{Semantics of pure programs}

\begin{align*}
  \puresemarg{\epuretydefault{\eM}{\eA}}&: \llbrspc\eDelta\to\llbrspc\eA\\
  \puresemarg{\epuretydefaultsem{\eco X}{\eA}}(\delta) &= \delta(\eco X) \\
  \puresemarg{\epuretydefaultsem{\etrue}{\ebool}}(\delta) &= \top \\
  \puresemarg{\epuretydefaultsem{\efalse}{\ebool}}(\delta) &= \bot \\
  \puresemarg{\epuretydefaultsem{\eM\eand\eN}{\ebool}}(\delta) &=
    \begin{cases}
      \top, &
        \text{if }
        \puresemarg{\epuretydefaultsem{\eM}{\ebool}}(\delta)
        =
        \puresemarg{\epuretydefaultsem{\eN}{\ebool}}(\delta)
        = \top
        \\
      \bot, & \text{otherwise}
    \end{cases}
    \\
  \puresemarg{\epuretydefaultsem{\enot\eM}{\ebool}}(\delta) &=
    \begin{cases}
      \bot, &
        \text{if }
        \puresemarg{\epuretydefaultsem{\eM}{\ebool}}(\delta)
        = \top
        \\
      \top, & \text{otherwise}
    \end{cases}
    \\
  \puresemarg{\epuretydefaultsem{\eunitelt}{\eunit}}(\delta) &= \star \\
  \puresemarg{\epuretydefaultsem{\epair\eM\eN}{\eA\etimes\eB}}(\delta) &=
     (\puresemarg{\epuretydefaultsem\eM\eA}(\delta),
     \puresemarg{\epuretydefaultsem\eN\eB}(\delta)) \\
  \puresemarg{\epuretydefaultsem{\efst~\eM}{\eA}}(\delta) &=
     \pi_1(\puresemarg{\epuretydefaultsem{\eM}{\eA\etimes\eB}}(\delta))
    \\
  \puresemarg{\epuretydefaultsem{\esnd~\eM}{\eB}}(\delta) &=
     \pi_2(\puresemarg{\epuretydefaultsem{\eM}{\eA\etimes\eB}}(\delta))
\end{align*}

\begin{align*}
  \puresemarg{\spuretydefault{\scoe}{\stau}}&: \llbr\sGamma\mbleto{\llbr\stau}\\
  \puresemarg{\spuretydefaultsem{\sco x}{\stau}}(\gamma) &= \gamma(\sco x) \\
  \puresemarg{\spuretydefaultsem{\strue}{\sbool}}(\gamma) &= \top \\
  \puresemarg{\spuretydefaultsem{\sfalse}{\sbool}}(\gamma) &= \bot \\
  \puresemarg{\spuretydefaultsem{r}{\sreal}}(\gamma) &= r\\
  \puresemarg{\spuretydefaultsem{\seone\sadd\setwo}{\sreal}}(\gamma) &=
    \puresemarg{\spuretydefaultsem{\seone}\sreal}(\gamma)+
    \puresemarg{\spuretydefaultsem\setwo\sreal}(\gamma) \\
  \puresemarg{\spuretydefaultsem{\sneg\scoe}{\sreal}}(\gamma) &=
    -\puresemarg{\spuretydefaultsem{\scoe}\sreal}(\gamma)\\
  \puresemarg{\spuretydefaultsem{\seone\smul\setwo}{\sreal}}(\gamma) &=
    \puresemarg{\spuretydefaultsem{\seone}\sreal}(\gamma)\cdot
    \puresemarg{\spuretydefaultsem\setwo\sreal}(\gamma) \\
  \puresemarg{\spuretydefaultsem{\seone\sle\setwo}{\sbool}}(\gamma) &=
  \begin{cases}
    \top, & \text{if }\puresemarg{\spuretydefaultsem{\seone}\sreal}(\gamma)
             \le \puresemarg{\spuretydefaultsem\setwo\sreal}(\gamma) \\
    \bot, & \text{otherwise}
  \end{cases}\\
  \puresemarg{\spuretydefaultsem{\sunitelt}{\sunit}}(\gamma) &= \star \\
  \puresemarg{\spuretydefaultsem{\spair\seone\setwo}{\ssigma\stimes\stau}}(\gamma) &=
    (\puresemarg{\spuretydefaultsem{\seone}{\ssigma}}(\gamma),
    \puresemarg{\spuretydefaultsem{\setwo}{\stau}}(\gamma))
    \\
  \puresemarg{\spuretydefaultsem{\sfst~\scoe}{\ssigma}}(\gamma) &=
    \pi_1(\puresemarg{\spuretydefaultsem{\scoe}{\ssigma\stimes\stau}}(\gamma))
    \\
  \puresemarg{\spuretydefaultsem{\ssnd~\scoe}{\stau}}(\gamma) &=
    \pi_2(\puresemarg{\spuretydefaultsem{\scoe}{\ssigma\stimes\stau}}(\gamma))
    \\
\end{align*}

\subsection{Effectful programs: high-level model}

\begin{definition}
Let $\giry$ be the distribution monad defined over measurable spaces.
Let $\prmonad$ be the writer monad transformer applied
to $\giry$ and the monoid $(\weightmonoid,1)$ of nonegative reals under multiplication.
Concretely, $\prmonad$ sends a measurable space $A$
to the set $\giry(\weightmonoid\times A)$.
Let $\gret~x$
and $\gletin x\mu{f(x)}$ denote the usual monad operations with respect to $\prmonad$,
and let $\overline{(\cdot)}:\giry(A)\to\prmonad(A)$ be the usual lifting operation.
In addition to the usual operations,
\begin{itemize}[leftmargin=*]
  \item
    Let $\score : \weightmonoid\to\prmonad\{\star\}$ be the
    function that sends a weight $w$ to the Dirac
    distribution $\delta_{(w, \star)}$ centered at $(w,\star)$.
\end{itemize}
\end{definition}

\begin{definition}
  For $p\in \R$,
  let $\mathrm{flip}(p)$
  be the Bernoulli distribution with parameter $p$
  if $p\in [0,1]$,
  and a point mass at $\bot$ otherwise.
\end{definition}

\begin{definition}
  For $a,b\in \R$,
  let $\unif(a,b)$
  be the uniform distribution on the interval $[a,b]$
  if $a\le b$,
  and a point mass at $\min(a,b)$ otherwise.
\end{definition}

\begin{definition}
  For $\lambda\in \R$,
  let $\pois(\lambda)$
  be the Poisson distribution with parameter $\lambda$
  if $\lambda > 0$,
  and a point mass at $0$ otherwise.
\end{definition}

\subsubsection{Exact sublanguage}

\begin{align*}
  \erasesem{\ecomptydefault{\eM}{\eA}}
    &:\llbr\sGamma\times\llbr\eDelta
    \to \prmonad{\llbrspc\eA}\\
  \erasesemarg{\ecomptydefaultsem{\eret~\eM}{\eA}}(\gamma,\delta) &=
    \gret(\puresemarg{\epuretydefaultsem{\eM}{\eA}}(\delta))
    \\
  \erasesemarg{\ecomptydefaultsem{\eletin{\eco X}\eM\eN}{\eB}}(\gamma,\delta) &=
    \left(\begin{aligned}
      &x \gets \erasesemarg{\ecomptydefaultsem\eM\eA}(\gamma,\delta);\\
      &\erasesemarg{\ecomptysem\sGamma{\eDelta,\eco X\ofty \eA}\eN\eB}(\gamma,\delta[\eco X\mapsto x])
    \end{aligned}\right)
    \\
  \erasesemarg{\ecomptydefaultsem{\eite\scoe\eM\eN}{\eA}}(\gamma,\delta) &=
    \left(\begin{aligned}
      &\mathrm{if}~\puresemarg{\spuretydefaultsem\scoe\sbool}(\gamma)\\
      &\mathrm{then}~\erasesemarg{\ecomptydefaultsem\eM\eA}(\gamma,\delta)\\
      &\mathrm{else}~\erasesemarg{\ecomptydefaultsem\eN\eA}(\gamma,\delta)
    \end{aligned}\right)\\
  \erasesemarg{\ecomptydefaultsem{\eflip~\scoe}{\ebool}}(\gamma,\delta) &=
    \overline{\flip(\puresemarg{\spuretydefaultsem\scoe\sreal}(\gamma))} \\
  \erasesemarg{\ecomptydefaultsem{\eobserve~\eM}{\eunit}}(\gamma,\delta) &=
      \score(\Ind{\puresemarg{\epuretydefaultsem\eM\ebool}(\delta) = \top})
    \\
  \erasesemarg{\ecomptydefaultsem{\esample[\scoe]}{\eA}}(\gamma,\delta) &=
    \valstoe\eA\stau{\erasesemarg{\scomptydefaultsem{\scoe}{\stau}}(\gamma,\delta)}
\end{align*}

\subsubsection{Sampling sublanguage}

\begin{align*}
  \erasesem{\scomptydefault{\scoe}{\stau}} &:
    \llbr\sGamma\times\llbr\eDelta
    \to \prmonad{\llbr\stau}
    \\
  \erasesemarg{\scomptydefaultsem{\sret~{\scoe}}{\stau}}(\gamma,\delta) &=
    \gret(\puresemarg{\spuretydefaultsem{\scoe}\stau}(\gamma))
    \\
  \erasesemarg{\scomptydefaultsem{\sletin {\sco x}{\seone}{\setwo}}{\stau}}(\gamma,\delta) &=
    \left(\begin{aligned}
      &x\gets \erasesemarg{\scomptydefaultsem{\seone}\ssigma}(\gamma,\delta);\\
      &\erasesemarg{\scomptysem{\sGamma,\sco x\ofty \ssigma}\eDelta\setwo\stau}(\gamma[\sco x\mapsto x],\delta)\\
    \end{aligned}\right)
    \\
  \erasesemarg{\scomptydefaultsem{\site\seone\setwo\sethr}{\stau}}(\gamma,\delta) &=
    \left(\begin{aligned}
      &\mathrm{if}~{\puresemarg{\spuretydefaultsem{\seone}\sbool}(\gamma)}\\
      &\mathrm{then}~{\erasesemarg{\scomptydefaultsem{\setwo}\stau}(\gamma,\delta)}\\
      &\mathrm{else}~{\erasesemarg{\scomptydefaultsem{\sethr}\stau}(\gamma,\delta)}
    \end{aligned}\right)
    \\
  \erasesemarg{\scomptydefaultsem{\sflip~\scoe}{\sreal}}(\gamma,\delta) &=
    \overline{\flip(\puresemarg{\spuretydefaultsem\scoe\sreal}(\gamma))}\\
  \erasesemarg{\scomptydefaultsem{\sunif~\seone~\setwo}{\sreal}}(\gamma,\delta) &= \overline{\unif(
    \puresemarg{\spuretydefaultsem\seone\sreal}(\gamma),
        \puresemarg{\spuretydefaultsem\setwo\sreal}(\gamma))}\\
  \erasesemarg{\scomptydefaultsem{\spois~\scoe}{\sreal}}(\gamma,\delta) &=
    \overline{\pois(\puresemarg{\spuretydefaultsem\scoe\sreal}(\gamma))}\\
  \erasesemarg{\scomptydefaultsem{\sexact[\eM]}{\stau}}(\gamma,\delta) &=
      \valetos\eA\stau{\erasesemarg{\ecomptydefaultsem\eM{\eA}}(\gamma,\delta)}\\
  \erasesemarg{\scomptydefaultsem{\sobserve{\seobs}{\sflip~\seone}}{\sreal}}(\gamma,\delta) &=
      \score\left(
          \flip(\puresemarg{\spuretydefaultsem\seone\sreal}(\gamma))
          (\puresemarg{\spuretydefaultsem{\seobs}\sbool}(\gamma))
      \right) \\
  \erasesemarg{\scomptydefaultsem{\sobserve{\seobs}{\spois~\seone}}{\sreal}}(\gamma,\delta) &=
      \score\left(
          \pois(\puresemarg{\spuretydefaultsem\scoe\sreal}(\gamma))
          (\puresemarg{\spuretydefaultsem{\seobs}\sreal}(\gamma))
      \right) \\
  \erasesemarg{\scomptydefaultsem{\sobserve{\seobs}{\sunif~\seone~\setwo}}{\sreal}}(\gamma,\delta) &=
      \score\left(
          \unif(
            \puresemarg{\spuretydefaultsem\seone\sreal}(\gamma),
                \puresemarg{\spuretydefaultsem\setwo\sreal}(\gamma))
          (\puresemarg{\spuretydefaultsem{\seobs}\sbool}(\gamma))
      \right) \\
\end{align*}

\subsection{Effectful programs: low-level model}

\begin{definition}
  A \emph{finite conditional probability space}
  is a triple $(\Omega,\mu,E)$
  where
    (1) $\Omega$ is a nonempty finite prefix of $\N$,
    (2) $\mu : \Omega\to[0,1]$ is a discrete probability distribution,
    and (3) $E$ is a subset of $\Omega$.
  Let $\fcps$ be the set of finite conditional probability spaces.
\end{definition}
For readability, finite conditional probability spaces
$(\Omega,\mu,E)$ will be written $\Omega$ unless disambiguation is needed.

\begin{lemma} \label{app:lem:fcps-mblespc}
  $\fcps$ is a measurable space.
\end{lemma}
\begin{proof}
  There is an injective function
  $i : \fcps \hookrightarrow \N \times \mathrm{list}(\R) \times \mathcal P_{\rm fin}\N$
  sending a finite conditional probability space
  $(\{0,\dots,n-1\},\mu,E)$
  to the triple
  $(n, [\mu(0),\dots,\mu(n-1)], E)$.
  The codomain of $i$ is a measurable space with $\sigma$-algebra defined in the
  standard way, and
  the image of $i$ is a measurable subset of this space.
  The injection $i$ identifies the image of $i$ with $\fcps$,
  making $\fcps$ a measurable space by taking preimages along $i$.
\end{proof}

\begin{definition}
  A \emph{map of finite conditional probability spaces}
  $f : (\Omega,\mu,E)\to(\Omega',\mu',E')$
  is a measure-preserving map $f : (\Omega,\mu)\to(\Omega',\mu')$
  such that $E\subseteq f^{-1}(E')$.
  For two finite conditional probability spaces
  $(\Omega,\mu,E)$ and $(\Omega',\mu',E')$,
  let $(\Omega,\mu,E)\xrightarrow{\fcps}(\Omega',\mu',E')$
  be the set of maps from
  $(\Omega,\mu,E)$ to $(\Omega',\mu',E')$.
\end{definition}

\begin{definition}
  For every $\eco A$ and $\sco \tau$
  let $\esemcomp A$ and $\ssemcomp\tau$ be the following
  $\fcps$-indexed families of sets:
\begin{align*}
  \esemcomp A &: \fcps\to\Set \\
  \esemcomp A\Omega &= \prmonad\left(\sum_{\Omega'\in\fcps} (\Omega'\xrightarrow{\fcps}\Omega)\times (\Omega'\to\llbr{\eco A}) \right)
\end{align*}
\begin{align*}
  \ssemcomp \tau &: \fcps\to\Set \\
  \ssemcomp \tau \Omega &= \prmonad\left(\sum_{\Omega'\in\fcps} (\Omega'\xrightarrow\fcps\Omega)\times \llbr{\sco \tau}\right)
\end{align*}
\end{definition}
\begin{proof}
  For these to be well-defined, the arguments to $\prmonad$
  must be measurable spaces.
  Elements of
  $\sum_{\Omega'\in\fcps} (\Omega'\xrightarrow\fcps\Omega)\times (\Omega'\to\llbr{\eco A})$
  are triples $(\Omega', f, X)$ where $\Omega'\in\fcps$ and $f$ is a map of
  finite conditional probability spaces and $X : \Omega'\to\llbr{\eco A}$.
  There is a canonical injective function $i$ from this set to the set
    $\fcps \times (\N\xrightarrow{\rm fin}\N) \times (\N\xrightarrow{\rm fin} \llbr{\eco A})$
  whose elements are triples $(\Omega', f, X)$
  where $\Omega'\in\fcps$ and $f,X$ are partial functions of type $\N\rightharpoonup\N$
  and $\N\rightharpoonup\llbr{\eco A}$ respectively with finite domain.
  This set is a measurable space by \cref{app:lem:fcps-mblespc}
  and by putting the discrete $\sigma$-algebras on
    $\N\xrightarrow{\rm fin}\N$ and $\N\xrightarrow{\rm fin} \llbr{\eco A}$.
  The image of $i$ is a measurable subset of this space,
  making
  $\sum_{\Omega'\in\fcps} (\Omega'\xrightarrow\fcps\Omega)\times (\Omega'\to\llbr{\eco A})$
  a measurable space too.
  The analogous argument also makes
  $\sum_{\Omega'\in\fcps} (\Omega'\xrightarrow\fcps\Omega)\times \llbr{\stau}$
  into a measurable space.
\end{proof}

\begin{definition}
  For two finite conditional probability spaces
  $(\Omega,\mu,E)$
  and
  $(\Omega',\mu',E')$,
  their \emph{tensor product},
  written $(\Omega,\mu,E)\otimes(\Omega',\mu',E')$,
  is
  $(f(\Omega\times\Omega'), \nu \circ f^{-1}, f(E\times E'))$
  where $\nu : \Omega\times\Omega'\to[0,1]$ is the measure
  $\nu(\omega,\omega') = \mu(\omega)\mu'(\omega')$
  and $f$ is an arbitrary isomorphism $\Omega\times\Omega'\to \{1,\dots,|\Omega| |\Omega'|\}$
  (such as $f(\omega,\omega') = |\Omega'|\omega + \omega'$).
  There are canonical projection maps
  \[
    \pi_1 : (\Omega,\mu,E)\otimes(\Omega',\mu',E') \to (\Omega,\mu,E)
    \quad\text{and}\quad
    \pi_2 : (\Omega,\mu,E)\otimes(\Omega',\mu',E') \to (\Omega',\mu',E').
    \]
  The tensor product has a unit, written $\fcpsunit$,
  defined as $(\Omega_\fcpsunit,\mu_\fcpsunit,E_\fcpsunit)$
  where $\Omega_\fcpsunit=E_\fcpsunit=\{0\}$
  and $\mu(0) = 1$.
\end{definition}

\begin{definition}
  To model sampling from finite conditional probability spaces,
  \begin{itemize}[leftmargin=*]
  \item
    For any nonempty
    set $\Omega$, let $\przero_\Omega : \prmonad(A)$ be the Dirac
    distribution at $(0,\nu) : \weightmonoid\times\Omega$ where $\nu$ is an
    arbitrary element of $\Omega$.
  \item
    For any finite conditional probability
    space $(\Omega,\mu,E)$, let $\mu|_E :\prmonad(\Omega)$
    be $\przero_{\Omega}$ if $\mu(E) = 0$, and the lifting of the
    conditional distribution of $\mu$ given $E$ if $\mu(E) > 0$.
  \end{itemize}
\end{definition}

\subsubsection{Exact sublanguage}

\begin{align*}
  \implsemarg{\ecomptydefault{\eM}{\eA}}(\Omega)
    &:\llbr\sGamma\times(\Omega\to\llbr\eDelta)
    \to \esemcomp\eA{(\Omega)}\\
  \implsemarg{\ecomptydefaultsem{\eret~\eM}{\eA}}(\Omega)(\gamma,D) &=
    \gret(\Omega,\idfn,\puresem{\epuretydefaultsem{\eM}{\eA}}\circ D)
    \\
  \implsemarg{\ecomptydefaultsem{\eletin{\eco X}\eM\eN}{\eB}}(\Omega)(\gamma,D) &=
    \left(\begin{aligned}
      &(\Omega_1,f_1,X) \gets \implsemarg{\ecomptydefaultsem\eM\eA}(\Omega)(\gamma,D);\\
      &(\Omega_2,f_2,Y) \gets \implsemarg{\eM}(\Omega_1)(\gamma,(D\circ f_1)[\eco X\mapsto X]);\\
      &\gret(\Omega_2, f_1\circ f_2, Y)
    \end{aligned}\right)
    \\
  \implsemarg{\ecomptydefaultsem{\eite\scoe\eM\eN}{\eA}}(\Omega)(\gamma,D) &=
    \left(\begin{aligned}
      &\mathrm{if}~\puresemarg{\spuretydefaultsem\scoe\sbool}(\gamma)\\
      &\mathrm{then}~\implsemarg{\ecomptydefaultsem\eM\eA}(\Omega)(\gamma,D)\\
      &\mathrm{else}~\implsemarg{\ecomptydefaultsem\eN\eA}(\Omega)(\gamma,D)
    \end{aligned}\right)\\
  \implsemarg{\ecomptydefaultsem{\eflip~\scoe}{\ebool}}(\Omega)(\gamma,D) &=
    \left(\begin{aligned}
      &p := \ite{\puresemarg{\spuretydefaultsem\scoe\sreal}(\gamma)\in[0,1]}
        {\puresemarg{\spuretydefaultsem\scoe\sreal}(\gamma)}
        {0}
        ;\\
      &\Omega_{\rm flip} := (\{0,1\},\mu,\{0,1\})\text{ where $\mu(1) = p$};\\
      &\Omega' := \Omega\otimes\Omega_{\rm flip}; \\
      &X := \omega'\mapsto \mathrm{if}~\pi_2(\omega')=1~\mathrm{then}~\top~\mathrm{else}~\bot;\\
      &\gret(\Omega',\pi_1,X)
    \end{aligned}\right)
    \\
  \implsemarg{\ecomptydefaultsem{\eobserve~\eM}{\eunit}}(\Omega,\mu,E)(\gamma,D) &=
    \left(\begin{aligned}
        &F := (\puresemarg{\epuretydefaultsem\eM\ebool}(\Omega)(D))^{-1}(\top);\\
        &\score(\mu|_E(F));\\
        &\gret((\Omega,\mu,E\cap F),\idfn,\star)
    \end{aligned}\right)
    \\
  \implsemarg{\ecomptydefaultsem{\esample[\scoe]}{\eA}}(\Omega)(\gamma,D) &=
  \left(\begin{aligned}
    &(\Omega',f,x) \gets \implsemarg{\scomptydefaultsem{\scoe}{\stau}}(\Omega)(\gamma,D);\\
    &\gret(\Omega',f,\_\mapsto \valstoe\eA\stau x)
  \end{aligned}\right)
\end{align*}

\subsubsection{Sampling sublanguage}

\begin{align*}
  \implsemarg{\scomptydefault{\scoe}{\stau}}(\Omega) &:
    \llbr\sGamma\times(\Omega\to\llbr\eDelta)
    \to \ssemcomp\stau(\Omega)
    \\
  \implsemarg{\scomptydefaultsem{\sret~{\scoe}}{\stau}}(\Omega)(\gamma,D) &=
    \gret(\Omega, \idfn,\puresemarg{\spuretydefaultsem{\scoe}\stau}(\gamma))
    \\
  \implsemarg{\scomptydefaultsem{\sletin {\sco x}{\seone}{\setwo}}{\stau}}(\Omega)(\gamma,D) &=
    \left(\begin{aligned}
      &(\Omega_1,f_1,x)\gets\implsemarg{\scomptydefaultsem{\seone}\ssigma}(\Omega)(\gamma,D)\\
      &(\Omega_2,f_2,y)\gets\implsemarg{\setwo}(\Omega_1)(\gamma[\sco x\mapsto x],D\circ f_1)\\
      &\gret(\Omega_2,f_1\circ f_2,y)
    \end{aligned}\right)
    \\
  \implsemarg{\scomptydefaultsem{\site\seone\setwo\sethr}{\stau}}(\Omega)(\gamma,D) &=
    \left(\begin{aligned}
      &\mathrm{if}~{\puresemarg{\spuretydefaultsem{\seone}\sbool}(\gamma)}\\
      &\mathrm{then}~{\implsemarg{\scomptydefaultsem{\setwo}\stau}(\Omega)(\gamma,D)}\\
      &\mathrm{else}~{\implsemarg{\scomptydefaultsem{\sethr}\stau}(\Omega)(\gamma,D)}
    \end{aligned}\right)
    \\
  \implsemarg{\scomptydefaultsem{\sflip~\scoe}{\sbool}}(\Omega)(\gamma,D) &=
    \left(\begin{aligned}
      &x\gets \overline{\flip(\puresemarg{\spuretydefaultsem\scoe\sreal}(\gamma))}\\
      &\gret(\Omega,\idfn,x)
    \end{aligned}\right)
    \\
  \implsemarg{\scomptydefaultsem{\sunif~\seone~\setwo}{\sreal}}(\Omega)(\gamma,D) &=
    \left(\begin{aligned}
      &x\gets \overline{\unif(\puresemarg{\spuretydefaultsem\seone\sreal}(\gamma),
        \puresemarg{\spuretydefaultsem\setwo\sreal}(\gamma))}\\
      &\gret(\Omega,\idfn,x)
    \end{aligned}\right)
    \\
  \implsemarg{\scomptydefaultsem{\spois~\scoe}{\sreal}}(\Omega)(\gamma,D) &=
    \left(\begin{aligned}
      &x\gets \overline{\pois(\puresemarg{\spuretydefaultsem\scoe\sreal}(\gamma))}\\
      &\gret(\Omega,\idfn,x)
    \end{aligned}\right)
    \\
  \implsemarg{\scomptydefaultsem{\sobserve{\seobs}{\sflip~\seone}}{\sunit}}(\Omega)(\gamma,D) &=
    \left(\begin{aligned}
        &w := \flip(\puresemarg{\spuretydefaultsem\seone\sreal}(\gamma))\left(\puresemarg{\spuretydefaultsem\seobs\sbool}(\gamma)\right)\\
        &\score(w);\\
        &\gret(\Omega,\idfn,\star)
    \end{aligned}\right)
    \\
  \implsemarg{\scomptydefaultsem{\sobserve{\seobs}{\sunif~\seone~\setwo}}{\sunit}}(\Omega)(\gamma,D) &=
    \left(\begin{aligned}
        &w := \unif(\puresemarg{\spuretydefaultsem\seone\sreal}(\gamma),
          \puresemarg{\spuretydefaultsem\setwo\sreal}(\gamma))
          \left(\puresemarg{\spuretydefaultsem\seobs\sbool}(\gamma)\right)\\
        &\score(w);\\
        &\gret(\Omega,\idfn,\star)
    \end{aligned}\right)
    \\
  \implsemarg{\scomptydefaultsem{\sobserve{\seobs}{\spois~\seone}}{\sunit}}(\Omega)(\gamma,D) &=
    \left(\begin{aligned}
        &w := \pois(\puresemarg{\spuretydefaultsem\seone\sreal}(\gamma))
          \left(\puresemarg{\spuretydefaultsem\seobs\sbool}(\gamma)\right)\\
        &\score(w);\\
        &\gret(\Omega,\idfn,\star)
    \end{aligned}\right)
    \\
  \implsemarg{\scomptydefaultsem{\sexact[\eM]}{\stau}}(\Omega)(\gamma,D) &=
    \left(\begin{aligned}
      &((\Omega',\mu',E'),f,X)\gets \implsemarg{\ecomptydefaultsem\eM{\eA}}(\Omega)(\gamma,D)\\
      &x\gets \big(\omega'\gets \mu'|_{E'};~ \gret(\valetos\eA\stau{X(\omega')})\big)\\
      &\gret((\Omega',\mu',E'\cap X^{-1}(x)), f, x)
    \end{aligned}\right)
\end{align*}

\subsection{Soundness} \label{app:soundness}

\begin{definition}
  Two computations $\mu,\nu : \prmonad A$
  are \emph{equal as importance samplers},
  written $\mu\preq\nu$,
  if for all bounded integrable $k : A\to\R$
  it holds that $\Ex_{(a,x)\sim \mu}[a\cdot k(x)]
  =\Ex_{(b,y)\sim \mu}[b \cdot k(y)]$.
\end{definition}

\begin{lemma}
  The equivalence relation $\preq$ is a congruence
  for the monad $\prmonad$:
  if $\mu\preq\nu : \prmonad A$
  and $f,g : A\to\prmonad B$
  with $f(x)\preq g(x)$
  for all $x$ in $A$
  then $(x\gets \mu; f(x))\preq(x\gets\nu;g(x))$.
\end{lemma}
\begin{proof}
  If $k : B\to \R$ bounded integrable
  then
  \begin{align*}
    \Ex_{(c,y)\sim(x\gets \mu;f(x))}[ck(y)]
    &=
    \Ex_{(a,x)\sim \mu}\left[\Ex_{(b,y)\sim f(x)}[a \cdot b \cdot k(y)]\right]
    \stackrel{(1)}=
    \Ex_{(a,x)\sim \mu}\left[\Ex_{(b,y)\sim g(x)}[a \cdot b \cdot k(y)]\right]\\
    &\stackrel{(2)}=
    \Ex_{(a,x)\sim \nu}\left[\Ex_{(b,y)\sim g(x)}[a \cdot b \cdot k(y)]\right]
    =
    \Ex_{(c,y)\sim(x\gets \nu;g(x))}[ck(y)]
  \end{align*}
  where $(1)$ follows from $f(x)\simeq g(x)$
  and $(2)$ from $\mu\simeq\nu$, using linearity of expectation throughout
  as needed.
\end{proof}

\begin{lemma} \label{app:lem:cond-twice}
  If $(\Omega,\mu,E)\in\fcps$
  and $F\subseteq\Omega$
  then $\mu|_E|_F = \mu|_{E\cap F}$.
\end{lemma}
\begin{proof}
  If $\mu(E\cap F) = 0$ then both sides are the zero measure.
  Otherwise for all $G$ have
  $(\mu|_E|_F)(G)
  = \mu|_E(F\cap G)/\mu|_E(F)
  = \mu(E\cap F\cap G)/\mu(E\cap F)
  = \mu|_{E\cap F}(G)$.
\end{proof}

\begin{lemma} \label{app:lem:score-interchange}
  If $(\Omega,\mu,E)\in\fcps$ then
  $\left(\begin{aligned}
    &\score(\mu(E));\\
    &\omega\gets \mu|_E;\\
    &\gret~\omega
  \end{aligned}\right)
  \preq
  \left(\begin{aligned}
    &\omega\gets \mu;\\
    &\score(\Ind{\omega\in E});\\
    &\gret~\omega
  \end{aligned}\right)$.
\end{lemma}
\begin{proof}
  For all $k : \Omega\to \R$ have
  \begin{align*}
    \Ex_{(a,\omega)\sim\text{LHS}}[k(\omega)]
    &=
    \Ex_{\omega\sim \mu|_E}[\mu(E)k(\omega)]
    =
    \sum_{\omega\in\Omega}\mu|_E(\omega)\mu(E)k(\omega)
    =
    \sum_{\omega\in\Omega}\mu(\omega \cap E)k(\omega)\\
    &=
    \sum_{\omega\in\Omega}\Ind{\omega\in E}\mu(\omega)k(\omega)
    =
    \Ex_{(a,\omega)\sim\text{RHS}} [a \cdot k(\omega)].
  \end{align*}
\end{proof}

\begin{lemma} \label{app:lem:cond-nested-integral}
  If $(\Omega,\mu,E)\in\fcps$ and $X : \Omega\to A$ with $A$ finite
  then \[\left(\begin{aligned}
    &x\gets \big(\omega\gets \mu; \gret(X\omega)\big);\\
    &\omega'\gets \mu|_{X^{-1}(x)};\\
    &\gret(x,\omega)
  \end{aligned}\right)
  =
  \left(\begin{aligned}
    &\omega'\gets \mu;\\
    &\gret(X\omega',\omega')
  \end{aligned}\right).\]
\end{lemma}
\begin{proof}
  All distributions involved are discrete, so it suffices to show
  LHS and RHS give the same probability to pairs $(a,b)$.
  \begin{align*}
    \text{LHS}(a,b)
    &= \sum_{\omega\in \Omega}\mu(\omega)\sum_{\omega\in\Omega}\mu|_{X^{-1}(X\omega)}(\omega')
      \Ind{(X\omega,\omega')=(a,b)}
    = \sum_{\omega\in \Omega,\omega'\in\Omega,X\omega=a,\omega'=b}
      \mu(\omega)
      \mu|_{X^{-1}(X\omega)}(\omega')\\
    &= \sum_{\omega\in\Omega,X\omega=a}\mu(\omega)\mu|_{X^{-1}(a)}(b)
    = \mu|_{X^{-1}(a)}(b)\mu(X^{-1}(a))
    = \mu(X^{-1}(a)\cap b)
    = \text{RHS}(a,b)
  \end{align*}
\end{proof}

\begin{theorem} \label{app:thm:semantic-soundness}
  The following hold:
  \begin{enumerate}
    \item If $\ecomptydefault\eM\eA$ then the following holds
      for all $\Omega,\mu,E,\gamma,D$:
      \[\left(\begin{aligned}
        &((\Omega',\mu',E'),f,X)\gets\implsemarg{\eM}(\Omega,\mu,E)(\gamma,D); \\
        &\omega'\gets \mu'|_{E'};\\
        &\gret(D(f(\omega')),{X\omega'})
      \end{aligned}\right)
      \preq
      \left(\begin{aligned}
        &\omega \gets \mu|_E; \\
        &x\gets \erasesemarg\eM(\gamma,D\omega);\\
        &\gret(D\omega,x)
      \end{aligned}\right)\]
    \item If $\scomptydefault\scoe\stau$ then the following holds
      for all $\Omega,\mu,E,\gamma,D$:
      \[\left(\begin{aligned}
        &((\Omega',\mu',E'),f,x)\gets\implsemarg{\scoe}(\Omega,\mu,E)(\gamma,D); \\
        &\omega'\gets \mu'|_{E'};\\
        &\gret(D(f(\omega)),x)
      \end{aligned}\right)
      \preq
      \left(\begin{aligned}
        &\omega \gets \mu|_E; \\
        &x\gets \erasesemarg\scoe(\gamma,D\omega);\\
        &\gret(D\omega, x)
      \end{aligned}\right)\]
  \end{enumerate}
\end{theorem}
\begin{proof}
  By induction on the typing rules. \begin{enumerate}[leftmargin=*]
  \item ~\begin{itemize}[align=left]
    \item[$({\ecomptydefault{\eret~\eM}{\eA}})$]
      \begin{align*}
        &\left(\begin{aligned}
            &((\Omega',\mu',E'),f,X)\gets\implsemarg{\eret~\eM}(\Omega,\mu,E)(\gamma,D); \\
            &\omega'\gets \mu'|_{E'};\\
            &\gret(D(f(\omega')),{X\omega'})
        \end{aligned}\right)\\
        &=\left(\begin{aligned}
            &((\Omega',\mu',E'),f,X)\gets \gret((\Omega,\mu,E),\idfn,\puresem{\epuretydefaultsem{\eM}{\eA}}\circ D);\\
            &\omega'\gets \mu'|_{E'};\\
            &\gret(D(f(\omega')),{X\omega'})
        \end{aligned}\right)\\
        &=
        \left(\begin{aligned}
            &\omega \gets \mu|_E; \\
            &\gret(D\omega,\puresemarg{\eM}(D\omega))
        \end{aligned}\right)
        =
        \left(\begin{aligned}
            &\omega \gets \mu|_E; \\
            &x\gets \erasesemarg{\eret~\eM}(\gamma,D\omega);\\
            &\gret(D\omega,x)
        \end{aligned}\right)
      \end{align*}
    \item[$({\ecomptydefault{\eletin{\eco X}\eM\eN}{\eB}})$]
      In this case we work backwards, rearranging RHS into LHS:
      \begin{align*}
        &\left(\begin{aligned}
            &\omega\gets \mu|_E;\\
            &y\gets \erasesemarg{\eletin{\eco X}\eM\eN}(\gamma,D\omega);\\
            &\gret(D\omega,y)
            \end{aligned}\right)
        =
        \left(\begin{aligned}
            &\omega\gets \mu|_E;\\
            &x\gets \erasesemarg\eM(\gamma,D\omega);\\
            &y\gets \erasesemarg\eN(\gamma,D\omega[\eco X\mapsto x]);\\
            &\gret(D\omega,y)
            \end{aligned}\right)\\
        &=
        \left(\begin{aligned}
            &(\delta,x)\gets\\
            &\hspace{1em}\begin{aligned}
                &\omega\gets \mu|_E;\\
                &x\gets \erasesemarg\eM(\gamma,D\omega);\\
                &\gret(D\omega,x)
            \end{aligned}\\
            &y\gets \erasesemarg\eN(\gamma,\delta[\eco X\mapsto x]);\\
            &\gret(D\omega,y)
            \end{aligned}\right)
        \IHpreq
        \left(\begin{aligned}
            &(\delta,x)\gets\\
            &\hspace{1em}\begin{aligned}
              &((\Omega_1,\mu_1,E_1),f,X)\gets \implsemarg{\eM}(\Omega,\mu,E)(\gamma,D);\\
              &\omega_1\gets \mu_1|_{E_1};\\
              &\gret(D(f(\omega_1)),X\omega_1)
            \end{aligned}\\
            &y\gets \erasesemarg\eN(\gamma,\delta[\eco X\mapsto x]);\\
            &\gret(D\omega,y)
            \end{aligned}\right)
          \\
        &=
        \left(\begin{aligned}
            &((\Omega_1,\mu_1,E_1),f,X)\gets \implsemarg{\eM}(\Omega,\mu,E)(\gamma,D);\\
            &\omega_1\gets \mu_1|_{E_1};\\
            &y\gets \erasesemarg\eN(\gamma,((D\circ f_1)[\eco X\mapsto X])(\omega_1));\\
            &\gret(D\omega,y)
            \end{aligned}\right)
        \\
        &\IHpreq
        \left(\begin{aligned}
        &((\Omega_1,\mu_1,E_1),f_1,X)\gets \implsemarg{\eM}(\Omega,\mu,E)(\gamma,D);\\
        &((\Omega_2,\mu_2,E_2),f_2,Y)\gets \implsemarg{\eN}(\Omega_1,\mu_1,E_1)(\gamma,(D\circ f_1)[\eco X\mapsto X]);\\
        &\omega_2\gets \mu_2|_{E_2};\\
        &\gret(D(f_1(f_2(\omega))),Y\omega_2)
        \end{aligned}\right)
        \\
        &=
        \left(\begin{aligned}
        &((\Omega',\mu',E'),f,Z)\gets \implsemarg{\eletin{\eco X}\eM\eN}(\Omega,\mu,E)(\gamma,D);\\
        &\omega'\gets \mu'|_{E'};\\
        &\gret(D(f(\omega)),Z\omega')
        \end{aligned}\right)
      \end{align*}
    \item[$({\ecomptydefault{\eite\scoe\eM\eN}{\eA}})$]
      \begin{align*}
        &\left(\begin{aligned}
            &((\Omega',\mu',E'),f,X)\gets\implsemarg{\eite\scoe\eM\eN}(\Omega,\mu,E)(\gamma,D); \\
            &\omega'\gets \mu'|_{E'};\\
            &\gret(D(f(\omega')),{X\omega'})
        \end{aligned}\right) \\
        &=\left(\begin{aligned}
            &((\Omega',\mu',E'),f,X)\gets\\
            &\hspace{1em}\begin{aligned}
              &\mathrm{if}~\puresemarg{\scoe}(\gamma)\\
              &\mathrm{then}~\implsemarg{\eM}(\Omega,\mu,E)(\gamma,D);\\
              &\mathrm{else}~\implsemarg{\eN}(\Omega,\mu,E)(\gamma,D);\\
            \end{aligned}\\
            &\omega'\gets \mu'|_{E'};\\
            &\gret(D(f(\omega')),{X\omega'})
        \end{aligned}\right)
        =\left(\begin{aligned}
            &\mathrm{if}~\puresemarg{\scoe}(\gamma)~\mathrm{then}\\
            &\hspace{1em}\begin{aligned}
                &((\Omega',\mu',E'),f,X)\gets\implsemarg{\eM}(\Omega,\mu,E)(\gamma,D);\\
                &\omega'\gets \mu'|_{E'};\\
                &\gret(D(f(\omega')),{X\omega'})
            \end{aligned}\\
            &\mathrm{else}\\
            &\hspace{1em}\begin{aligned}
                &((\Omega',\mu',E'),f,X)\gets\implsemarg{\eN}(\Omega,\mu,E)(\gamma,D);\\
                &\omega'\gets \mu'|_{E'};\\
                &\gret(D(f(\omega')),{X\omega'})
            \end{aligned}
        \end{aligned}\right)
        \\
        &\stackrel{\text{IH}\times2}\preq
        \left(\begin{aligned}
            &\mathrm{if}~\puresemarg{\scoe}(\gamma)~\mathrm{then}\\
            &\hspace{1em}\begin{aligned}
                &\omega\gets \mu|_E;\\
                &x\gets\erasesemarg{\eM}(\gamma,D\omega);\\
                &\gret(D\omega,x)
            \end{aligned}\\
            &\mathrm{else}\\
            &\hspace{1em}\begin{aligned}
                &\omega\gets \mu|_E;\\
                &x\gets\erasesemarg{\eN}(\gamma,D\omega);\\
                &\gret(D\omega,x)
            \end{aligned}
        \end{aligned}\right)
        =
        \left(\begin{aligned}
            &\omega\gets \mu|_E;\\
            &x\gets\\
            &\hspace{1em}\begin{aligned}
              &\mathrm{if}~\puresemarg{\scoe}(\gamma)\\
              &\mathrm{then}~\erasesemarg{\eM}(\gamma,D\omega)\\
              &\mathrm{else}~\erasesemarg{\eN}(\gamma,D\omega)
            \end{aligned}
            \\
            &\gret(D\omega,x)
        \end{aligned}\right)\\
        &=
        \left(\begin{aligned}
            &\omega\gets \mu|_E;\\
            &x\gets \erasesemarg{\eite\scoe\eM\eN}(\gamma,D\omega);\\
            &\gret(D\omega,x)
        \end{aligned}\right)
      \end{align*}
    \item[$({\ecomptydefault{\eflip~\scoe}{\ebool}})$]
      Let $p = (\ite{\puresemarg{\spuretydefaultsem\scoe\sreal}(\gamma)\in[0,1]}{\puresemarg{\spuretydefaultsem\scoe\sreal}(\gamma)}0)$.
      \begin{align*}
        &\left(\begin{aligned}
            &((\Omega',\mu',E'),f,X)\gets\implsemarg{\eflip~\scoe}(\Omega,\mu,E)(\gamma,D); \\
            &\omega'\gets \mu'|_{E'};\\
            &\gret(D(f(\omega')),{X\omega'})
        \end{aligned}\right)
        =
        \left(\begin{aligned}
            &(\omega,b)\gets (\mu\otimes\flip(p))|_{E\times\llbrspc\ebool};\\
            &\gret(D\omega,b)
        \end{aligned}\right)
        \\
        &=
        \left(\begin{aligned}
            &\omega \gets \mu|_E; \\
            &b\gets \flip(p);\\
            &\gret(D\omega,b)
        \end{aligned}\right)
        =
        \left(\begin{aligned}
            &\omega \gets \mu|_E; \\
            &x\gets \erasesemarg{\eflip~\scoe}(\gamma,D\omega);\\
            &\gret(D\omega,x)
        \end{aligned}\right)
      \end{align*}
    \item[$({\ecomptydefault{\eobserve~\eM}{\eunit}})$]
      Let $F$ be the subset $(\llbr{\eM}(\Omega,\mu,E)\circ D)^{-1}(\top)$ of $\Omega$.
      \begin{align*}
      &\left(\begin{aligned}
        &((\Omega',\mu',E'),f,X)\gets\implsemarg{\eobserve~\eM}(\Omega,\mu,E)(\gamma,D); \\
        &\omega'\gets \mu'|_{E'};\\
        &\gret(D(f(\omega')),{X\omega'})
      \end{aligned}\right)
      =\left(\begin{aligned}
        &\score(\mu|_E(F));\\
        &\omega\gets \mu|_{E\cap F};\\
        &\gret(D\omega,\star)
      \end{aligned}\right)\\
      &\stackcref{app:lem:cond-twice}=\left(\begin{aligned}
        &\score(\mu|_E(F));\\
        &\omega\gets \mu|_E|_F;\\
        &\gret(D\omega,\star)
      \end{aligned}\right)
      \stackcref{app:lem:score-interchange}\preq
      \left(\begin{aligned}
        &\omega \gets \mu|_E; \\
        &\score(\Ind{\omega\in F});\\
        &\gret(D\omega,\star)
        \end{aligned}\right)
      \\
      &=
      \left(\begin{aligned}
        &\omega \gets \mu|_E; \\
        &\score(\Ind{\puresemarg{\eM}(D\omega) = \top});\\
        &\gret(D\omega,\star)
        \end{aligned}\right)
      =
      \left(\begin{aligned}
        &\omega \gets \mu|_E; \\
        &x\gets \erasesemarg{\eobserve~\eM}(\gamma,D\omega);\\
        &\gret(D\omega,x)
        \end{aligned}\right)
      \end{align*}
    \item[$({\ecomptydefault{\esample[\scoe]}{\eA}})$]
      \begin{align*}
      &\left(\begin{aligned}
        &((\Omega',\mu',E'),f,X)\gets\implsemarg{\esample[\scoe]}(\Omega,\mu,E)(\gamma,D); \\
        &\omega'\gets \mu'|_{E'};\\
        &\gret(D(f(\omega')),X\omega')
      \end{aligned}\right)
      \\
      &=
      \left(\begin{aligned}
        &((\Omega',\mu',E'),f,X)\gets\\
        &\hspace{1em}
            \left(\begin{aligned}
                &((\Omega',\mu',E'),f,x) \gets \implsemarg\scoe(\Omega,\mu,E)(\gamma,D);\\
                &\gret((\Omega',\mu',E'),f,\_\mapsto \valstoe\eA\stau{x})
            \end{aligned}\right);
            \\
        &\omega'\gets \mu'|_{E'};\\
        &\gret(D(f(\omega')),X\omega')
      \end{aligned}\right)
      \\
      &=
      \left(\begin{aligned}
        &((\Omega',\mu',E'),f, x) \gets \implsemarg\scoe(\Omega,\mu,E)(\gamma,D);\\
        &\omega'\gets \mu'|_{E'};\\
        &\gret(D(f(\omega')),\valstoe\eA\stau x)
      \end{aligned}\right)
      \\
      &\IHpreq
      \left(\begin{aligned}
        &\omega \gets \mu|_E; \\
        &x\gets \erasesemarg\scoe(\gamma,D\omega);\\
        &\gret(D\omega,\valstoe\eA\stau x)
      \end{aligned}\right)
      =
      \left(\begin{aligned}
        &\omega \gets \mu|_E; \\
        &x\gets \erasesemarg{\esample[\scoe]}(\gamma,D\omega);\\
        &\gret(D\omega,x)
      \end{aligned}\right)
    \end{align*}
  \end{itemize}
  \item ~\begin{itemize}[align=left]
    \item[$({\scomptydefault{\sret~{\scoe}}{\stau}})$]
      \begin{align*}
        &\left(\begin{aligned}
            &((\Omega',\mu',E'),f,x)\gets\implsemarg{\sret~\scoe}(\Omega,\mu,E)(\gamma,D); \\
            &\omega'\gets \mu'|_{E'};\\
            &\gret(D(f(\omega)),x)
        \end{aligned}\right)
        \\
        &=\left(\begin{aligned}
            &\omega \gets \mu|_E; \\
            &\gret(D\omega, \puresemarg{\scoe}(\gamma))
        \end{aligned}\right)
        =
        \left(\begin{aligned}
            &\omega \gets \mu|_E; \\
            &x\gets \erasesemarg{\sret~\scoe}(\gamma,D\omega);\\
            &\gret(D\omega, x)
        \end{aligned}\right)
      \end{align*}
    \item[$({\scomptydefault{\sletin {\sco x}{\seone}{\setwo}}{\stau}})$]
      In this case we work backwards, rearranging RHS into LHS:
      \begin{align*}
        &\left(\begin{aligned}
            &\omega \gets \mu|_E; \\
            &x\gets \erasesemarg{\sletin{\sco x}\seone\setwo}(\gamma,D\omega);\\
            &\gret(D\omega,x)
        \end{aligned}\right)
        =
        \left(\begin{aligned}
            &\omega \gets \mu|_E; \\
            &x\gets \erasesemarg{\seone}(\gamma,D\omega);\\
            &y\gets \erasesemarg{\setwo}(\gamma[\sco x\mapsto x],D\omega);\\
            &\gret(D\omega,y)
        \end{aligned}\right)
        \\
        &=
        \left(\begin{aligned}
            &(\delta,x)\gets\\
            &\hspace{1em} \left(\begin{aligned}
                &\omega\gets \mu|_E; \\
                &x\gets \erasesemarg{\seone}(\gamma,D\omega);\\
                &\gret(D\omega,x)
              \end{aligned}\right);\\
            &y\gets \erasesemarg{\setwo}(\gamma[\sco x\mapsto x],\delta);\\
            &\gret(\delta,y)
        \end{aligned}\right)
        \IHpreq
        \left(\begin{aligned}
            &(\delta,x)\gets\\
            &\hspace{1em} \left(\begin{aligned}
                &((\Omega_1,\mu_1,E_1),f_1,x)\gets\implsemarg{\seone}(\Omega,\mu,E)(\gamma,D); \\
                &\omega_1\gets \mu_1|_{E_1}; \\
                &\gret(D(f_1(\omega_1)),X\omega_1)
              \end{aligned}\right);\\
            &y\gets \erasesemarg{\setwo}(\gamma[\sco x\mapsto x],\delta);\\
            &\gret(\delta,y)
        \end{aligned}\right)
        \\
        &=
        \left(\begin{aligned}
            &((\Omega_1,\mu_1,E_1),f_1,x)\gets\implsemarg{\seone}(\Omega,\mu,E)(\gamma,D); \\
            &\omega_1\gets \mu_1|_{E_1}; \\
            &\delta := D(f_1(\omega_1));\\
            &y\gets \erasesemarg{\setwo}(\gamma[\sco x\mapsto x],\delta);\\
            &\gret(\delta,y)
        \end{aligned}\right)\\
        &\IHpreq
        \left(\begin{aligned}
            &((\Omega_1,\mu_1,E_1),f_1,x)\gets\implsemarg{\seone}(\Omega,\mu,E)(\gamma,D); \\
            &((\Omega_2,\mu_2,E_2),f_2,y)\gets\implsemarg{\setwo}
              (\Omega_1,\mu_1,E_1)(\gamma[\sco x\mapsto x],D\circ f_1); \\
            &\omega_2\gets \mu_2|_{E_2};\\
            &\gret(D(f_1(f_2(\omega_2))),y)
        \end{aligned}\right)
        \\
        &=
        \left(\begin{aligned}
            &((\Omega',\mu',E'),f,z)\gets\implsemarg{\sletin{\sco x}\seone\setwo}(\Omega,\mu,E)(\gamma,D); \\
            &\omega'\gets \mu'|_{E'};\\
            &\gret(D(f(\omega')),z)
        \end{aligned}\right)
      \end{align*}
    \item[$({\scomptydefault{\site\seone\setwo\sethr}{\stau}})$]
      \begin{align*}
        &\left(\begin{aligned}
            &((\Omega',\mu',E'),f,x)\gets\implsemarg{\site\seone\setwo\sethr}(\Omega,\mu,E)(\gamma,D); \\
            &\omega'\gets \mu'|_{E'};\\
            &\gret(D(f(\omega')),x)
        \end{aligned}\right)\\
        &=\left(\begin{aligned}
            &((\Omega',\mu',E'),f,x)\gets\\
            &\hspace{1em}\begin{aligned}
              &\mathrm{if}~\puresemarg{\seone}(\gamma)\\
              &\mathrm{then}~\implsemarg{\setwo}(\Omega,\mu,E)(\gamma,D);\\
              &\mathrm{else}~\implsemarg{\sethr}(\Omega,\mu,E)(\gamma,D);\\
            \end{aligned}\\
            &\omega'\gets \mu'|_{E'};\\
            &\gret(D(f(\omega')),x)
        \end{aligned}\right)
        =\left(\begin{aligned}
            &\mathrm{if}~\puresemarg{\seone}(\gamma)~\mathrm{then}\\
            &\hspace{1em}\begin{aligned}
                &((\Omega',\mu',E'),f,x)\gets\implsemarg{\setwo}(\Omega,\mu,E)(\gamma,D);\\
                &\omega'\gets \mu'|_{E'};\\
                &\gret(D(f(\omega')),x)
            \end{aligned}\\
            &\mathrm{else}\\
            &\hspace{1em}\begin{aligned}
                &((\Omega',\mu',E'),f,x)\gets\implsemarg{\sethr}(\Omega,\mu,E)(\gamma,D);\\
                &\omega'\gets \mu'|_{E'};\\
                &\gret(D(f(\omega')),x)
            \end{aligned}
        \end{aligned}\right)
        \\
        &\stackrel{\text{IH}\times2}\preq
        \left(\begin{aligned}
            &\mathrm{if}~\puresemarg{\seone}(\gamma)~\mathrm{then}\\
            &\hspace{1em}\begin{aligned}
                &\omega\gets \mu|_E;\\
                &x\gets\erasesemarg{\setwo}(\gamma,D\omega);\\
                &\gret(D\omega,x)
            \end{aligned}\\
            &\mathrm{else}\\
            &\hspace{1em}\begin{aligned}
                &\omega\gets \mu|_E;\\
                &x\gets\erasesemarg{\sethr}(\gamma,D\omega);\\
                &\gret(D\omega,x)
            \end{aligned}
        \end{aligned}\right)
        =
        \left(\begin{aligned}
            &\omega\gets \mu|_E;\\
            &x\gets\\
            &\hspace{1em}\begin{aligned}
              &\mathrm{if}~\puresemarg{\seone}(\gamma)\\
              &\mathrm{then}~\erasesemarg{\setwo}(\gamma,D\omega)\\
              &\mathrm{else}~\erasesemarg{\sethr}(\gamma,D\omega)
            \end{aligned}
            \\
            &\gret(D\omega,x)
        \end{aligned}\right)
        \\
        &=
        \left(\begin{aligned}
            &\omega\gets \mu|_E;\\
            &x\gets \erasesemarg{\site\seone\setwo\sethr}(\gamma,D\omega);\\
            &\gret(D\omega,x)
        \end{aligned}\right)
      \end{align*}
    \item[$({\scomptydefault\scoe{\stau}}$ for $\scoe = \sflip~\seone$ or $\scoe = \sunif~\seone~\setwo$ or $\scoe = \spois~{\seone}$)]
      ~\linebreak
      If $\scoe = \sflip~{\seone}$ for some $\seone$,
      then let $\nu$
      be the distribution $\flip(\puresemarg{\spuretydefaultsem{\seone}\sreal}(\gamma))$.

      If $\scoe = \sunif~\seone~\setwo$ then let
      $\nu$ be the distribution $\unif(\puresemarg{\spuretydefaultsem\seone\sreal}(\gamma),
        \puresemarg{\spuretydefaultsem\seone\sreal}(\gamma))$.

      If $\scoe = \spois~{\seone}$ for some $\seone$,
      then let $\nu$
      be the distribution $\pois(\puresemarg{\spuretydefaultsem{\seone}\sreal}(\gamma))$.

      In all cases,
      \begin{align*}
        &\left(\begin{aligned}
            &((\Omega',\mu',E'),f,x)\gets\implsemarg{\scoe}(\Omega,\mu,E)(\gamma,D); \\
            &\omega'\gets \mu'|_{E'};\\
            &\gret(D(f(\omega')),x)
        \end{aligned}\right)
        =\left(\begin{aligned}
            &((\Omega',\mu',E'),f,x)\gets\\
            &\hspace{1em}\left(\begin{aligned}
                &x\gets\overline\nu;\\
                &\gret((\Omega,\mu,E),\idfn,x)
            \end{aligned}\right);\\
            &\omega'\gets \mu'|_{E'};\\
            &\gret(D(f(\omega')),x)
        \end{aligned}\right)\\
        &=
        \left(\begin{aligned}
            &x\gets \overline\nu;\\
            &\omega\gets \mu|_{E};\\
            &\gret(D\omega,x)
        \end{aligned}\right)
        =
        \left(\begin{aligned}
            &\omega\gets \mu|_{E};\\
            &x\gets \erasesemarg{\scoe}(\gamma,D\omega);\\
            &\gret(D\omega,x)
        \end{aligned}\right)
      \end{align*}
    \item[$({\scomptydefault{\sobserve{\seobs}{e}} {\sunit}}$ for $e = \sflip~\seone$ or $e = \sunif~\seone~\setwo$ or $e = \spois~{\seone}$)] %
      ~\linebreak
      If $e = \sflip~{\seone}$ for some $\seone$,
      then let $\nu$
      be the distribution $\flip(\puresemarg{\spuretydefaultsem{\seone}\sreal}(\gamma))$.

      If $e = \sunif~\seone~\setwo$ then let
      $\nu$ be the distribution $\unif(\puresemarg{\spuretydefaultsem\seone\sreal}(\gamma),
        \puresemarg{\spuretydefaultsem\seone\sreal}(\gamma))$.

      If $e = \spois~{\seone}$ for some $\seone$,
      then let $\nu$
      be the distribution $\pois(\puresemarg{\spuretydefaultsem{\seone}\sreal}(\gamma))$.

      In all cases,
      \begin{align*}
        &\left(\begin{aligned}
            &((\Omega',\mu',E'),f,x)\gets\implsemarg{\sobserve{\seobs}{\scoe}}(\Omega,\mu,E)(\gamma,D); \\
            &\omega'\gets \mu'|_{E'};\\
            &\gret(D(f(\omega')),x)
        \end{aligned}\right)
        =\left(\begin{aligned}
            &\score(\overline\nu(\puresemarg\seobs(\gamma)));\\
            &\omega \gets \mu|_{E};\\
            &\gret(D\omega,\star)
        \end{aligned}\right)\\
      \end{align*}
      \cont{}'s $\score$ does not effect the FCPS and we commute with $\mu|_E$,
      \begin{align*}
        &\overset{\textrm{comm}}{=}\left(\begin{aligned}
            &\omega\gets \mu|_{E};\\
            &\score(\overline\nu(\puresemarg\seobs(\gamma)));\\
            &\gret(D\omega,\star)
        \end{aligned}\right)
        =
        \left(\begin{aligned}
            &\omega\gets \mu|_{E};\\
            &x\gets \erasesemarg{\scoe}(\gamma,D\omega);\\
            &\gret(D\omega,x)
        \end{aligned}\right)
      \end{align*}
    \item[$({\scomptydefault{\sexact[\eM]}{\stau}})$]
      \begin{align*}
        &\left(\begin{aligned}
            &((\Omega',\mu',E'),f,x)\gets\implsemarg{\sexact[\eM]}(\Omega,\mu,E)(\gamma,D); \\
            &\omega'\gets \mu'|_{E'};\\
            &\gret(D(f(\omega')),x)
        \end{aligned}\right)
        \\
        &=\left(\begin{aligned}
            &((\Omega',\mu',E'),f,x)\gets\\
            &\hspace{1em}\left(\begin{aligned}
                    &((\Omega',\mu',E'),f,X)\gets \implsemarg{\ecomptydefaultsem\eM{\eA}}(\Omega,\mu,E)(\gamma,D)\\
                &x\gets \big(\omega'\gets \mu'|_{E'};~ \gret(\valetos\eA\stau{X\omega'})\big)\\
                &\gret((\Omega',\mu',E'\cap X^{-1}(x)), x)
                \end{aligned}\right)\\
            &\omega'\gets \mu'|_{E'};\\
            &\gret(D(f(\omega')),x)
        \end{aligned}\right)
        \\
        &=
        \left(\begin{aligned}
            &((\Omega',\mu',E'),f,X)\gets \implsemarg{\ecomptydefaultsem\eM{\eA}}(\Omega,\mu,E)(\gamma,D)\\
            &x\gets \big(\omega'\gets \mu'|_{E'};~ \gret(\valetos\eA\stau{X\omega'})\big)\\
            &E'' := E'\cap X^{-1}(x);\\
            &\omega'\gets \mu'|_{E''};\\
            &\gret(D(f(\omega')),x)
        \end{aligned}\right)\\
        &\stackcref{app:lem:cond-twice,app:lem:cond-nested-integral}=
        \left(\begin{aligned}
            &((\Omega',\mu',E'),X)\gets \implsemarg{\ecomptydefaultsem\eM{\eA}}(\Omega,\mu,E)(\gamma,D)\\
            &\omega'\gets \mu'|_{E'};\\
            &\gret(D(f(\omega')),\valetos\eA\stau{X\omega'})
        \end{aligned}\right)\\
      &\IHpreq\left(\begin{aligned}
        &\omega \gets \mu|_E; \\
        &x\gets \erasesemarg{\eM}(\gamma,D\omega);\\
        &\gret(D\omega, \valetos\eA\stau x)
      \end{aligned}\right)
      =\left(\begin{aligned}
        &\omega \gets \mu|_E; \\
        &x\gets \erasesemarg{\sexact[\eM]}(\gamma,D\omega);\\
        &\gret(D\omega, x)
      \end{aligned}\right)
      \end{align*}
  \end{itemize}
  \end{enumerate}
\end{proof}

\begin{definition}
  For a closed program $\scompty{\sco{\pmb\cdot}}{\eco{\pmb\cdot}}\scoe\stau$,
  let $\sevalprog(e) $ be the computation
  \[\left(\begin{aligned}
    &(\_,\_,x)\gets \implsemarg{\scoe}(\fcpsunit)(\emptyset,\emptyset);\\
    &\gret~ x
  \end{aligned}\right) : \prmonad{\llbr\stau}\]
  where $\emptyset$ denotes the empty substitution.
  Let $\seraseevalprog(e)$ be the computation
  $\erasesemarg{\scoe}(\emptyset,\emptyset) : \prmonad{\llbr\stau}$.
\end{definition}

\begin{theorem}[\cont{} soundness] \label{app:thm:toplevel-soundness}
  If $\scompty{\sco\cdot}{\eco\cdot}\scoe\stau$
  then
  $\sevalprog(\scoe) \preq \seraseevalprog(\scoe)$.
\end{theorem}
\begin{proof} Apply \cref{app:thm:semantic-soundness}. \end{proof}

\subsection{Evaluation details}
In this section, we include three example programs from our evaluation to
showcase the syntax of our implementation. These include the 15-node arrival model with tree-topology, the 9-node reachability model with grid topology.

We additionally provide our evaluations, inclusive of standard error over the 100 runs.

\begin{figure}[H]
\begin{lstlisting}[
 basicstyle=\footnotesize\ttfamily,
  language=custom,mathescape=true,caption={\textsc{Reachability-9}},label={lst:reliability:apx},escapechar=|]
exact {
  let x00 = flip 1.0 / 3.0 in
  let x01 = if x00 then flip 1.0 / 4.0 else flip 1.0 / 5.0 in
  let x10 = if x00 then flip 1.0 / 4.0 else flip 1.0 / 5.0 in
  let diag = sample {
    x02 ~ if x01 then bern(1.0 / 4.0) else bern(1.0 / 5.0);
    x20 ~ if x10 then bern(1.0 / 4.0) else bern(1.0 / 5.0);
    x11 ~ if  x10 &&  x01 then bern(1.0 / 6.0)
     else if  x10 && !x01 then bern(1.0 / 7.0)
     else if !x10 &&  x01 then bern(1.0 / 8.0)
                          else bern(1.0 / 9.0);
    (x20, x11, x02)
  } in
  let x20 = diag[0] in
  let x11 = diag[1] in
  let x02 = diag[2] in

  let x12 = if  x11 &&  x02 then flip 1.0 / 6.0
       else if  x11 && !x02 then flip 1.0 / 7.0
       else if !x11 &&  x02 then flip 1.0 / 8.0
                            else flip 1.0 / 9.0 in
  let x21 = if  x20 &&  x11 then flip 1.0 / 6.0
       else if  x20 && !x11 then flip 1.0 / 7.0
       else if !x20 &&  x11 then flip 1.0 / 8.0
                            else flip 1.0 / 9.0 in
  let x22 = if  x21 &&  x12 then flip 1.0 / 6.0
       else if  x21 && !x12 then flip 1.0 / 7.0
       else if !x21 &&  x12 then flip 1.0 / 8.0
                            else flip 1.0 / 9.0 in
  observe x22 in
  ( x00, x01, x02
  , x10, x11, x12
  , x20, x21, x22
  )
}
\end{lstlisting}
\caption{The network reachability program for the 9-node grid topology}
\end{figure}

\begin{figure}[H]
\begin{lstlisting}[
 basicstyle=\footnotesize\ttfamily,
 language=custom,mathescape=true,caption={\textsc{Arrival-15}},label={lst:arrival:apx},escapechar=|]
exact fn network () -> Bool {
  let n30r = true in
  let n20r = if  n30r then flip 1.0 / 2.0 else false in
  let n31r = if !n20r then flip 1.0 / 2.0 else false in

  let n10r = if  n20r then flip 1.0 / 2.0 else false in
  let n21r = if !n10r then flip 1.0 / 2.0 else false in
  let n32r = if  n21r then flip 1.0 / 2.0 else false in
  let n33r = if !n21r then flip 1.0 / 2.0 else false in

  let n0   = n10r in

  let n10l = if  n0   then flip 1.0 / 2.0 else false in

  let n20l = if  n10l then flip 1.0 / 2.0 else false in
  let n21l = if !n10l then flip 1.0 / 2.0 else false in

  let n30l = if  n20l then flip 1.0 / 2.0 else false in
  let n31l = if !n20l then flip 1.0 / 2.0 else false in
  let n32l = if  n21l then flip 1.0 / 2.0 else false in
  let n33l = if !n21l then flip 1.0 / 2.0 else false in
  observe n32l in
  n0
}

sample {
  ix ~ poisson(3.0);
  npackets <- 0;
  while ix > 0 {
    traverses <- exact(network());
    npackets <- if traverses { npackets + 1 } else { npackets };
    ix <- ix - 1;
    true
  };
  npackets
}
\end{lstlisting}
\caption{The network reachability program for the 9-node grid topology}
\end{figure}

\begin{figure}[H]
\begin{lstlisting}[
 basicstyle=\footnotesize\ttfamily,
 language=custom,mathescape=true,caption={\textsc{Gossip-4}},label={lst:gossip:apx}]
sample fn forward(ix : Int) -> Int {
  s ~ discrete(1.0 / 3.0, 1.0 / 3.0, 1.0 / 3.0);
  if s < ix { s } else { s + 1 }
}

exact fn node(nid : Int) -> (Int, Int) {
  let p1 = sample(forward(nid)) in
  let p2 = sample(forward(nid)) in
  (p1, p2)
}

exact fn network_step(
  n0 : Bool, n1 : Bool, n2: Bool, n3 : Bool, next : Int
) -> (Bool, Bool, Bool, Bool, Int, Int) {
  let n0 = n0 || (next == 0) in
  let n1 = n1 || (next == 1) in
  let n2 = n2 || (next == 2) in
  let n3 = n3 || (next == 3) in
  let fwd = node(next) in
  (n0, n1, n2, n3, fwd[0], fwd[1])
}

sample fn as_num(b : Bool) -> Float {
  if (b) { 1.0 } else { 0.0 }
}

sample {
  p <- exact(node(0));
  p1 <- p[0];  p2 <- p[1];
  i0 <- true; i1 <- false; i2 <- false; i3 <- false;
  q  <- []; q  <- push(q, p1); q  <- push(q, p2);
  num_steps ~ discrete(0.25,0.25,0.25,0.25);
  num_steps <- num_steps + 4;
  while (num_steps > 0) {
    nxt <- head(q);
    q   <- tail(q);
    state <- exact(network_step(i0, i1, i2, i3, nxt));
    i0 <- state[0]; i1 <- state[1]; i2 <- state[2]; i3 <- state[3];
    q  <- push(q, state[4]);
    q  <- push(q, state[5]);
    num_steps <- num_steps - 1;
    true
  };
  n0 <- as_num(i0);
  n1 <- as_num(i1);
  n2 <- as_num(i2);
  n3 <- as_num(i3);
  (n0 + n1 + n2 + n3)
}
\end{lstlisting}
\caption{The network reachability program for the 9-node grid topology}
\end{figure}

\begin{figure}[H]
\begin{sideways}
  \begin{minipage}{\textheight}
\centering
\begin{tabular}{c|ccccccccc}
  \toprule
        \multirow{2}{*}{Model} & \multicolumn{2}{c}{PSI} &              \multicolumn{2}{c}{Pyro}         &            \multicolumn{2}{c}{\host{} (\cont{})} &                \multicolumn{2}{c}{\host{}} \\
                               &       L1 &      Time(s) &                    L1 &              Time(s)  &                    L1 &                 Time(s)  &                      L1 &         Time(s)  \\
\midrule
 arrival/tree-15               &        --- &            --- &          0.365 $\pm$ 0.004 &         12.713 $\pm$0.025 &          0.355 $\pm$0.004 &    \textbf{0.247 $\pm$0.001} &   \textbf{0.337 $\pm$0.003} &     0.349 $\pm$0.002 \\
 arrival/tree-31               &        --- &            --- &          0.216 $\pm$0.002 &         26.366 $\pm$0.054 &          0.218 $\pm$0.002 &    \textbf{0.561 $\pm$0.004} &   \textbf{0.179 $\pm$0.002} &     0.754 $\pm$0.002 \\
 arrival/tree-63               &        --- &            --- &          0.118 $\pm$0.002 &         53.946 $\pm$0.086 &          0.120 $\pm$0.002 &    \textbf{1.469 $\pm$0.003} &   \textbf{0.093 $\pm$0.002} &     1.912 $\pm$0.004 \\
\midrule
 alarm                        &      t/o &          t/o &           1.290 $\pm$0.056 &         16.851 $\pm$0.024 &          1.173 $\pm$0.049 & \textbf{0.433 $\pm$0.002} &  \textbf{0.364 $\pm$0.015} & 14.444 $\pm$0.008 \\
 insurance                    &      t/o &          t/o &           0.149 $\pm$0.008 &         13.724 $\pm$0.020 &          0.144 $\pm$0.007 & \textbf{1.104 $\pm$0.012} &  \textbf{0.099 $\pm$0.006} & 11.406 $\pm$0.015 \\
\midrule
 gossip/4                     &        -- &            -- & \textbf{0.119 $\pm$0.002} & \textbf{6.734 $\pm$0.027} & \textbf{0.119 $\pm$0.001} &             0.720 $\pm$0.002 &            0.118 $\pm$0.002 &         0.812 $\pm$0.014  \\
 gossip/10                    &        -- &            -- &         0.533 $\pm$0.003  &         6.786 $\pm$0.009  &         0.531 $\pm$0.003  &             1.561 $\pm$0.006 &   \textbf{0.524 $\pm$0.003} & \textbf{1.373 $\pm$0.004} \\
 gossip/20                    &        -- &            -- &         0.747 $\pm$0.003  &         7.064 $\pm$0.010  & \textbf{0.745 $\pm$0.003} &             3.565 $\pm$0.005 &            0.750 $\pm$0.003 & \textbf{2.888 $\pm$0.003} \\
\midrule
\end{tabular}
\caption{%
Empirical results of our benchmarks of the arrival, hybrid Bayesian
network, and gossip tasks. ``\host{} (\cont{})'' shows the evaluation of a baseline \cont{} program
with no boundary crossings into \disc{}, evaluations under the ``\host{}''
column performs interoperation. ``t/o'' indicates a timeout beyond 30 minutes,
and ``---'' indicates that the problem is not expressible in PSI because of an
unbounded loop. %
}\label{apx:eval:approx}

  \end{minipage}
\end{sideways}
\end{figure}

\begin{figure}[H]
\begin{sideways}
  \begin{minipage}{\textheight}
  \centering
  \begin{tabular}{c|cc|ccccccc}
      \toprule
  \multirow{2}{*}{\# Nodes} &               PSI &      \host{} (\disc{})   &   \multicolumn{2}{c}{Pyro}                                  &  \multicolumn{2}{c}{\host{} (\cont{})}         &  \multicolumn{2}{c}{\host{}}                     \\
                            &           Time(s) &         Time(s)          &          L1                 &      Time(s)                  &                     L1 &       Time(s)         &                  L1        &       Time(s)       \\
    \midrule
             9              & 546.748 $\pm$ 8.018 & \textbf{ 0.001 $\pm$0.000} &    0.080 $\pm$0.002           &         3.827 $\pm$0.008          &           0.079 $\pm$0.002 & \textbf{0.067 $\pm$0.001} &      \textbf{0.033 $\pm$0.001} &        0.098 $\pm$0.001 \\
             36             &               t/o & \textbf{ 0.089 $\pm$0.002} &    1.812 $\pm$0.009           &        14.952 $\pm$0.025          &           0.309 $\pm$0.004 & \textbf{0.277 $\pm$0.002} &      \textbf{0.055 $\pm$0.002} &        1.169 $\pm$0.004 \\
             81             &               t/o & \textbf{40.728 $\pm$0.276} &    7.814 $\pm$0.017           &        33.199 $\pm$0.049          &           0.680 $\pm$0.005 & \textbf{0.887 $\pm$0.002} &      \textbf{0.079 $\pm$0.002} &       81.300 $\pm$0.278 \\
  \end{tabular}
  \caption{%
    Exact and approximate results for models performing approximate inference}%
  \label{apx:eval:reliability}
  \end{minipage}
\end{sideways}
\end{figure}

\section{Comparison with Nested Inference} \label{app:sec:comparison-with-nested-inference}

It is interesting to contemplate the relationship between the nested inference
approach and \multippl{}.  A crisp comparison --- for instance, a formal
expressivity result establishing that it is not possible to represent our
multi-language interoperation using nested inference --- is difficult, due to
(1) the large variety of different approaches to nested inference, and (2) the
fact that such expressivity results are very hard even for very restricted
languages, let alone rich general-purpose probabilistic programming languages.
It would be very interesting to investigate the relative expressivity of
multi-language interoperation and nested inference, but
such an investigation is beyond the scope of this paper.

At the very least, what we can say is that \multippl{}'s low-level denotational
semantics, and hence also its inference strategy, is markedly different from the
standard measure-theoretic semantics of nested inference, such as
\citet{staton2017commutative}'s model of nested queries.
In \citet{staton2017commutative}, probabilistic computations denote measure-theoretic
kernels. The computation \(\mathsf{normalize}(t)\)
represents a nested query: its
takes in a probabilistic
computation \(t\) of type \(A\)
and produces a deterministic computation
that can yield one of three possible outcomes:
\begin{enumerate}
\item a tuple \((1,(e,d))\) consisting of a normalizing constant
  \(e\) and a distribution \(d\) over elements of type \(A\),
\item a tuple \((2,())\) signalling that the normalizing constant
  was zero,
\item or a tuple \((3,())\) signalling that the normalizing constant
  was infinity.
\end{enumerate}
Soundness of nested inference is then justified by the
following equational reasoning principle, reproduced here from
\citet{staton2017commutative}:
\begin{equation} \label{eqn:nested-query}
  \llbr{t} =
  \llbr{\begin{aligned}
    \mathsf{case~normalize}(t)~\mathsf{of}
      ~&(1,(e,d))\Rightarrow\mathsf{score}(e);\mathsf{sample}(d)\\
      |~&(2,())\Rightarrow\mathsf{score}(0);t\\
      |~&(3,())\Rightarrow t
  \end{aligned}}
\end{equation}
Disregarding the edge cases where the normalizing constant
is zero or infinity,
\cref{eqn:nested-query} says
that running a probabilistic computation \(t\)
is the same as (1) computing a representation of
the distribution of \(t\), either by exact inference
or by ``freezing a simulation'' and examining ``a histogram that has been built''
in the words of \citet{staton2017commutative},
and then (2) resampling from this distribution
and scoring by the normalizing constant.

There are a number of points that prevent this model of nested inference
from being applied directly to justify correctness of \multippl{}, which
may help to clarify the difference between nested inference and \multippl{}'s
inference interoperation:
\begin{itemize}
\item It is about a kernel-based model where programs take deterministic values
as input, but MultiPPL's \disc{} programs take random variables as input. This
is an important difference: because \disc{} programs take random variables as
input, a \disc{} program that simply makes use of a free variable in context
produces a random variable, not a fixed deterministic value. In contrast, such
a program always denotes a deterministic point mass under a kernel-based
semantics.

\item \multippl{}'s semantics of \disc{} programs is stateful, and this is
necessary to model how exact inference works in our implementation.
Contrastingly, the kernel-based model of \citet{staton2017commutative} is not
stateful in this way, and so would not have been sufficient for establishing
our main soundness theorem.

\item Because the model of \citet{staton2017commutative} is not stateful,
  it can't account for the stateful updates to the probability space that \multippl{}
  performs in order to ensure sample consistency.

\item Finally, \cref{eqn:nested-query} suggests to use importance reweighting via
  \(\mathsf{score}\) to ensure sound
  nesting of exact inference within an approximate-inference context, by properly
  taking the normalizing constant produced by the nested query into account.
  This is quite different from how \multippl{}'s \(\sexact\) boundary form handles
  the nesting of \disc{} subterms in \cont{} contexts --- as shown in \cref{fig:low-level-effectful},
  the low-level semantics of \(\sexact\) does not perform importance reweighting.
  Instead, importance reweighting occurs in the semantics of \disc{}-observe
  statements. Thus \cref{eqn:nested-query} does not explain \multippl{}'s importance reweighting scheme.
\end{itemize}
Together, these points show that multi-language inference
is distinct enough from nested inference that the standard measure-theoretic
model of nested queries from \citet{staton2017commutative}
cannot be used directly to justify key aspects of \multippl{}'s inference strategy,
such as the need for sample consistency and when importance reweighting is performed.
Though \citet{staton2017commutative} is just one approach to modelling nested inference,
it being a relatively well-established approach suggests that there are indeed
fundamental differences between nested inference and \multippl{}'s multilanguage inference.

\end{document}